\documentclass{article}

\usepackage{makeidx}
\usepackage{latexsym}
\usepackage{amsfonts}
\usepackage{amssymb}
\usepackage{amsmath}
\usepackage{amstext}
\usepackage{amsthm}
\usepackage{mathrsfs}
\usepackage{geometry}
\usepackage{graphicx}
\usepackage{mathabx}

\usepackage{color}

\newtheorem*{thmA}{Theorem A1}
\newtheorem*{thmB}{Theorem A2}
\newtheorem*{thmC}{Theorem B1}
\newtheorem*{thmD}{Theorem B2}
\newtheorem{theorem}{Theorem}

\newtheorem{lemma}[theorem]{Lemma}

\newtheorem{definition}[theorem]{Definition}

\newtheorem{proposition}[theorem]{Proposition}

\begin{document}

\title{A generalization to networks of Young's \\ characterization of the Borda rule
}

\author{\textbf{Daniela Bubboloni} \\
{\small {Dipartimento di Matematica e Informatica U.Dini} }\\
\vspace{-6mm}\\
{\small {Universit\`{a} degli Studi di Firenze} }\\
\vspace{-6mm}\\
{\small {viale Morgagni 67/a, 50134 Firenze, Italy}}\\
\vspace{-6mm}\\
{\small {e-mail: daniela.bubboloni@unifi.it}}\\
\vspace{-6mm}\\
{\small tel: +39 055 2759667}\\
\vspace{-6mm}\\
{\small https://orcid.org/0000-0002-1639-9525} \and \textbf{Michele Gori}
 \\
{\small {Dipartimento di Scienze per l'Economia e  l'Impresa} }\\
\vspace{-6mm}\\
{\small {Universit\`{a} degli Studi di Firenze} }\\
\vspace{-6mm}\\
{\small {via delle Pandette 9, 50127, Firenze, Italy}}\\
\vspace{-6mm}\\
{\small {e-mail: michele.gori@unifi.it}}\\
\vspace{-6mm}\\
{\small tel: +39 055 2759707}\\
\vspace{-6mm}\\
{\small https://orcid.org/0000-0003-3274-041X}}

\maketitle

\begin{abstract} 
\noindent We prove that, for any given set of networks satisfying suitable conditions, the net-oudegree network solution, the net-indegree network solution, and the total network solution are the unique network solutions on that set satisfying neutrality, consistency and cancellation. The generality of the result obtained allows to get an analogous result for social choice correspondences: for any given set of preference profiles satisfying suitable conditions, the net-oudegree social choice correspondence, the net-indegree social choice correspondence and the total social choice correspondence are the unique social choice correspondences on that set satisfying neutrality, consistency and cancellation. Using the notable fact that several well-known voting rules coincide with the restriction of net-oudegree social choice correspondence to appropriate sets of preference profiles, we are able to deduce a variety of new and known characterization theorems for the Borda rule, the Partial Borda rule, the Averaged Borda rule, the Approval Voting, the Plurality rule and the anti-Plurality rule, among which Young's characterization of the Borda rule and Fishburn's characterization of the Approval Voting.
\end{abstract}

\vspace{4mm}

\noindent \textbf{Keywords:} Network; Net-outdegree; Social choice theory; Voting; Borda rule; Approval voting.

\vspace{2mm}

\noindent \textbf{JEL classification:} D71.

\noindent \textbf{MSC classification:} 91B14, 05C20.

\section{Introduction}

A network on a set $A$ is a triple $N=(A,A^2_*,c)$, where $A$ is called the set of vertices of $N$, $A^2_*=\{(x,y)\in A^2:x\neq y\}$ is called the set of arcs of $N$  and $c$ is a function from $A^2_*$ to $\mathbb{Q}$ called the capacity associated with $N$. Networks can be used to represent a variety of situations. For instance, they can be used to mathematically represent competitions. Indeed, assume to have a set of teams which played a certain number of matches among each other and to know, for every team, the number of matches it won against any other team. Then we can represent that competition by a network $N=(A,A^2_*,c)$, where $A$ is the set of teams and, for every $(x, y) \in A^2_*$, $c(x, y)$ counts the number of matches in which $x$ beat $y$, assuming that ties are not allowed. As another example assume that $A$ is a set of alternatives. A network $N=(A,A^2_*,c)$ can represent a special type of information about the alternatives. Indeed, for every $(x, y) \in A^2_*$, $c(x, y)$ can be thought as a numerically evaluation of the strength of the statement ``$x$ is at least as good as $y$''. That strength might stem by considering several criteria to assess alternatives, assigning a suitable numerical weight to each criteria, and summing the weights of the criteria for which $x$ is at least as good as $y$. For instance, that strength can be computed by counting the number of voters/experts in a given panel declaring that $x$ is at least as good as $y$. 

Given a subset $\mathcal{D}$ of the set $\mathcal{N}(A)$ of all the networks on $A$, a network solution on $\mathcal{D}$ is a function from $\mathcal{D}$ to the set of nonempty subsets of $A$, a network method on $\mathcal{D}$ is instead a function from $\mathcal{D}$ to the set of binary relations of $A$. 
Typically, network solutions [methods] are defined by associating a score to each vertex of a network and considering the vertices maximizing the score [the complete and transitive relation on the vertices consistent with the scores]. 
If now $A$ is a set of teams and $\mathcal{D}$ is the set of networks that represent a complete description of wins and losses in a potential competition among the teams in $A$, we can see a network solution [method] on $\mathcal{D}$ as a procedure that allows to determine, for any conceivable competition involving teams in $A$, the set of winners [the team rankings]. If instead $A$ represents a set of alternatives and $\mathcal{D}$ is the set of networks that can be thought as a possible description of the numerical evaluations of the strength of all the statements of the type ``$x$ is at least as good as $y$'' where  $(x, y) \in A^2_*$, then we can see a network solution [method] as a procedure to determine, for any conceivable pairwise evaluation of the elements of $A$, the best alternatives [the social preference]. 

Many network solutions and network methods are considered and studied in the literature.
Langville and Meyer (2012) describe lots of ranking procedures used in different contexts (like social choice, voting, sport competitions, web search engines, psychology and statistics) that can be formalized as network methods.  Laslier (1997, Chapters 3 and 10) presents an extensive survey of network solutions defined on the set of tournaments\footnote{A tournament $R$ on $A$ is an asymmetric relation on $A$ such that, for every $(x,y)\in A^2_*$, $(x,y)\in R$ or $(y,x)\in R$. A relation $R$ on a set $A$ can be identified with a 0-1-network $N=(A,A^2_*,c)$ where, for every $(x,y)\in A^2_*$, 
$c(x,y)=1$ if $(x,y)\in R$ and $c(x,y)=0$ if $(x,y)\notin R$.} and on the set of balanced networks\footnote{A network $N=(A,A^2_*,c)$ is balanced if there exists a constant $k\in\mathbb{Q}$ such that, for every $(x,y)\in A^2_*$, $c(x,y)+c(y,x)=k$.}.

One of the simplest network solution [method], considered by Brans et al. (1986) and Bouyssou (1992), is obtained by associating with each vertex of a network a score given by its net-outdegree, namely the difference between its outdegree and its indegree, and then selecting the vertices maximizing the score value [ranking the vertices according to the scores]. More precisely, given $N=(A,A^2_*,c)\in \mathcal{N}(A)$, the score of $x\in A$ is defined as
\[
\delta^N(x)=\sum_{y\in A\setminus\{x\}}c(x,y)-\sum_{y\in A\setminus\{x\}}c(y,x).
\]
Such a network solution [method], which is defined on $\mathcal{N}(A)$, is called here net-outdegree network solution [method].\footnote{Brans et al. (1986) and Bouyssou (1992) refer to such network method as net-flow network method, but we think that the name net-outdegree better adheres to the rational behind the definition.} It is worth noting that the net-outdegree network solution [method] coincides with the classic Copeland solution [method] when applied to the $0$-$1$ networks corresponding to complete relations and then, in particular, corresponding to tournaments; with the flow network solution [method] proposed by Gvozdik (1987) and deepened by Belkin and Gvozdik (1989) and Bubboloni and Gori (2018) when applied to balanced networks with capacity whose values are nonnegative integers; with the outflow network solution [method] by van den Brink and Gilles (2009) when applied to balanced networks. 
Another simple network solution [method] is obtained by associating with each vertex of a network a score given by its net-indegree, namely the difference between its indegree and its outdegree, and then selecting the vertices maximizing the score value [ranking the vertices according to the scores]. Such a network solution [method], which is defined on $\mathcal{N}(A)$, is called here net-indegree network solution [method]. In the paper we are mainly interested in the net-outdegree network solution, denoted by $\mathscr{O}$, and the net-indegree network solution, denoted by $\mathscr{I}$. Of course, we have that, for every $N\in \mathcal{N}(A)$,
\[
\mathscr{O}(N)=\underset{x\in A}{\mathrm{argmax}}\,\delta^N(x)\quad\mbox{ and }\quad \mathscr{I}(N)=\underset{x\in A}{\mathrm{argmax}}\,\left(-\delta^N(x)\right)=\underset{x\in A}{\mathrm{argmin}}\,\delta^N(x).
\]
In order to better explain our findings, it is  also important to consider the total network solution, denoted by $\mathscr{T}$, which associates the whole set of alternatives with any element of $\mathcal{N}(A)$.

In order to evaluate the quality of a network solution is certainly important to understand whether it satisfies desirable properties. 
Among the possible properties one can consider there are are neutrality, consistency and cancellation. A network solution $\mathscr{F}$ satisfies neutrality if the names of the elements of $A$ are immaterial to the selection process; consistency if, every time the outcomes of $\mathscr{F}$ at two networks have nonempty intersection, then that intersection must be the outcome of $\mathscr{F}$ at the sum of the two networks; cancellation if $\mathscr{F}$ selects all the vertices when applied to reversal symmetric networks, namely those networks whose capacity of every arc $(x,y)$ equals the capacity of the arc $(y,x)$.
Interpreting $A$ as a set of teams, networks\footnote{We are referring in this example to networks with capacity having values on the set of nonnegative integers.} as a complete description of the wins and losses in a competition (where ties are not allowed), and a network solution as a procedure for selecting the winners of each possible competition, the above described properties sound very natural. Indeed, neutrality implies that no team or group of teams has an exogenous advantage on the others; consistency means that if there are teams that are the winners of a competition $N$ and also of a competition $N'$, then those teams are exactly the winners of the competition that considers the outcomes of all the matches played in $N$ and $N'$; cancellation means that if every team won and lost the same number of times against any other team, then all the teams are the winners.

The first main result of this paper is the following theorem.\footnote{Theorem A1 corresponds to Theorem \ref{characterization-net}. In the paper, given $A$, $B$ and $C$ sets with $B\subseteq A$ and $f:A\to C$, we denote with $f_{|B}$ the restriction of $f$ to $B$, that is, $f_{|B}:B\to C$ is defined, for every $b\in B$, by $f_{|B}(b)=f(b)$.}

\begin{thmA}
Let $\mathcal{D}$ be a regular subset of $\mathcal{N}(A)$. Then $\mathscr{O}_{|\mathcal{D}}$, $\mathscr{I}_{|\mathcal{D}}$ and $\mathscr{T}_{|\mathcal{D}}$ are the unique network solutions on $\mathcal{D}$ satisfying neutrality, consistency and cancellation.
\end{thmA}

A set of networks is regular if it satisfies the technical but mild conditions described in Definition \ref{regular-net}. A variety of interesting sets of networks particularly useful for applications are regular: the set of networks [balanced networks], the set of networks [balanced networks] whose capacity has nonnegative values, and the set of networks [balanced networks] whose capacity has values in the set of nonnegative integers (Theorem \ref{reg-scc}). Note that, in particular, Theorem A1 immediately implies the following theorem that provides a characterization of the restriction of the net-outdegree network solution to any regular set of networks.\footnote{Theorem A2 corresponds to Theorem \ref{crucial-net}.}

\begin{thmB}
Let $\mathcal{D}$ be a regular subset of $\mathcal{N}(A)$. Then $\mathscr{O}_{|\mathcal{D}}$ is the unique network solution on $\mathcal{D}$ satisfying neutrality, consistency, cancellation and $\mathscr{O}$-coherence.
\end{thmB}

A network solution $\mathscr{F}$ on $\mathcal{D}$ satisfies $\mathscr{O}$-coherence if there exists a network $N$ in $\mathcal{D}$ such that $\mathscr{F}(N)=\mathscr{O}(N)$, $\mathscr{F}(N)\neq \mathscr{I}(N)$ and $\mathscr{F}(N)\neq A$. In other words, for at least an element of $\mathcal{D}$, $\mathscr{F}$ coincides with $\mathscr{O}$ and is different from $\mathscr{I}$  and $\mathscr{T}$.

Some characterization results for suitable restrictions of the net-outdegree network solution [method] are available in the literature. Indeed, the net-outdegree network method restricted to tournaments is characterized by Rubenstein (1980); the net-oudegree network method [solution] restricted to complete relations is characterized by Henriet (1985); the net-oudegree network method restricted to the set of networks whose capacity has values in $[0,1]$ is characterized by Bouyssou (1992); the net-outdegree network method on the set of all the networks whose capacity has values in the set of nonnegative integers is characterized by Barber\`a and Bossert (2023).
It is worth noticing that, differently from our result, all those characterizations involve a monotonicity property, while none of them considers consistency. Thus, Theorem A2 is qualitatively not comparable with those results.

The proof of Theorem A1 is mainly inspired by the proof of the characterization theorem of the Borda rule by Young (1974). Over the years, some scholars tried to simplify Young's arguments by avoiding linear algebra and graph theory and deduced weaker results following a different approach (see, for instance, Hansson and Sahlquist, 1976; see also the comments about the proof of Young's result on the Borda rule in Al\'os-Ferrer, 2006). Differently from those authors and in line with Young's approach, in this paper we exploit linear algebra, graph theory, and also network theory to generalize Young's findings to the framework of network solutions.\footnote{The idea of associating a certain weighted graph with a preference profile, as well as the use of linear algebra, appears also in Zwicker (1991), Duddy et al. (2016) and Brandl and Peters (2019).}   Ultimately, such a research approach turned out fruitful, since it not only led to Theorems A1 and A2 but also allowed, as a byproduct, to determine a common root to a variety of characterization theorems in social choice theory, as described below.

Consider a set $A$ of alternatives and a countable set $V$ of potential voters whose preferences are modelled as binary relations on $A$.
A preference profile $p$ is a function from a nonempty and finite subset of $V$ to the set $\mathbf{R}(A)$ of the binary relations on $A$. Considering, for every pair of distinct alternatives $x$ and $y$, the number $c_p(x,y)$ of voters in the domain of $p$ who think that $x$ is at least as good as $y$, we construct the network $N(p)=(A, A^2_*,c_p)$ associated with $p$. 

Given a subset $\mathbf{D}$ of the set $\mathbf{R}(A)^*$ of all the possible preference profiles, a social choice correspondence ({\sc scc}) on $\mathbf{D}$ is a function from $\mathbf{D}$ to the set of nonempty subsets of $A$. A {\sc scc} on $\mathbf{D}$ is then a procedure that allows to select some alternatives in $A$ on the basis of the preferences of some voters in $V$. The set $\mathbf{D}$ is usually determined by imposing that voters' preferences satisfy some properties like, for instance, being orders, linear orders or dichotomous orders on $A$.  Consider now the three {\sc scc} on $\mathbf{R}(A)^*$, called net-outdegree {\sc scc}, net-indegree {\sc scc} and total {\sc scc} and respectively denoted by $O$, $I$ and $T$, defined, for every preference profile $p$, by $O(p)=\mathscr{O}(N(p))$, $I(p)=\mathscr{I}(N(p))$ and $T(p)=A$. 

After finding conditions that guarantee that a social choice correspondence on a set $\mathbf{D}$ of preference profiles induces a network solution on the set of the networks of type $N(p)$ for $p\in \mathbf{D}$ (Proposition \ref{passaggio-reti}), we can exploit Theorem A1 to prove the second main result of the paper.\footnote{Theorem B1 corresponds to Theorem \ref{characterization-profiles}.}

\begin{thmC}
Let $\mathbf{D}$ be a regular subset of $\mathbf{R}(A)^*$. Then $O_{|\mathbf{D}}$, $I_{|\mathbf{D}}$ and $T_{|\mathbf{D}}$ are the unique social choice correspondences on $\mathbf{D}$ satisfying neutrality, consistency  and cancellation.
\end{thmC} 

A set of preference profiles is regular if it satisfies the technical but mild conditions described in Definition \ref{regular}.
Most of the classic domains of social choice correspondences are actually regular (Theorem \ref{reg-scc}).
A {\sc scc} $F$ satisfies neutrality if alternative names are immaterial; consistency\footnote{Some authors use the term ``reinforcement'' instead of ``consistency''  (see, for instance, Moulin, 1988; Young, 1988; Myerson, 1995).} if, given two preference profiles $p_1$ and $p_2$ referring to disjoint groups of voters and for which $F$ selects two subsets of $A$ with nonempty intersection, then $F$ selects that intersection when computed at the preference profile obtained by combining $p_1$ and $p_2$; cancellation if $F$ selects the whole set $A$ of alternatives on those preference profiles for which, for every pair of distinct alternatives $x$ and $y$, the number of voters who think that $x$ is at least as good as $y$ is the same as the number of voters who think that $y$ is at least as good as $x$. 

As an immediate consequence of Theorem B1 we can prove the following characterization theorem.\footnote{Theorem B2 corresponds to Theorem \ref{crucial}.}

\begin{thmD}
Let $\mathbf{D}$ be a regular subset of $\mathbf{R}(A)^*$. Then $O_{|\mathbf{D}}$ is the unique social choice correspondence on $\mathbf{D}$ satisfying neutrality, consistency, cancellation and $O$-coherence.
\end{thmD} 

A {\sc scc} $F$ on $\mathbf{D}$ satisfies $O$-coherence if  there exists a preference profile $p$ in $\mathbf{D}$ such that $F(p)=O(p)$, $F(p)\neq I(p)$ and $F(p)\neq A$. $O$-coherence is a very weak property and it is implied by the classic property of faithfulness on those domains of preference profiles where individual preferences are assumed to be orders (see Section \ref{faoc}). Theorem B2 has important consequences, once some further properties of the net-outdegree {\sc scc} are explained.

First of all, note that, for every $p\in \mathbf{R}(A)^*$, $O(p)$ corresponds to the set of  of elements of $A$ maximizing the integer valued function defined, for every $x\in A$, by
\[
\sum_{y\in A\setminus\{x\}}c_p(x,y)-\sum_{y\in A\setminus\{x\}}c_p(y,x).
\]
That function appears in Young (1974) as an alternative way to define the Borda scores of alternatives when voters' preferences are linear orders on $A$. In the same paper, the author also suggests the possibility to use that function to define the Borda scores of the alternatives when voter's preferences are orders or partial orders and of course that idea can be naturally extended to get a {\sc scc} on $\mathbf{R}(A)^*$. Indeed, $O$ turns out to be an unifying object since it coincides with classic and well-known {\sc scc}s when its domain is restricted to special subsets of $\mathbf{R}(A)^*$. To be more precise, $O$ coincides 
with the Borda rule when voters' preferences are linear orders; 
with the generalized Borda rule 
when voters' preferences are orders; 
with the Partial Borda rule 
when voters' preferences are partial orders;
with the Averaged Borda rule 
when voters' preferences are top truncated orders;
with the Approval Voting when voters' preferences are  dichotomous orders;
with the Plurality rule when voters' preferences are  dichotomous orders with a single alternative better than all the others; 
with the anti-Plurality rule when voters' preferences are  dichotomous orders with a single alternative worse than all the others (see Section \ref{Ouguale}).

Properly qualifying the set $\mathbf{D}$ of preference profiles, Theorem B2 allows to obtain many characterization theorems for classic {\sc scc}s (Theorem \ref{lista}). Those theorems imply in turn many known characterization theorems. In particular, we deduce from our results the characterizations of the Borda rule by Young (1974, Theorem 1), of the Partial Borda rule by Cullinan (2014, Theorem 2), of the Averaged Borda rule by Terzopoulou and Endriss (2021, Theorem 6), of the Approval Voting by Fishburn (1979, Theorem 4), of the Plurality rule by Sekiguchi (2012, Theorem 1), of the anti-Plurality rule by Kurihara (2018, Theorem 1). Those characterizations appear in the paper as Theorems \ref{young-theorem},  \ref{culli-theorem},  \ref{Fishburn}, \ref{terzo-endriss}, \ref{Seki} and \ref{kuri}, respectively. We also prove in Theorem \ref{borda-young} that neutrality, consistency, cancellation and faithfulness characterize the generalized Borda rule, a result stated without proof by Young (1974).

\section{Networks and their properties}\label{sec-network}

\subsection{Preliminary notation}

We denote by $\mathbb{N}$ the set of positive integers and we set $\mathbb{N}_0\coloneq \mathbb{N}\cup\{0\}$ and $\mathbb{Q}_+\coloneq \{q\in\mathbb{Q}:q\ge 0\}$. For $k\in \mathbb{N}$ we set $[k]\coloneq \{x\in \mathbb{N}: x\leq k\}$ and $[k]_0\coloneq \{x\in \mathbb{N}_0: x\leq k\}$.   
For $A, B, C$ sets and functions $f:A\rightarrow B$, $g:B\rightarrow C$ we denote by $gf:A\rightarrow B$ the composition of $f$ and $g$, that is the
function defined, for every $a\in A$, by $gf(a)=g(f(a)).$

Given $X$  a finite set, we denote by $|X|$ the size of $X$,  by $X^2_d$ the set $\{(x,y)\in X^2:x=y\}$, by $X^2_*$ the set $X^2\setminus X^2_d$, by $P(X)$ the set of subsets of $X$, by $P_*(X)$ the set of nonempty subsets of $X$, by $P_k(X)$ the set of subsets of $X$ of size $k$ with $k\in \mathbb{N}$.

We denote by $\mathrm{Sym}(X)$ the set of the permutations of $X$. 
Recall that $\mathrm{Sym}(X)$ is a group under composition of functions. If $|X|=n\geq 2$, given distinct $x_1,\dots,x_k\in X$ with $2\leq k\leq n$, we denote by $(x_1\cdots x_k)$ the permutation in $\mathrm{Sym}(X)$ defined by $(x_1\cdots x_k)(x_i)=x_{i+1}$ for all $i\in [k-1]$, $(x_1\cdots x_k)(x_k)=x_{1}$ and $(x_1\cdots x_k)(x)=x$ for all $x\in X\setminus\{x_1,\dots, x_k\}$. The identity function on $X$ is denoted by $id$. Thus, $id\in\mathrm{Sym}(X)$ and $id$ is defined, 
for every $x\in X$, by $id(x)=x$.

Throughout the paper vector spaces are meant to be defined on the field $\mathbb{Q}$. Let $V$ be a vector space. If $W\subseteq V$ is a subspace of $V$, we write $W\le V$. Consider now $\mathbb{X}\subseteq \mathbb{Q}$ and $Z\subseteq V$. We denote by $\mathbb{X}Z$ the set of the finite linear combinations with coefficients in $\mathbb{X}$ of elements in $Z$.  As usual, we include among those linear combinations the empty linear combination and assume that it equals the null vector $0$ of $V$. In particular, $\mathbb{X}\varnothing=\varnothing Z=\{0\}$. Note also that $\mathbb{Q}Z$ is the subspace of $V$ generated by $Z$.

\subsection{Networks}

Let $A$ be a finite set with $|A|=m\ge 2$. 
A  network on $A$ is a triple $N=(A,A^2_*,c)$, where $c$ is a function from $A^2_*$ to $\mathbb{Q}$.
We say that $A$ is the set of vertices of $N$, $A^2_*$ is the set of  arcs of $N$ and $c$ is the capacity associated with $N$.
The set of networks on $A$ is denoted by $\mathcal{N}(A)$.

Let $N=(A,A^2_*,c)\in\mathcal{N}(A)$. The reversal of $N$ is the network $N^r=(A,A^2_*,c^r)\in\mathcal{N}(A)$ where, 
for every $(x,y)\in A^2_*$, $c^r(x,y)=c(y,x)$. If $N=N^r$ we say that $N$ is reversal symmetric.
Given $\psi\in \mathrm{Sym}(A)$, we set $N^{\psi}=(A,A^2_*,c^{\psi}) \in\mathcal{N}(A)$ where $c^{\psi}$ is defined, for every $(x,y)\in A^2_*$, by
$c^{\psi}(x,y)=c(\psi^{-1}(x), \psi^{-1}(y))$. Given $k\in \mathbb{Q}$, $N$ is  
$k$-constant if, for every $(x,y)\in A^2_*$, $c(x,y)=k$;
$k$-balanced if, for every $(x,y)\in A^2_*$, $c(x,y)+c(y,x)=k$.
The unique $k$-constant network is denoted by $N(k)$. 
$N$ is balanced if $N$ is $k$-balanced for some $k\in\mathbb{Q}$; constant if $N$ is $k$-constant for some $k\in\mathbb{Q}$.

Given $N=(A,A^2_*,c^N)\in\mathcal{N}(A)$, $M=(A,A^2_*,c^M)\in\mathcal{N}(A)$ and $k\in\mathbb{Q}$, we set $k N=(A,A^2_*,kc^N)\in \mathcal{N}(A)$ 
and $N+M=(A,A^2_*,c^N+c^M)$. That defines a vector space structure on $\mathcal{N}(A)$ having as zero vector $N(0).$
Moreover,  for every $N,M\in \mathcal{N}(A)$, $h,k\in \mathbb{Q}$ and $\psi,\sigma \in \mathrm{Sym}(A)$, the following properties hold true\footnote{In particular,  $\mathcal{N}(A)$ is a $\mathbb{Q} \mathrm{Sym}(A)$-module.}
\begin{equation*}\label{somma}
(hN+kM)^\psi=h N^\psi+ kM^\psi,\quad (N^\psi)^\sigma=N^{\sigma\psi},
\end{equation*}
\begin{equation*}\label{sigma}
(hN)^r=hN^r,\quad (N+M)^r=N^r+M^r,\quad (N^\psi)^r=(N^r)^\psi.
\end{equation*}

For every  $(x,y)\in A^2_*$, we denote by $N_{xy}$ the network on $A$
whose capacity is $1$ on the arc $(x,y)$ and $0$ elsewhere. Note that, for every $N=(A,A^2_*,c^N)\in\mathcal{N}(A)$, we have
\[
N=\sum_{(x,y)\in A^2_*}c^N(x,y)N_{xy}.
\]
As a consequence,  the set $\{N_{xy}: (x,y)\in A^2_*\}$ is a basis for the vector space $\mathcal{N}(A)$ and thus $\mathrm{dim}\, \mathcal{N}(A)=m(m-1)$. For every $x\in A$, we call the network
\[
N_{x}\coloneq \sum_{y\in A\setminus\{x\}}N_{xy}
\]
the outstar network on the vertex $x$. It is easy to prove that the networks $N_x$ for $x\in A$ are linearly independent in $\mathcal{N}(A)$.
Observe that, for every $(x,y)\in A^2_*$ and $\psi\in \mathrm{Sym}(A)$, we have $N_{xy}^\psi=N_{\psi(x)\psi(y)}$ and $N_{xy}^r=N_{yx}$. As a consequence,
for every $x\in A$ and $\psi\in \mathrm{Sym}(A)$, we have $N_{x}^\psi=N_{\psi(x)}$.

For every $B\subseteq A$, we set 
\begin{equation}\label{Bcompleto}
K_B\coloneq \sum_{(x,y)\in B_*^2}( N_{xy}+N_{yx})
\end{equation}
and call it the complete network on $B$. Note that $K_B$ is reversal symmetric. 

We finally introduce the concept of cycle network. Let $\gamma=x_1\,x_2\dots x_k\,x_1$  be a $k$-cycle of the complete directed graph $(A, A^2_*)$, for some $2\leq k\leq m$.\footnote{That means that $\gamma$ is the subgraph of $(A, A^2_*)$ having as vertex set the set of $k$ distinct  vertices  $\{x_1,x_2,\dots, x_k\}\subseteq A$ and arcs $(x_i,x_{i+1})$, for $1\leq i\leq k-1$, and $(x_k,x_1)$. }  The 
$k$-cycle network  associated with $\gamma$ is defined by
\[
C_{x_1\,x_2\dots x_k\,x_1}\coloneq \left(\sum_{i=1}^{k-1}N_{x_i\,x_{i+1}}\right)+N_{x_k\,x_1}.
\]
A network $C\in \mathcal{N}(A)$ is called a cycle network if it is a $k$-cycle network for some $2\leq k\leq m$.

\subsection{The net-outdegree}\label{net-sec}

Let $N=(A,A^2_*,c)\in\mathcal{N}(A)$. The net-outdegree associated with $N$ is the function $\delta^N:A\rightarrow\mathbb{Q}$ defined, for every $x\in A$, by
\begin{equation}\label{deltadef}
\delta^N(x)\coloneq \sum_{y\in A\setminus\{x\}}c(x,y)-\sum_{y\in A\setminus\{x\}}c(y,x).
\end{equation}
The number $\delta^N(x)$ is called net-outdegree of $x$ in $N$.
$N$ is called pseudo-symmetric if, for every $x\in A$, $\delta^N(x)=0$.

Let $\mathbb{Q}^A$ be the set of functions from $A$ to $\mathbb{Q}$. Of course, $\mathbb{Q}^A$ is a vector space of dimension $m$. We denote by $0$ its zero, that is, the function from $A$ to $\mathbb{Q}$ that associates with every $x\in A$ the value $0$.
Consider now the function  
\begin{equation*}\label{delta-lin}
\delta:\mathcal{N}(A)\rightarrow \mathbb{Q}^A,\qquad \delta(N)\coloneq \delta^N.
\end{equation*}
Thus, for every $N\in \mathcal{N}(A)$ and $x\in A$, we have $\delta(N)(x)=\delta^N(x)$. We call $\delta$ the net-outdegree. Recalling the definition \eqref{deltadef} and the definition of sum of networks,
it is immediately shown that, for every $N,M\in \mathcal{N}(A)$, $h,k\in \mathbb{Q}$, we have 
\[
\delta^{hN+kM}=h\delta^{N}+k\delta^{M}.
\]
In other words $\delta$ is a linear function between vector spaces. Thus, $\mathrm{Ker}(\delta)$ is the set of pseudo-symmetric networks. Moreover, recalling the definition of $N^{\psi}$, it is easily checked that, for every $N\in \mathcal{N}(A)$ and $\psi\in \mathrm{Sym}(A),$ we have
\begin{equation}\label{delta-psi}
\delta^{N^{\psi}}=\delta^{N}\psi^{-1}.
\end{equation}
Some basic computations involving $\delta$ are easy. For instance, for every $x,y\in A$ with $x\neq y$, we have that
\begin{equation*}\label{delta-base}
\delta(N_{xy})=\chi_{\{x\}}-\chi_{\{y\}}
\end{equation*}
and hence
\begin{equation}\label{delta-spruzzo}
\delta(N_{x})=(m-1) \chi_{\{x\}}-\chi_{A\setminus \{x\}}
\end{equation}
where, for every  $S\subseteq A$, $\chi_{S}\in \mathbb{Q}^A$ denotes the characteristic function of $S$.

We now study the image of $\delta$ and the dimension of the kernel of $\delta$.
In particular, by Proposition \ref{im-delta-1}, we understand that the image of $\delta$ is generated by the image through $\delta$ of the outstar networks.

\begin{proposition}\label{im-delta-1} 
$\delta(\mathcal{N}(A))=\mathbb{Q}\{\delta(N_x):x\in A\}$ and $\mathrm{dim\, } \delta(\mathcal{N}(A))= m-1$.
\end{proposition}
\begin{proof} First of all note that, for every $N\in \mathcal{N}(A)$, we have
\[
\sum_{x\in A}\delta^N(x)=\sum_{x\in A}\left(\sum_{y\in A\setminus\{x\}}c(x,y)-\sum_{y\in A\setminus\{x\}}c(y,x)\right)
=\sum_{(x,y)\in A_*^2}c(x,y)-\sum_{(x,y)\in A_*^2}c(y,x)= 0.
\]
Thus
\[
\delta\left(\mathcal{N}(A)\right)\subseteq \left\{f\in \mathbb{Q}^A: \sum_{x\in A}f(x)=0\right\}\subsetneq\mathbb{Q}^A.
\]
As a consequence, $\delta\left(\mathcal{N}(A)\right)$ is a proper subspace of $\mathbb{Q}^A$ and hence
\begin{equation}\label{dim-delta-ineq}
\mathrm{dim}\, \delta\left(\mathcal{N}(A)\right)\leq m-1.
\end{equation}

Consider now the vectors $\delta(N_x)\in \mathbb{Q}^A$ for $x\in A$. Since 
\[
\sum_{x\in A}\delta(N_x)=\delta\left(\sum_{x\in A}N_x\right)=\delta(N(1))=0,
\]
those vectors are dependent. Choose now $x_1,\dots, x_{m-1}$ distinct elements of $A=\{x_1,\dots, x_{m-1},x_m\}$. We show that the vectors $\delta(N_{x_i})$ for $i\in[m-1]$ are independent.
Assume that $\sum_{i=1}^{m-1}q_i\delta(N_{x_i})=0$ for certain $q_i\in \mathbb{Q}$, $i\in[m-1]$. Let $f$ and $g$ be the functions in $\mathbb{Q}^A$ respectively defined by
\[
f=(m-1)\sum_{i=1}^{m-1}q_i\chi_{\{x_i\}}\quad\mbox{ and }\quad g=\sum_{i=1}^{m-1}q_i\chi_{A\setminus\{x_i\}}.
\]
Using \eqref{delta-spruzzo}, we have that $f=g$. Since $f(x_m)=g(x_m)$, we get $\sum_{i=1}^{m-1}q_i=0$. Thus, for every $j\in[m-1]$, from $f(x_j)=g(x_j)$ we get $(m-1)q_j= \sum_{i=1, i\neq j}^{m-1}q_i=-q_j$ and that implies $q_j=0$.

Thus, $\mathrm{dim\,} \delta\left(\mathcal{N}(A)\right)\geq m-1$. 
By \eqref{dim-delta-ineq}, we finally deduce that  
\[
\mathrm{dim\, } \delta\left(\mathcal{N}(A)\right)= m-1=
\mathrm{dim\,}\mathbb{Q}\{\delta(N_x):x\in A\},
\] 
and hence also $\delta\left(\mathcal{N}(A)\right)=\mathbb{Q}\{\delta(N_x):x\in A\}$.
\end{proof}

\begin{proposition}\label{ker-delta} 
$\mathrm{dim\, Ker}(\delta)=(m-1)^2.$
\end{proposition}
\begin{proof} 
 From Proposition \ref{im-delta-1}, we deduce 
 $$m(m-1)=\mathrm{dim}\, \mathcal{N}(A)=\mathrm{dim}\, \delta(\mathcal{N}(A))+\mathrm{dim\, Ker} (\delta)=m-1 +\mathrm{dim\, Ker} (\delta)$$ and therefore $\mathrm{dim\, Ker} (\delta)=(m-1)^2$.
\end{proof}

\subsection{Special sets of networks}

From now on
\begin{itemize}
\item $\mathcal{R}(A)$ denotes the set of reversal symmetric networks; 
\item $\mathcal{C}(A)$ denotes the set of constant networks; 
\item $\mathcal{B}_k(A)$ denotes the set of $k$-balanced networks, where $k\in \mathbb{Q}$; 
\item $\mathcal{B}(A)$ denotes the set of balanced networks; 
\item $\mathcal{PS}(A)$ denotes the set of pseudo-symmetric networks; 
\item $\mathcal{O}(A)$ denotes the set $\mathbb{Q}\{N_x:x\in A\}$;
\item for every $\mathbb{X}\subseteq \mathbb{Q}$, $\mathcal{N}_\mathbb{X}(A)$ denotes the set of networks whose capacity has values in $\mathbb{X}$.
\end{itemize}
Note that $\mathcal{PS}(A)=\mathrm{Ker}(\delta)$. We propose  now some propositions stating properties of the aforementioned sets of networks. Those propositions will turn out crucial in proving the main results of the paper. We remark that Proposition \ref{ps-3} is inspired by a reasoning proposed in Young (1974). The proof of Proposition \ref{generale-reti} is in the appendix.

\begin{proposition} \label{generale-reti} The following facts hold true.
\begin{itemize}
\item[$(i)$]
$\mathcal{C}(A)$, $\mathcal{R}(A)$, $\mathcal{B}_0(A)$, $\mathcal{B}(A)$, $\mathcal{PS}(A)$ and $\mathcal{O}(A)$ are subspaces\footnote{To be more precise, they are $\mathbb{Q} \mathrm{Sym}(A)$-submodules of the $\mathbb{Q} \mathrm{Sym}(A)$-module $\mathcal{N}(A)$. Note also that $\mathcal{B}_k(A)$ is a subspace of $\mathcal{N}(A)$ if and only if $k=0$.} of $\mathcal{N}(A)$. 
\item[$(ii)$]
\begin{itemize}
\item[$(a)$] $\mathrm{dim}\,\mathcal{C}(A)=1$;
\item[$(b)$] $\mathrm{dim}\,\mathcal{R}(A)=\frac{m(m-1)}{2}$;
\item[$(c)$] $\mathrm{dim}\,\mathcal{B}_0(A)=\frac{m(m-1)}{2}$;
\item[$(d)$] $\mathrm{dim}\,\mathcal{B}(A)=\frac{m(m-1)}{2}+1$;
\item[$(e)$] $\mathrm{dim\,}\mathcal{PS}(A)=(m-1)^2$;
\item[$(f)$] $\mathrm{dim}\,\mathcal{O}(A)=m$;
\item[$(g)$] $\mathrm{dim}\, \left(\mathcal{PS}(A)\cap \mathcal{B}(A)\right)=\frac{(m-1)(m-2)}{2}+1$.
\end{itemize}
\item[$(iii)$]
\begin{itemize}
\item[$(a)$] $\mathcal{R}(A)\subseteq \mathcal{PS}(A)$;
\item[$(b)$] $\mathcal{C}(A)= \mathcal{R}(A)\cap \mathcal{B}(A)$;
\item[$(c)$] $\mathcal{B}(A)= \mathcal{C}(A)\oplus \mathcal{B}_0(A)$;
\item[$(d)$] $\mathcal{N}(A)=\mathcal{R}(A)\oplus\mathcal{B}_0(A)=\mathcal{R}(A)+\mathcal{B}(A)$;
\item[$(e)$] $\mathcal{C}(A)=\mathcal{O}(A)\cap\mathcal{PS}(A)$;
\item[$(f)$] $\mathcal{N}(A)=\mathcal{PS}(A)+\mathcal{O}(A)$;
\item[$(g)$] $\mathcal{R}(A)=\left(\mathcal{O}(A)+\mathcal{R}(A)\right)\cap \mathcal{PS}(A)$.
\end{itemize}
\item[$(iv)$] For every $\mathcal{X}\in \{\mathcal{C}(A), \mathcal{R}(A), \mathcal{B}(A), \mathcal{PS}(A)\}$, $\mathcal{X}^r=\mathcal{X}$.

\end{itemize}
\end{proposition}

\begin{proposition}\label{costanti-sim} 
Let $N\in \mathcal{N}(A)$. Then $N\in \mathcal{C}(A)$ if and only if $N^{\psi}=N$ for all $\psi\in \mathrm{Sym}(A).$
\end{proposition}

\begin{proof} Assume that $N\in \mathcal{C}(A)$ is $k$-constant and let $\psi\in \mathrm{Sym}(A)$. Then, for every $(x,y)\in A^2_*$, we have $c^{\psi}(x,y)=c(\psi^{-1}(x), \psi^{-1}(y))=k=c(x,y)$, which means $N^{\psi}=N$. Conversely assume that $N^{\psi}=N$ for all $\psi\in \mathrm{Sym}(A)$. Then we have $c=c^{\psi}$ for  all $\psi\in \mathrm{Sym}(A)$. Let $(x,y)\in A^2_*$ and $(u,v)\in A^2_*$. Since $x\neq y$ and $u\neq v$, there exists $\psi\in \mathrm{Sym}(A)$ such that $\psi(u)=x$ and $\psi(v)=y$ such that $c(x,y)=c^{\psi}(x,y)=c(\psi^{-1}(x), \psi^{-1}(y))=c(u,v)$, and thus the function $c$ is constant.
\end{proof}

\begin{proposition}\label{ps-cycles}
Let $N\in\mathcal{PS}(A)\cap \mathcal{N}_{\mathbb{N}_0}(A)$. Then $N=N(0)$ or
there exist $s\in \mathbb{N}$ and a sequence $(C_i)_{i=1}^s$ of cycle networks such that $N=\sum_{i=1}^s C_i$.\footnote{Proposition \ref{ps-cycles} can be deduced by classic results on graph homology. However, since we could not find a statement perfectly in line with ours, we preferred to produce an original proof.}
\end{proposition}

\begin{proof}
Given $N=(A,A^2_*,c)\in \mathcal{N}_{\mathbb{N}_0}(A)$, let us define
\[
G^N\coloneq \sum_{(x,y)\in A^2_*} c(x,y)\in\mathbb{N}_0.
\]
For every $r\in \mathbb{N}_0$, let 
$\mathcal{PS}^r(A)\coloneq \{N\in \mathcal{PS}(A)\cap \mathcal{N}_{\mathbb{N}_0}(A): G^N= r\}$.
We complete the proof showing that, for every $r\in \mathbb{N}_0$, the following statement holds true:
\begin{quote}
Stat$(r)$: if $N\in \mathcal{PS}^r(A)$, then $N=N(0)$ or
there exist $s\in \mathbb{N}$ and a sequence $(C_i)_{i=1}^s$ of cycle networks such that $N=\sum_{i=1}^s C_i$.
\end{quote}
We work by induction on $r$.
If $r=0$, $\mathcal{PS}^0(A)=\{N(0)\}$ and Stat$(r)$ is true.
Let $r\ge 0$ and assume that Stat$(t)$ is true for every $t\le r$. We have to prove that  Stat$(r+1)$ is true.
Let $N=(A,A^2_*,c)\in  \mathcal{PS}^{r+1}(A)$. For every $x\in A$, define
\[
I(x)\coloneq \{y\in A\setminus\{x\}: c(x,y)\ge 1\},\quad O(x)\coloneq \{y\in A\setminus\{x\}: c(y,x)\ge 1\},
\]
and consider the set $A'=\{x\in A: O(x)\neq\varnothing\}$. 
First of all, note that, since $G^N=r+1\geq 1$, there exists $x^*\in A'$. Moreover, for every $x\in A'$, $O(x)\neq\varnothing$ and, since in particular $N\in \mathcal{PS}(A)$, we have that $I(x)\neq\varnothing$. Finally, for every $x\in A'$, $I(x)\subseteq A'$ since if $y\in I(x)$, then $x\in O(y)\neq\varnothing$. 

For every $x\in A'$, pick an element of $I(x)$ and denote it by $d(x)$. Thus, $d$ defines a function from $A'$ to $A'$ such that, for every $x\in A'$, $d(x)\neq x$.
Consider now the sequence $(x_j)_{j=1}^\infty$ in $A'$ recursively defined as follows: $x_1=x^*$ and, for every $j\in\mathbb{N}$, 
$x_{j+1}=d(x_{j})$. Thus, for every $j\in\mathbb{N}$, $x_{j+1}\in I(x_j)$ and thus $c(x_j,x_{j+1})\ge 1$.
Let 
\[
J\coloneq \{j\in \mathbb{N}: \exists k\in \mathbb{N}, k< j, \hbox{ such that } x_j=x_k\}.
\]
We show that $J\neq \varnothing$. Assume, by contradiction, that $J=\varnothing$. Then, for every $j,k\in \mathbb{N}$ with $k< j$, we have that $x_{j}\neq x_k$. 
As a consequence, the vertices in the sequence $(x_j)_{j=1}^\infty$ are distinct, against the fact that $A$ is finite. 
We can then consider $j^*\coloneq \min J$. Then there exists $k^*\in \mathbb{N}$ with $k^*<j^*$ such that $x_{j^*}=x_{k^*}$,
and, for every $j,k\in\mathbb{N}$ with $k<j<j^*$, $x_{j}\neq x_k$, that is, the vertices $x_1,\ldots, x_{j^*-1}$ are distinct. Since, for every $j\in\mathbb{N}$, 
$x_j\neq x_{j+1}$, we also have that $k^*\neq j^*-1$. Then $x_{k^*}\dots x_{j^*-1}x_{j^*}$ is a cycle of the complete directed graph $(A, A^2_*)$ because the vertices $x_{k^*},\dots, x_{j^*-1} $ are at least two and distinct and $x_{j^*}=x_{k^*}$.

Consider then the cycle network $C\coloneq \sum_{j=k^*}^{j^*-1} N_{(x_j,x_{j+1})}$ and note that $N-C\in \mathcal{N}_{\mathbb{N}_0}(A)$, because, for every $j\in \{k^*,\dots, j^*-1\}$, we have that $c(x_j,x_{j+1})\geq 1$. Moreover, we have that $N-C\in \mathcal{PS}(A)$ and $G^{N-C}=t$ for some $t\leq r$. It follows that 
$N=(N-C)+C,$ with $N-C\in \mathcal{PS}^t(A)$ for some $t\leq r$. By the inductive assumption, we have then that $N-C=N(0)$ or 
there exist $s\in \mathbb{N}$ and a sequence $(C_i)_{i=1}^s$ of cycle networks such that $N-C=\sum_{i=1}^s C_i$. Thus, $N=C$ or $N=C+\sum_{i=1}^s C_i$ and that proves that  Stat$(r+1)$ is true, as desired.
\end{proof}

\begin{proposition}\label{ps-2}
$\mathcal{PS}(A)=\mathbb{Q}\{C\in \mathcal{N}(A):C \hbox{ is a } k\hbox{-cycle, for some } 2\leq k\leq m \}$.
\end{proposition}
\begin{proof}
If $C$ is a $k$-cycle then trivially $C\in \mathcal{PS}(A)$. Thus, $\mathbb{Q}\{C\in \mathcal{N}(A):C \hbox{ is a } k\hbox{-cycle, for some } 2\leq k\leq m \}\subseteq\mathcal{PS}(A)$. Let us prove the opposite inclusion. Let $N\in \mathcal{PS}(A)$. 
It is immediate to observe that there are $l,h\in \mathbb{N}$ such that $\widetilde N\coloneq lN+N(h)\in  \mathcal{N}_{\mathbb{N}_0}(A)$.
Since $N\in \mathcal{PS}(A)$, $N(h)\in \mathcal{C}(A)\le \mathcal{PS}(A)$ and $\mathcal{PS}(A)$ is a subspace of $\mathcal{N}(A)$, we deduce that 
$\widetilde N\coloneq lN+N(h)\in  \mathcal{PS}(A)\cap  \mathcal{N}_{\mathbb{N}_0}(A)$.
By Proposition \ref{ps-cycles}, we then have $\widetilde N=N(0)$ or $\widetilde N=\sum_{i=1}^s C_i$ for some cycle networks $C_i$. It follows that $N=-\frac{1}{l}N(h)$ or $N=\frac{1}{l}\widetilde N-\frac{1}{l}N(h)=\frac{1}{l}\sum_{i=1}^s C_i-\frac{1}{l}N(h)$. Since $N(h)\in \mathcal{C}(A)\subseteq \mathcal{PS}(A)\cap \mathcal{N}_{\mathbb{N}_0}(A)$, using again Proposition \ref{ps-cycles}, we have that $N(h)$ is a sum of cycle networks and thus $N$ is a $\mathbb{Q}$-linear combination of cycle networks.
\end{proof}

\begin{proposition}\label{ps-3}  $\mathcal{PS}(A)=\mathbb{Q}\{C\in \mathcal{N}(A):C \hbox{ is a } m\hbox{-cycle } \}+\mathcal{R}(A).$
\end{proposition}

\begin{proof} By Proposition \ref{ps-2} and recalling that $\mathcal{PS}(A)\supseteq \mathcal{R}(A)$
and that $\mathcal{PS}(A)$ is a subspace of $\mathcal{N}(A)$, we have that
\[
\mathcal{PS}(A)\supseteq\mathbb{Q}\{C\in \mathcal{N}(A):C \hbox{ is a } m\hbox{-cycle} \}+\mathcal{R}(A).
\]
For the other inclusion, by Proposition \ref{ps-2},
it is enough to show that, for every $k$ with  $2\leq k\leq m-1$, every $k$-cycle network is a $\mathbb{Q}$-linear combination of $(k+1)$-cycle networks and of a reversal symmetric network. Consider then $2\leq k\leq m-1$. Let $C_{x_1\,x_2\dots x_k\,x_1}$ be the $k$-cycle network associated with the $k$-cycle
 $x_1\,x_2\dots x_k\,x_1$. Then we have that 
\begin{equation}\label{forma-ciclo}
C_{x_1\,x_2\dots x_k\,x_1}=N_{x_k\,x_1}+\sum_{i=1}^{k-1}N_{x_i\,x_{i+1}}.
\end{equation}
Recall that $x_1,\ldots,x_k$ are distinct. Let $x_{k+1}\in A\setminus\{x_1, \dots, x_{k}\}$ and let  $\widetilde N\coloneq \sum_{i=1}^k(N_{x_{k+1}\, x_i}+N_{x_i\,x_{k+1}})\in \mathcal{R}(A).$ 
We show that the following equality holds
\[
\widetilde N+(k-1)C_{x_1\,x_2\dots x_k\,x_1}=C_{x_{k+1}\,x_1\,x_2\dots x_k\,x_{k+1}}+C_{x_{k+1}\,x_2\dots x_k\,x_{1}\,x_{k+1}}+\cdots+C_{x_{k+1}x_k\,x_1\,x_2\dots\,x_{k-1}\,x_{k+1}}.
\]
Indeed, taking into account \eqref{forma-ciclo} and the definition of $\widetilde N$, we have that
\[
C_{x_{k+1}\,x_1\,x_2\dots x_k\,x_{k+1}}+C_{x_{k+1}\,x_2\dots x_k\,x_{1}\,x_{k+1}}+\cdots+C_{x_{k+1}x_k\,x_1\,x_2\dots\,x_{k-1}\,x_{k+1}}
\]
\[=\sum_{i=1}^k(N_{x_{k+1}\, x_i}+N_{x_i\,x_{k+1}})+(C_{x_1\,x_2\dots x_k\,x_1}-N_{x_k\,x_1})+(C_{x_1\,x_2\dots x_k\,x_1}-N_{x_1\,x_2})+\dots+(C_{x_1\,x_2\dots x_k\,x_1}-N_{x_{k-1}\,x_k})
\]
\[
=\widetilde N+(k-1)C_{x_1\,x_2\dots x_k\,x_1}.
\]
It follows that $C_{x_1\,x_2\dots x_k\,x_1}$ is a $\mathbb{Q}$-linear combination of $(k+1)$-cycle networks and of a reversal symmetric network, as desired.
\end{proof}

We conclude the section by defining two important properties a set of networks may meet.

\begin{definition}
Let $\mathcal{D}$ be a subset of $\mathcal{N}(A)$.
\begin{itemize}
	\item $\mathcal{D}$ is {\rm closed under permutation of vertices} {\rm (CPV)} if, 
for every $N\in \mathcal{D}$ and $\psi\in\mathrm{Sym}(A)$, $N^\psi\in \mathcal{D}$;
	\item $\mathcal{D}$ is {\rm closed under addition} {\rm (CA)} if, for every $N,M\in \mathcal{D}$, 
$N+M\in \mathcal{D}$.
\end{itemize}
\end{definition}
By Proposition \ref{generale-reti} the sets $\mathcal{C}(A)$, $\mathcal{R}(A)$, $\mathcal{B}_0(A)$, $\mathcal{B}(A)$, $\mathcal{PS}(A)$, $\mathcal{O}(A)$ and $\mathcal{N}(A)$ are subspaces and hence CA. Moreover, it is easily checked that they are CVP.

\section{Network solutions}\label{network solutions}
Let $\mathcal{D}\subseteq \mathcal{N}(A)$. A network solution ({\sc ns}) on $\mathcal{D}$  is a function from  $\mathcal{D}$  to $P_*(A)$. Thus, a network solution on $\mathcal{D}$  is a procedure that associates a nonempty set of vertices with any network belonging to  $\mathcal{D}$. A network solution on $\mathcal{N}(A)$ is simply called a network solution.

We propose three remarkable examples of network solutions: the net-outdegree {\sc ns}, denoted by $\mathscr{O}$, that associates with any $N\in\mathcal{N}(A)$ the set
\[
\mathscr{O}(N)\coloneq  \underset{x\in A}{\mathrm{argmax}}\,\delta^N(x);
\]
the net-indegree {\sc ns}, denoted by $\mathscr{I}$, that associates with any $N\in\mathcal{N}(A)$ the set
\[
\mathscr{I}(N)\coloneq  \underset{x\in A}{\mathrm{argmin}}\,\delta^N(x);
\]
the total {\sc ns}, denoted by $\mathscr{T}$, that associates with any $N\in\mathcal{N}(A)$ the set $\mathscr{T}(N)\coloneq A$.
Thus, the net-outdegree {\sc ns} selects the alternatives maximizing the net-outdegree;  the net-indegree {\sc ns} selects the alternatives minimizing the net-outdegree or, equivalently, maximizing the net-indegree (defined as the opposite of the net-outdegree); the total {\sc ns} always selects the whole set of the alternatives.

The next proposition highlights the importance of the set of pseudo-symmetric networks when we are dealing with the three {\sc ns}s just defined.

\begin{proposition}\label{flat-o} Let $N\in \mathcal{N}(A)$. Then the following facts hold true.
\begin{itemize}
\item[$(i)$] If $N\in \mathcal{PS}(A)$, then $\mathscr{O}(N)=\mathscr{I}(N)=A$.
\item[$(ii)$] If $N\not\in \mathcal{PS}(A)$, then $\mathscr{O}(N)\cap \mathscr{I}(N)=\varnothing$ and $|\{\mathscr{O}(N),\mathscr{I}(N),A\}|=3$.
\end{itemize}
\end{proposition}

\begin{proof} $(i)$ Let $N\in \mathcal{PS}(A)$. Then $\delta^N=0\in \mathbb{Q}^A$ and all the elements in $A$ are both maxima and minima for $\delta^N$. Thus, we get $\mathscr{O}(N)=\mathscr{I}(N)=A$. 

$(ii)$ Let $N\not \in \mathcal{PS}(A)$. We prove first that $\delta^N$ cannot be a constant function. Assume, by contradiction, that there exists $k\in \mathbb{Q}$ such that, for every $x\in A$, $\delta^N(x)=k$. Then, we have
\[
k|A|=\sum_{x\in A}\delta^N(x)=\sum_{x\in A}\left(\sum_{y\in A\setminus \{x\}}c(x,y)-\sum_{y\in A\setminus \{x\}}c(y,x)\right)
\]
\[
=\sum_{x\in A}\sum_{y\in A\setminus \{x\}}c(x,y)-\sum_{x\in A}\sum_{y\in A\setminus \{x\}}c(y,x)
=\sum_{(x,y)\in A^2_*}c(x,y)-\sum_{(x,y)\in A^2_*}c(y,x)=0, 
\]
which implies $k=0$. Hence we have that $N\in \mathcal{PS}(A)$, a contradiction. 

Since $\delta^N$ is not constant, we immediately deduce that $\mathscr{O}(N)\cap \mathscr{I}(N)=\varnothing$ and $|\{\mathscr{O}(N),\mathscr{I}(N),A\}|=3$.
\end{proof}

Let us introduce now some of the properties a {\sc ns} may meet. Those properties are crucial for our purposes as they are the ones involved in any characterization theorem of the paper.

\begin{definition}\label{ocoerenza}
Let $\mathcal{D}\subseteq\mathcal{N}(A)$ and $\mathscr{F}$ be a network solution on $\mathcal{D}$. 
\begin{itemize}
	\item $\mathscr{F}$ satisfies {\rm neutrality} if $\mathcal{D}$ is {\rm CPV} and, for every $N\in \mathcal{D}$ and $\psi\in\mathrm{Sym}(A)$,  $\mathscr{F}(N^\psi)=\psi(\mathscr{F}(N))$;
	\item $\mathscr{F}$ satisfies {\rm consistency} if $\mathcal{D}$ is {\rm CA} and, for every $N,M\in \mathcal{D}$ with $\mathscr{F}(N)\cap \mathscr{F}(M)\neq\varnothing$, $\mathscr{F}(N+M)=\mathscr{F}(N)\cap \mathscr{F}(M)$;
	\item $\mathscr{F}$ satisfies {\rm cancellation} if, for every $N\in \mathcal{D}\cap \mathcal{R}(A)$, $\mathscr{F}(N)=A$.
\end{itemize}
\end{definition}

We stress that stating that a network solution satisfies neutrality (resp. consistency) implies that the domain of that network solution is CPV (resp. CA). The following proposition shows that the net-outdegree {\sc ns}, the net-indegree {\sc ns} and the total {\sc ns} satisfy the three aforementioned properties when restricted to suitable domains.

\begin{proposition}\label{Osoddisfa}
Let $\mathcal{D}\subseteq\mathcal{N}(A)$. The following facts hold true.
\begin{itemize}
\item[$(i)$] If $\mathcal{D}$ is {\rm CPV}, then $\mathscr{O}_{|\mathcal{D}}$, $\mathscr{I}_{|\mathcal{D}}$ and $\mathscr{T}_{|\mathcal{D}}$ satisfy neutrality.
\item[$(ii)$] If $\mathcal{D}$ is {\rm CA}, then $\mathscr{O}_{|\mathcal{D}}$, $\mathscr{I}_{|\mathcal{D}}$ and $\mathscr{T}_{|\mathcal{D}}$ satisfy consistency.
\item[$(iii)$] $\mathscr{O}_{|\mathcal{D}}$, $\mathscr{I}_{|\mathcal{D}}$ and $\mathscr{T}_{|\mathcal{D}}$ satisfy cancellation.
\end{itemize}
\end{proposition}

\begin{proof} 
We prove $(i)$, $(ii)$ and $(iii)$ for $\mathscr{O}_{|\mathcal{D}}$ only. The proofs for $\mathscr{I}_{|\mathcal{D}}$ are analogous, while the ones for $\mathscr{T}_{|\mathcal{D}}$ are  straightforward.

$(i)$ Let $\mathcal{D}$ be {\rm CPV}. Let $N\in \mathcal{D}$ and $\psi\in \mathrm{Sym}(A)$. We have to prove that $\mathscr{O}(N^{\psi})= \psi( \mathscr{O}(N))$. For every $x\in A$, by \eqref{delta-psi}, we have that $\delta^{N^{\psi}}(x)=\delta^{N}(\psi^{-1}(x))$. As a consequence, $x^*\in \underset{x\in A}{\mathrm{argmax}}\,\delta^{N^{\psi}}(x)=\mathscr{O}(N^{\psi})$ if and only if $\psi^{-1}(x^*)\in \underset{x\in A}{\mathrm{argmax}}\,\delta^N(x)=\mathscr{O}(N).$ Thus, $\mathscr{O}(N^{\psi})= \psi (\mathscr{O}(N))$.

$(ii)$ Let $\mathcal{D}$ be {\rm CA}.  Let $N,M\in \mathcal{D}$ with $\mathscr{O}(N)\cap \mathscr{O}(M)\neq\varnothing$. We have to prove that $\mathscr{O}(N+M)=\mathscr{O}(N)\cap \mathscr{O}(M)$.
We first claim that 
\begin{equation}\label{mxsomma}
 \underset{x\in A}{\mathrm{argmax}}\,\delta^N(x)\cap \underset{x\in A}{\mathrm{argmax}}\,\delta^M(x)= \underset{x\in A}{\mathrm{argmax}}\,(\delta^N(x)+\delta^M(x)).
\end{equation}
Let $x^*\in \underset{x\in A}{\mathrm{argmax}}\,\delta^N(x)\cap \underset{x\in A}{\mathrm{argmax}}\,\delta^M(x)=\mathscr{O}(N)\cap \mathscr{O}(M)\neq\varnothing.$ Then we have $\delta^N(x^*)\geq \delta^N(x)$ and $\delta^M(x^*)\geq \delta^M(x)$, for all $x\in A.$ It follows that $\delta^N(x^*)+\delta^M(x^*)\geq \delta^N(x)+\delta^M(x)$, for all $x\in A.$ Hence  the maximum of the function $\delta^N+\delta^M\in \mathbb{Q}^A$ is attained at $x^*$ and $x^*\in  \underset{x\in A}{\mathrm{argmax}}\,(\delta^N(x)+\delta^M(x))$. In particular, $\underset{x\in A}{\mathrm{argmax}}\,\delta^N(x)\cap \underset{x\in A}{\mathrm{argmax}}\,\delta^M(x)\subseteq  \underset{x\in A}{\mathrm{argmax}}\,(\delta^N(x)+\delta^M(x))$. 
Let next $y^*\in  \underset{x\in A}{\mathrm{argmax}}\,(\delta^N(x)+\delta^M(x))$. Then we have 
\begin{equation}\label{intermedia}
\delta^N(y^*)+\delta^M(y^*)=\delta^N(x^*)+\delta^M(x^*).
\end{equation}
By the definition of $x^*$, we have $\delta^N(y^*)\leq \delta^N(x^*)$ and $\delta^M(y^*)\leq \delta^M(x^*)$. Hence \eqref{intermedia} implies $\delta^N(y^*)=\delta^N(x^*)$ and $\delta^M(y^*)=\delta^M(x^*)$. Thus $y^*\in \underset{x\in A}{\mathrm{argmax}}\,\delta^N(x)\cap \underset{x\in A}{\mathrm{argmax}}\,\delta^M(x)$ and \eqref{mxsomma} is proved. 

Now, by \eqref{mxsomma} and the linearity of $\delta$, we finally obtain 
\[
\mathscr{O}(N+M)=\underset{x\in A}{\mathrm{argmax}}\,(\delta^{N+M}(x))=\underset{x\in A}{\mathrm{argmax}}\,(\delta^N(x)+\delta^M(x))
\]
\[
=\underset{x\in A}{\mathrm{argmax}}\,\delta^N(x)\cap \underset{x\in A}{\mathrm{argmax}}\,\delta^M(x)=\mathscr{O}(N)\cap \mathscr{O}(M).
\]

$(iii)$ Let $N\in \mathcal{D}\cap \mathcal{R}(A)$.
Since $N\in \mathcal{PS}(A)=\mathrm{Ker}(\delta)$, we have that, for every $x\in A$, $\delta^N(x)=0$. Thus, $\mathscr{O}(N)=A$.
\end{proof}

In the rest of the section, we propose some useful propositions involving the properties in Definition \ref{ocoerenza} that will be used throughout the paper.

\begin{proposition}\label{sum-R}
Let $\mathcal{D}\subseteq \mathcal{N}(A)$ and $\mathscr{F}$ be  a {\sc ns} on $\mathcal{D}$ satisfying  consistency and cancellation.  Let $N\in \mathcal{D}$ and $R\in \mathcal{D}\cap \mathcal{R}(A)$. Then $\mathscr{F}(N+R)=\mathscr{F}(N)$.
\end{proposition}

\begin{proof} By cancellation,  we have $\mathscr{F}(R)=A$ and thus $\mathscr{F}(N)\cap \mathscr{F}(R)\neq\varnothing$. Thus, by consistency, we have that
$\mathscr{F}(N+R)=\mathscr{F}(N)\cap \mathscr{F}(R)=\mathscr{F}(N)\cap A=\mathscr{F}(N)$.
\end{proof}

\begin{proposition}\label{k-sum}
Let $\mathcal{D}\subseteq \mathcal{N}(A)$ and $\mathscr{F}$ be  a {\sc ns} on $\mathcal{D}$ satisfying consistency. Let $k\in\mathbb{N}$ with $k\ge 2$ and $N_1,\ldots, N_k\in \mathcal{D}$ be such that $\mathscr{F}(N_1)\cap \ldots \cap \mathscr{F}(N_k)\neq\varnothing$. Then 
\begin{equation}\label{ind-1}
\mathscr{F}(N_1+...+N_k)=\mathscr{F}(N_1)\cap \ldots\cap \mathscr{F}(N_k).
\end{equation}
\end{proposition}

\begin{proof}
Let us work by induction on $k$. If $k=2$, then \eqref{ind-1} is true by consistency of $\mathscr{F}$. Consider now $k\in\mathbb{N}$, assume that \eqref{ind-1} is true for $k$ and prove it for $k+1$. Let $N_1,\ldots, N_{k+1}\in \mathcal{D}$ be such that $\mathscr{F}(N_1)\cap \ldots \cap \mathscr{F}(N_{k+1})\neq\varnothing$. In particular, $\mathscr{F}(N_1)\cap \ldots \cap \mathscr{F}(N_{k})\neq\varnothing$ and 
so, by the inductive assumption, we get $\mathscr{F}(N_1+...+N_k)=\mathscr{F}(N_1)\cap \ldots\cap \mathscr{F}(N_k)$. Once defined $N=N_1+...+N_k$, we have that
\[
\mathscr{F}(N)\cap\mathscr{F}(N_{k+1})=\left(\mathscr{F}(N_1)\cap \ldots \cap \mathscr{F}(N_{k})\right)\cap \mathscr{F}(N_{k+1})\neq\varnothing.
\]
Thus, by consistency, 
\[
\mathscr{F}(N+N_{k+1})=\mathscr{F}(N)\cap\mathscr{F}(N_{k+1})=\mathscr{F}(N_1)\cap \ldots \cap \mathscr{F}(N_{k+1}).
\] 
Since $N+N_{k+1}=N_1+...+N_k$, we get \eqref{ind-1} for $k+1$.
\end{proof}

\begin{proposition}\label{scalar-vector-q}
Let $\mathcal{D}\subseteq\mathcal{N}(A)$ and $\mathscr{F}$ be a {\sc ns} on $\mathcal{D}$ satisfying consistency. Let $N\in \mathcal{D}$ and $q\in\mathbb{Q}$ be such that $q>0$ and $qN\in \mathcal{D}$. Then $\mathscr{F}(qN)=\mathscr{F}(N)$.
\end{proposition}

\begin{proof} We first show the thesis for $q=k\in\mathbb{N}$, working  by induction on $k$. If $k=1$ there is nothing to prove.
Consider now $k\in\mathbb{N}$, assume that $kN\in\mathcal{D}$, $\mathscr{F}(kN)=\mathscr{F}(N)$ and prove that 
$(k+1)N\in\mathcal{D}$ and $\mathscr{F}((k+1)N)=\mathscr{F}(N)$. 
Since $(k+1)N=kN+N$ and $\mathcal{D}$ is CA  we have that $(k+1)N\in\mathcal{D}$.
Since $\mathscr{F}(kN)\cap \mathscr{F}(N)=\mathscr{F}(N)\cap \mathscr{F}(N)=\mathscr{F}(N)\neq\varnothing$,
using consistency we have that
$\mathscr{F}((k+1)N)=\mathscr{F}(kN+N)=\mathscr{F}(kN)\cap \mathscr{F}(N)=\mathscr{F}(N)$, as desired.

Let us next consider $q\in\mathbb{Q}$ with $q>0$. Then there exists
$k\in\mathbb{N}$ such that $kq\in\mathbb{N}$. Since $k, kq\in\mathbb{N}$, by what shown before, we have that $\mathscr{F}((kq)N)=\mathscr{F}(N)$ and $\mathscr{F}((kq)N)=\mathscr{F}(k(qN))=\mathscr{F}(qN)$. Then we conclude that $\mathscr{F}(qN)=\mathscr{F}(N)$.
\end{proof}

\section{Main characterization result for {\sc ns}s}\label{cinque}

This long section is devoted to the proof of the main result of the paper, namely Theorem \ref{characterization-net} stated in Section \ref{sec7}. Its proof requires some preliminary work, proposed in Sections \ref{sub-1} and \ref{sub-2}. Finally, Section \ref{sub-4} 
is devoted to describe some specific and remarkable situations where Theorem \ref{characterization-net} can be applied.

\subsection{First preparatory theorem}\label{sub-1}

In the following proposition we identify the possible values a neutral network solution can get on outstar networks, provided its domain contains at least one of them.

\begin{proposition}\label{NUOVO}
Let $\mathcal{D}$ be such that $\mathcal{D}\cap\{N_x:x\in A\}\neq\varnothing$, and $\mathscr{F}$ be a {\sc ns} on $\mathcal{D}$ satisfying neutrality. Then $\{N_x:x\in A\}\subseteq \mathcal{D}$. Moreover, one and only one of the following facts holds true:
\begin{itemize}
\item[$(a)$] for every $x\in A$, $\mathscr{F}(N_x)=\{x\}$;
\item[$(b)$] for every $x\in A$, $\mathscr{F}(N_x)=A\setminus \{x\}$;
\item[$(c)$] for every $x\in A$, $\mathscr{F}(N_x)=A$.
\end{itemize}
\end{proposition}

\begin{proof}
Let us first prove $\{N_x:x\in A\}\subseteq \mathcal{D}$. 
Let $y\in A$. We have to prove that $N_{y}\in\mathcal{D}$. 
Since $\mathcal{D}\cap \{N_x:x\in A\}\neq\varnothing$, we know that there exists 
$z\in A$ such that $N_z\in\mathcal{D}$. Since $\mathscr{F}$ satisfies neutrality, we also know that $\mathcal{D}$ is CPV. Consider then $\psi=(yz)\in \mathrm{Sym}(A)$. Thus, $N_{y}=N_{\psi(z)}=(N_z)^\psi\in \mathcal{D}$.

Let us prove now that, for every $x\in A$, $\mathscr{F}(N_x)=\{x\}$ or $\mathscr{F}(N_x)=A\setminus \{x\}$ or  $\mathscr{F}(N_x)=A$.
Consider then $x\in A$ and assume by contradiction that $\mathscr{F}(N_x)=X$, where $X\neq \{x\}$, $X\neq A\setminus \{x\}$ and  $X\neq A$. 
Then, there exist $y,z\in A\setminus \{x\}$ with $y\neq z$ such that $y\in X$ and $z\not\in X$.
Considering $\psi=(yz)$, we have that $(N_x)^\psi=N_{\psi(x)}=N_x$ and, using neutrality, we get
\[
X=\mathscr{F}(N_x)=\mathscr{F}((N_x)^\psi)=\psi(\mathscr{F}(N_x))=\psi(X).
\]
Thus we have that $X=\psi(X)$. However, $z=\psi(y)\in \psi(X)$ and $z\not\in X$, a contradiction.

Consider now $x\in A$. We know that $\mathscr{F}(N_x)=\{x\}$ or $\mathscr{F}(N_x)=A\setminus \{x\}$ or  $\mathscr{F}(N_x)=A$. 

Assume first that $\mathscr{F}(N_x)=\{x\}$. We  prove that, for every $y\in A$, we have $\mathscr{F}(N_y)=\{y\}$. Indeed, considering $\psi=(xy)$, we have $N_y=(N_x)^\psi$. Neutrality of $\mathscr{F}$ implies then that $\mathscr{F}(N_y)=\mathscr{F}((N_x)^\psi)=\psi(\mathscr{F}(N_x))=\{\psi(x)\}=\{y\}$.

Assume now that $\mathscr{F}(N_x)=A\setminus \{x\}$ and prove that, for every $y\in A$, $\mathscr{F}(N_y)=A\setminus \{y\}$. Indeed, considering $\psi=(xy)$, we have that $N_y=(N_x)^\psi$. Neutrality of $\mathscr{F}$ implies then that $\mathscr{F}(N_y)=\mathscr{F}((N_x)^\psi)=\psi(\mathscr{F}(N_x))=A\setminus \{\psi(x)\}=A\setminus \{y\}$.
 
Assume finally that $\mathscr{F}(N_x)=A$ and prove that, for every $y\in A$, $\mathscr{F}(N_y)=A$. Indeed, considering $\psi=(xy)$, we have that $N_y=(N_x)^\psi$. Neutrality of $\mathscr{F}$ implies then that $\mathscr{F}(N_y)=\mathscr{F}((N_x)^\psi)=\psi(\mathscr{F}(N_x))=\psi(A)=A$.
\end{proof}

The next proposition shows that the properties of neutrality, consistency and cancellation, if jointly satisfied, force a network solution to always associate the whole set of vertices with each pseudo-symmetric network, provided that the domain of the network solution is a vector space containing all the pseudo-symmetric networks.

\begin{proposition}\label{const-su-PS}
Let $\mathcal{V}$ be a subspace of $\mathcal{N}(A)$ such that 
$\mathcal{PS}(A)\subseteq \mathcal{V}$ and
$\mathscr{F}$ be a {\sc ns} on $\mathcal{V}$ satisfying neutrality, consistency and cancellation.
Then, for every $N\in \mathcal{PS}(A)$, $\mathscr{F}(N)=A$.
\end{proposition}

\begin{proof}
By Proposition \ref{ps-3}, using $\mathcal{PS}(A)\subseteq \mathcal{V}$, we have that any $m$-cycle network belongs to $\mathcal{V}$.
Let $C$ be a $m$-cycle network with associated cycle $\gamma=x_1\cdots x_m x_1$.  Note that $A=\{x_1,\dots,x_m\}$. Let $q\in \mathbb{Q}$ and consider $\psi\in \mathrm{Sym}(A)$ given by $\psi=(x_1\,x_2\,\dots\, x_m)$. Then, for every $n\in[m]$, we have $(qC)^{\psi^n}=qC$ and thus, since $\mathscr{F}$ satisfies neutrality, we also have 
$
\mathscr{F}(qC)=\mathscr{F}((qC)^{\psi^n})=\psi^n(\mathscr{F}(qC)).
$
This means that $\mathscr{F}(qC)$ is a nonempty subset of $A$ which is transformed into itself by the action of the group $\langle \psi\rangle$. 
Since $\langle \psi\rangle$ has only one orbit on $A$, the only possibility is that 
\begin{equation}\label{F-cicli}
\mathscr{F}(qC)=A.
\end{equation}
Consider now $N\in \mathcal{PS}(A)$. By Proposition \ref{ps-3}, we know that $\mathcal{PS}(A)=\mathbb{Q}\{C\in \mathcal{N}(A):C \hbox{ is a } m\hbox{-cycle} \}+\mathcal{R}(A)$. Thus, we can find $s\in\mathbb{N}$, rational numbers $q_1,\ldots, q_s$, $m$-cycle networks $C_1,\ldots, C_s$ and $R\in \mathcal{R}(A)$ such that $N=R+\sum_{i=1}^s q_i C_i$. By \eqref{F-cicli}, for every $i\in [s]$, we have that $\mathscr{F}(q_i C_i)=A$. Moreover, since $R\in \mathcal{V}\cap\mathcal{R}(A)$,
by cancellation  of $\mathscr{F}$, we also have $\mathscr{F}(R)=A$. Thus, since $\mathscr{F}$ satisfies consistency and cancellation, we can apply Propositions \ref{sum-R} and \ref{k-sum} and finally obtain
\[
\mathscr{F}(N)=\mathscr{F}\left(R+\sum_{i=1}^s q_i C_i\right)=\mathscr{F}\left(\sum_{i=1}^s q_i C_i\right)=\bigcap_{i=1}^s\mathscr{F}\left(q_i C_i\right)=A.
\]
\end{proof}

The following characterization theorem constitutes the conceptual core of the paper. Indeed, Theorem \ref{characterization-net}, the main result of the paper, is basically a corollary of it once some further technical facts will be established (Section \ref{sub-2}).

\begin{theorem}\label{subspace}
Let $\mathcal{V}$ be a subspace of $\mathcal{N}(A)$ such that $\mathcal{V}\cap \mathcal{PS}(A)\in\{\mathcal{R}(A),\mathcal{PS}(A)\}$ and $\mathcal{V}\cap \{N_x:x\in A\}\neq\varnothing$. Let  $\mathscr{F}$ be a {\sc ns} on $\mathcal{V}$ satisfying neutrality, consistency and cancellation. Then $\mathscr{F}\in \{\mathscr{O}_{|\mathcal{V}},\mathscr{I}_{|\mathcal{V}},\mathscr{T}_{|\mathcal{V}}\}$.
\end{theorem}

\begin{proof} 
First, note that the condition $\mathcal{V}\cap \mathcal{PS}(A)\in\{\mathcal{R}(A),\mathcal{PS}(A)\}$ implies $\mathcal{R}(A)\subseteq \mathcal{V}$. Moreover, by Proposition \ref{NUOVO}, we also have $\{N_x:x\in A\}\subseteq\mathcal{V}$.

Let $\delta_{\mid\mathcal{V}}$ be the restriction to $\mathcal{V}$ of the net-outdegree (see Section \ref{net-sec}) and note that $\delta_{\mid\mathcal{V}}$ is a linear map from $\mathcal{V}$ to $\mathbb{Q}^A$.
By the fact that, for every $x\in A$, the networks $N_x$ belong to $\mathcal{V}$, we deduce that
$\mathbb{Q}\{\delta(N_x):x\in A\}\subseteq \delta_{\mid\mathcal{V}} (\mathcal{V})$. Using Proposition \ref{im-delta-1}, we have then that 
$\delta(\mathcal{N}(A))\subseteq\delta_{\mid\mathcal{V}} (\mathcal{V})$. Since obviously $\delta_{\mid\mathcal{V}} (\mathcal{V})\subseteq\delta(\mathcal{N}(A)),$ we then get 
\begin{equation}\label{immagine}
\delta_{\mid\mathcal{V}} (\mathcal{V})=\delta(\mathcal{N}(A))=\mathbb{Q}\{\delta(N_x):x\in A\}.
\end{equation} 
Moreover, again by Proposition \ref{im-delta-1}, we have $\mathrm{dim}\left(\delta_{\mid\mathcal{V}} (\mathcal{V})\right)=m-1$.
For what concerns the kernel, we have that 
$\mathrm{Ker}\left(\delta_{\mid\mathcal{V}}\right) =\mathcal{V}\cap \mathrm{Ker}\left(\delta\right)=\mathcal{V}\cap \mathcal{PS}(A)\in\{\mathcal{R}(A),\mathcal{PS}(A)\}$.
Observe now that,
\begin{equation}\label{claim}
\mbox{for every $N\in \mathrm{Ker}\left(\delta_{\mid\mathcal{V}}\right)$, we have $\mathscr{F}(N)=A$.}
\end{equation}
Indeed, if $\mathrm{Ker}\left(\delta_{\mid\mathcal{V}}\right) =\mathcal{R}(A)$, that immediately follows from the fact that $\mathscr{F}$ satisfies cancellation. If instead $\mathrm{Ker}\left(\delta_{\mid\mathcal{V}}\right) =\mathcal{PS}(A)$, that immediately follows from Proposition \ref{const-su-PS}.

We now show that two networks in $\mathcal{V}$ with the same image by $\delta$ also have the same image by $\mathscr{F}$. 
Let $N,M\in \mathcal{V}$ be such that $\delta_{\mid\mathcal{V}}(N)=\delta_{\mid\mathcal{V}}(M)$. Then we have $\delta_{\mid\mathcal{V}}(N-M)=0\in \mathbb{Q}^A$, that is, $N-M\in \mathrm{Ker}\left(\delta_{\mid\mathcal{V}}\right)$ and thus, by \eqref{claim}, $\mathscr{F}(N-M)=A.$ As a consequence, by consistency of $\mathscr{F}$, we get
$
\mathscr{F}(N)=\mathscr{F}(M+(N-M))=\mathscr{F}(M).
$

By Proposition \ref{NUOVO} we know that $\mathscr{F}$ satisfies exactly one of the conditions $(a)$, $(b)$ and $(c)$ stated in that proposition. 
We complete the proof of the theorem showing that 
if $\mathscr{F}$ satisfies condition $(a)$, then $\mathscr{F}=\mathscr{O}_{|\mathcal{V}}$;
if $\mathscr{F}$ satisfies condition $(b)$, then $\mathscr{F}=\mathscr{I}_{|\mathcal{V}}$; 
if $\mathscr{F}$ satisfies condition $(c)$, then $\mathscr{F}=\mathscr{T}_{|\mathcal{V}}$.

\vspace{2mm}

\noindent{\it Case $(a)$.} Assume that, for every $x\in A$, $\mathscr{F}(N_x)=\{x\}$ and prove that, for every $N\in \mathcal{V}$, we have $\mathscr{F}(N)=\mathscr{O}_{|\mathcal{V}}(N)$.

Fix $N\in \mathcal{V}$ and look to the rational numbers $\delta^N(x)$ for $x\in A$. Let us enumerate the elements of $A$ is such a way that $A=\{x_1,\ldots,x_m\}$ and 
\begin{equation}\label{delta-in-fila}
\delta^N(x_1)\geq \delta^N(x_2)\geq \cdots\geq \delta^N(x_m).
\end{equation}
Our task is then to show that 
\[
\mathscr{F}(N)=\left\{x_k\in A:\delta^N(x_k)=\delta^N(x_1)\right\}.
\]
By \eqref{immagine}, there exist rational numbers $\lambda_1,\ldots, \lambda_m$ such that 
\begin{equation}\label{comb-lin}
\delta(N)=\sum_{j=1}^m \lambda_j\delta(N_{x_j}).
\end{equation}
Moreover, we have that 
\begin{equation}\label{comb-lin-triv}
\sum_{j=1}^m \delta(N_{x_j})=0.
\end{equation}
Indeed, since $\mathcal{C}(A)\subseteq\mathcal{PS}(A)=\mathrm{Ker}\left(\delta\right)$, we have that
\[
\sum_{j=1}^m \delta(N_{x_j})=\delta\left(\sum_{j=1}^m N_{x_j}\right)=\delta(N(1))=0.
\]
From \eqref{comb-lin} and \eqref{comb-lin-triv}, an easy computation leads to 
\begin{equation}\label{comb-lin2}
\delta(N)=\sum_{j=1}^{m-1} (\lambda_j-\lambda_{j+1})\sum_{i\leq j}\delta(N_{x_i}).
\end{equation}
Using the linearity of $\delta$ and \eqref{comb-lin2} we get
\[
\delta(N)=\delta\left(\sum_{j=1}^{m-1} (\lambda_j-\lambda_{j+1})\sum_{i\leq j}N_{x_i}\right),
\]
where $\sum_{j=1}^{m-1} (\lambda_j-\lambda_{j+1})\sum_{i\leq j}N_{x_i}\in \mathcal{V}$. Since we have proved that two networks in $\mathcal{V}$ with the same image by $\delta$ have also the same image by $\mathscr{F}$, we have that 
\begin{equation}\label{comb-lin4}
\mathscr{F}(N)=\mathscr{F}\left(\sum_{j=1}^{m-1} (\lambda_j-\lambda_{j+1})\sum_{i\leq j}N_{x_i}\right).
\end{equation}
From \eqref{comb-lin}, for every $i\in[m]$, we have   
\begin{equation}\label{conto}
\delta^N(x_i)=\sum_{j=1}^m \lambda_j\delta^{N_{x_j}}(x_i)=\left(\sum_{j=1, j\neq i}^m \lambda_j(-1)\right) +\lambda _i(m-1)=m\lambda _i-\sum_{j=1}^m \lambda_j.
\end{equation}
It follows that, for every $i,j\in[m]$, we have 
\begin{equation}\label{delta-ij}
\delta^N(x_i)-\delta^N(x_j)=m(\lambda _i-\lambda _j).
\end{equation}
By \eqref{delta-in-fila}, $i\leq j$ implies $\delta^N(x_i)\geq\delta^N(x_j)$ and therefore, using \eqref{delta-ij}, we also get $\lambda _i\geq \lambda _j$. 
Hence, we deduce that 
\begin{equation}\label{lambda-in-fila}
\lambda_1\geq \lambda_2\geq \cdots\geq \lambda_m.
\end{equation}
Moreover, \eqref{delta-ij} also implies
$\{x_k\in A:\delta^N(x_k)=\delta^N(x_1)\}=\{x_k\in A:\lambda_k=\lambda_1\}$ and hence our purpose becomes to prove
\begin{equation}\label{what2}
\mathscr{F}(N)=\{x_k\in A:\lambda_k=\lambda_1\}.
\end{equation}
For every $j\in[m]$, define $A_j=\{x_i\in A: i\le j\}$. We claim that,
\begin{equation}\label{claim2}
\mbox{for every $j\in [m]$, $\mathscr{F}\left(\sum_{i\leq j}N_{x_i}\right)=A_j$.}
\end{equation}
We show \eqref{claim2} by backward induction on $j\in[m]$. If $j=m$, then 
\[
\sum_{i\leq m}N_{x_i}=N(1)\in \mathcal{V}\cap \mathcal{R}(A),
\]
and so, since $\mathscr{F}$ satisfies cancellation, we get
$\mathscr{F}\left(\sum_{i\leq m}N_{x_i}\right)=\mathscr{F}(N(1))=A=A_m$. 
Let now $j\in [m-1]$ be such that $\mathscr{F}\left(\sum_{i\leq j+1}N_{x_i}\right)=A_{j+1}$ and prove that 
$\mathscr{F}\left(\sum_{i\leq j}N_{x_i}\right)=A_j$.
Consider the subgroup $S_j$ of $\mathrm{Sym}(A)$ given by $S_j=\{\psi\in \mathrm{Sym}(A): \psi(A_j)=A_j\}$ and observe that, for every $\psi\in S_j$, we have that 
\[
\left(\sum_{i\leq j}N_{x_i}\right)^{\psi}=\sum_{i\leq j}N^{\psi}_{x_i}=\sum_{i\leq j}N_{\psi(x_i)}=\sum_{i\leq j}N_{x_i}.
\]
Assume, by contradiction, that there exists $k\ge j+1$ such that $x_k\in \mathscr{F}(\sum_{i\leq j}N_{x_i})$. Consider now $\psi=(x_{j+1} \cdots x_m)\in S_j$. Then, using the neutrality of $\mathscr{F}$, we have that, for every $n\in \mathbb{N}$, 
\[
\psi^n(x_k)\in \psi^n\left(\mathscr{F}\left(\sum_{i\leq j}N_{x_i}\right)\right)=\mathscr{F}\left(\left(\sum_{i\leq j}N_{x_i}\right)^{\psi^n}\right)
=\mathscr{F}\left(\sum_{i\leq j}N_{x_i}\right).
\]
As a consequence, we have $\{x_{j+1},\dots, x_m\}\subseteq \mathscr{F}\left(\sum_{i\leq j}N_{x_i}\right)$.
By inductive assumption, we know $\mathscr{F}\left(\sum_{i\leq j+1}N_{x_i}\right)=A_{j+1}$. Since by assumption we know that 
$\mathscr{F}(N_{x_{j+1}})=\{x_{j+1}\}$ and since $\mathscr{F}\left(\sum_{i\leq j}N_{x_i}\right)\cap \mathscr{F}(N_{x_{j+1}})=\{x_{j+1}\}$, by consistency of $\mathscr{F}$, we get
\[
A_{j+1}=\mathscr{F}\left(\sum_{i\leq j+1}N_{x_i}\right)=\mathscr{F}\left(\sum_{i\leq j}N_{x_i}\right)\cap \mathscr{F}\left(N_{x_{j+1}}\right)=\{x_{j+1}\},
\]
that is a contradiction since $|A_{j+1}|=j+1\ge 2$. Thus, we have shown that $\mathscr{F}\left(\sum_{i\leq j}N_{x_i}\right)\subseteq A_j$.
Since $\mathscr{F}\left(\sum_{i\leq j}N_{x_i}\right)\neq\varnothing$, there exists $x_k\in \mathscr{F}\left(\sum_{i\leq j}N_{x_i}\right)$ with $k\le j$. 
Let $\psi=(x_1\cdots x_j)\in S_j$. Then, using the neutrality of $\mathscr{F}$ we have that, for every $n\in \mathbb{N}$,
\[
\psi^n(x_k)\in \psi^n\left(\mathscr{F}\left(\sum_{i\leq j}N_{x_i}\right)\right)=\mathscr{F}\left(\left(\sum_{i\leq j}N_{x_i}\right)^{\psi^n}\right)
=\mathscr{F}\left(\sum_{i\leq j}N_{x_i}\right).
\]
As a consequence, $\mathscr{F}\left(\sum_{i\leq j}N_{x_i}\right)\supseteq A_j$. Thus, we can finally conclude that $\mathscr{F}(\sum_{i\leq j}N_{x_i})= A_j$. That completes the proof of \eqref{claim2}.
Note that, by \eqref{claim}, $\mathscr{F}(N(0))=A$. Then, by \eqref{claim2} and Proposition \ref{scalar-vector-q}, we have that, for every $j\in[m-1]$, 
\[ 
\mathscr{F}\left((\lambda_j-\lambda_{j+1})\sum_{i\leq j}N_{x_i}\right)=B_j\coloneq 
\left\{
\begin{array}{lll}
A&\mbox{ if } \lambda_j=\lambda_{j+1}\\
A_j&\mbox{ if } \lambda_j>\lambda_{j+1}\\
\end{array}
\right.
\]
holds.
Thus
\[
\bigcap_{j=1}^{m-1}\mathscr{F}\left((\lambda_j-\lambda_{j+1})\sum_{i\leq j}N_{x_i}\right)=\bigcap_{j=1}^{m-1}B_j \supseteq\bigcap_{j=1}^{m-1}A_j =\{x_1\}\neq \varnothing.
\]
By \eqref{comb-lin4},  consistency of $\mathscr{F}$ and Proposition \ref{k-sum}, we therefore have that
\[
\mathscr{F}(N)
=\mathscr{F}\left(\sum_{j=1}^{m-1} (\lambda_j-\lambda_{j+1})\sum_{i\leq j}N_{x_i}\right)=B\coloneq \bigcap_{j=1}^{m-1}B_j.
\]
In order to get \eqref{what2}, we are left with proving that
\begin{equation}\label{finimola2}
B=\{x_k\in A:\lambda_k=\lambda_1\}.
\end{equation}
Let $x_k\in A$ be such that $\lambda_1>\lambda_k$. Then, by \eqref{lambda-in-fila}, there exists $j^*\in [m-1]$ with $j^*+1\leq k$ such that $
\lambda_{j^*}>\lambda_{j^*+1}$. Thus, $B_{j^*}=A_{j^*}$ and $B\subseteq A_{j^*}$. 
Note that $k\ge j^*+1$ implies $x_k\not\in A_{j^*}$. We then deduce that $x_k\not\in B$. Hence we have 
$B\subseteq\{x_k\in A:\lambda_k=\lambda_1\}$. 
Let now $x_k\in A$ be such that $\lambda_1=\lambda_k$. We prove that $x_k\in B$ showing that $x_k\in B_j$ for all $j\in[m-1]$.
Consider $j\in[m-1]$. Assume first that $\lambda_j>\lambda_{j+1}$. Then $B_j=A_j$. 
Then, by \eqref{lambda-in-fila}, we also have $ \lambda_1>\lambda_{s}$ for all $s\in \{j+1,\ldots,m\}$.
As a consequence, we have $k\neq s$ for all $s\in \{j+1,\ldots,m\}$ that implies $k\leq j$. It follows that $x_k\in A_k\subseteq A_j=B_j$. 
Assume now that $\lambda_j=\lambda_{j+1}$. Thus $B_j=A$ and thus, clearly, $x_k\in B_j$. 
Hence we also have $\{x_k\in A:\lambda_k=\lambda_1\}\subseteq B$  and \eqref{finimola2} is  proved.

\vspace{2mm}

\noindent {\it Case $(b)$.} Assume that, for every $x\in A$, $\mathscr{F}(N_x)=A\setminus \{x\}$ and prove that, for every $N\in \mathcal{V}$, we have $\mathscr{F}(N)=\mathscr{I}_{|\mathcal{V}}(N)$.

Let $\mathscr{G}$ be the {\sc ns} on $\mathcal{V}$ defined, for every $N\in\mathcal{V}$, by $\mathscr{G}(N)=\mathscr{F}(-N)$.
Note that, being $\mathcal{V}$ a subspace of $\mathcal{N}(A)$, we have that $N\in\mathcal{V}$ implies $-N\in \mathcal{V}$.
We have that $\mathscr{G}$ satisfies neutrality. Indeed, let $N\in\mathcal{V}$ and $\psi\in\mathrm{Sym}(A)$.
Then $\mathscr{G}(N^\psi)=\mathscr{F}(-N^\psi)=\mathscr{F}((-N)^\psi)=\psi(\mathscr{F}(-N))=\psi(\mathscr{G}(N))$.
We also have that $\mathscr{G}$ satisfies consistency. Indeed, let $N,M\in\mathcal{V}$ be such that 
$\mathscr{G}(N)\cap \mathscr{G}(M)\neq \varnothing$. Thus $\mathscr{F}(-N)\cap \mathscr{F}(-M)\neq \varnothing$ and hence, using the 
fact that $\mathscr{F}$ satisfies consistency, we get
\[
\mathscr{G}(N+M)=\mathscr{F}(-N-M)=\mathscr{F}(-N)\cap \mathscr{F}(-M)=\mathscr{G}(N)\cap \mathscr{G}(M).
\]
We finally have that $\mathscr{G}$ satisfies cancellation. Indeed, let $N\in\mathcal{R}(A)$. Then $-N\in \mathcal{R}(A)$ and, using the fact that $\mathscr{F}$ satisfies cancellation, we get $\mathscr{G}(N)=\mathscr{F}(-N)=A$.

Let us prove now that, for every $x\in A$, $\mathscr{G}(N_x)=\{x\}$. Let $x\in A$. By assumption we know that $\mathscr{F}(N_x)=A\setminus\{x\}$. Thus, $\mathscr{G}(-N_x)=\mathscr{F}(N_x)=A\setminus\{x\}$.
Assume now by contradiction that $\mathscr{G}(N_x)\neq \{x\}$. Thus, $\mathscr{G}(N_x)\cap (A\setminus\{x\})\neq\varnothing$ and hence
$\mathscr{G}(N_x)\cap \mathscr{G}(-N_x)\neq\varnothing$.
Since $\mathscr{G}$ satisfies consistency and cancellation, we conclude that
\[
A=\mathscr{G}(N(0))=\mathscr{G}(N_x+(-N_x))=\mathscr{G}(N_x)\cap \mathscr{G}(-N_x)=\mathscr{G}(N_x)\cap (A\setminus\{x\}).
\] 
Since $\mathscr{G}(N_x)\cap (A\setminus\{x\})\neq A$, we have a contradiction.

Thus, we can apply to $\mathscr{G}$ the Case $(a)$ previously proved. We get then that
$\mathscr{G}=\mathscr{O}_{|\mathcal{V}}$. As a consequence we have that, for every $N\in\mathcal{V}$,
\[
\mathscr{F}(N)=\mathscr{G}(-N)=\underset{x\in A}{\mathrm{argmax}}\,\delta^{(-N)}(x)
=\underset{x\in A}{\mathrm{argmax}}\,-\delta^N(x)=\underset{x\in A}{\mathrm{argmin}}\,\delta^N(x)=\mathscr{I}(N).
\]
We conclude then that $\mathscr{F}=\mathscr{T}_{|\mathcal{V}}$, as desired

\vspace{2mm}

\noindent {\it Case $(c)$.} Assume that, for every $x\in A$, $\mathscr{F}(N_x)=A$ and prove that, for every $N\in \mathcal{V}$, we have $\mathscr{F}(N)=A$.

Fix $N\in \mathcal{V}$ and look to the rational numbers $\delta^N(x)$ for $x\in A$. Let us enumerate the elements of $A$ is such a way that $A=\{x_1,\ldots,x_m\}$ and \eqref{delta-in-fila} holds true. By \eqref{immagine}, there exist rational numbers $\lambda_1,\ldots, \lambda_m$ such that \eqref{comb-lin} holds true. Using the same argument developed in Case $(a)$ we deduce \eqref{comb-lin4} and \eqref{lambda-in-fila}.
Let us prove now that
\begin{equation}\label{equalA}
\mbox{for every $j\in[m-1]$,\ }\mathscr{F}\left(\sum_{i\leq j}N_{x_i}\right)=A.
\end{equation}
Indeed, consider $j\in[m-1]$. Since we know that, for every $x\in A$, $\mathscr{F}\left(N_{x}\right)=A$, by Proposition \ref{k-sum}, we have that
\[
\mathscr{F}\left(\sum_{i\leq j}N_{x_i}\right)=\bigcap_{i\le j}\mathscr{F}\left(N_{x_i}\right)=A.
\]
Let us prove  now that
\begin{equation}\label{equalA2}
\mbox{for every $j\in[m-1]$,\  }\mathscr{F}\left((\lambda_j-\lambda_{j+1})\sum_{i\leq j}N_{x_i}\right)=A.
\end{equation}
Indeed, let $j\in[m-1]$. By \eqref{lambda-in-fila} we know that $\lambda_j-\lambda_{j+1}\ge 0$. Assume first that $\lambda_j-\lambda_{j+1}= 0$. Since $N(0)\in \mathcal{R}(A)\subseteq \mathcal{V}$, using cancellation we obtain 
\[
\mathscr{F}\left((\lambda_j-\lambda_{j+1})\sum_{i\leq j}N_{x_i}\right)=\mathscr{F}(N(0))=A.
\]
If instead $\lambda_j-\lambda_{j+1}> 0$, we prove \eqref{equalA2}  by recalling \eqref{equalA} and applying Proposition \ref{scalar-vector-q}.
Finally, by Proposition \ref{k-sum} and using \eqref{equalA2} and \eqref{comb-lin4}, we conclude that 
\[
\mathscr{F}(N)=\mathscr{F}\left(\sum_{j=1}^{m-1} (\lambda_j-\lambda_{j+1})\sum_{i\leq j}N_{x_i}\right)=\bigcap_{j=1}^{m-1}
\mathscr{F}\left((\lambda_j-\lambda_{j+1})\sum_{i\leq j}N_{x_i}\right)=A.
\]
\end{proof}

\subsection{Extending network solutions}\label{sub-2}

Given $\mathcal{D},\mathcal{D}'\subseteq \mathcal{N}(A)$ with $\mathcal{D}\subseteq\mathcal{D}'$, a {\sc ns} $\mathscr{F}$ on $\mathcal{D}$ and a {\sc ns} $\mathscr{F}'$ on $\mathcal{D}'$, we say that $\mathscr{F}'$ is an extension of $\mathscr{F}$ if, for every $N\in \mathcal{D}$, $\mathscr{F}'(N)=\mathscr{F}(N)$.

The next result, whose proof is in the appendix, states conditions that allow to extend a network solution defined on a subset $\mathcal{D}$ of $\mathcal{N}(A)$ and satisfying neutrality, consistency and cancellation to the vector space $\mathbb{Q}\mathcal{D}+\mathcal{R}(A)$ still maintaining those properties. That passage to a vector space paves the road for exploiting Theorem \ref{subspace}. 

\begin{theorem}\label{ext-finale}
Let $\mathcal{D}\subseteq\mathcal{N}(A)$ and $\mathscr{F}$ be a {\sc ns} on $\mathcal{D}$ satisfying neutrality, consistency and cancellation. Assume that, for every $N\in \mathbb{Z}\mathcal{D}$, there exists $R\in \mathcal{D}\cap \mathcal{R}(A)$ such that $N+R\in\mathcal{D}$. Then there exists $\mathscr{F}':\mathbb{Q}\mathcal{D}+\mathcal{R}(A)\to P_*(A)$ extension of $\mathscr{F}$ satisfying neutrality, consistency and cancellation.
\end{theorem}

\subsection{The main result for {\sc ns}s}\label{sec7}

Let us introduce first the following crucial definition.

\begin{definition}\label{regular-net}
Let  $\mathcal{D}\subseteq \mathcal{N}(A)$. $\mathcal{D}$ is regular if the following conditions hold true:
\begin{itemize}
	\item[$(a)$] $\mathcal{D}$ is $\mathrm{CVP}$ and $\mathrm{CA}$;
	\item[$(b)$] for every $N\in \mathbb{Z}\mathcal{D}$, there exists $R\in \mathcal{D}\cap \mathcal{R}(A)$ such that 
	$N+R\in\mathcal{D}$;
	\item[$(c)$] $(\mathbb{Q}\mathcal{D}+\mathcal{R}(A))\cap \mathcal{PS}(A)\in\{\mathcal{R}(A),\mathcal{PS}(A)\}$;
	\item[$(d)$] $(\mathbb{Q}\mathcal{D}+\mathcal{R}(A))\cap \{N_x: x\in A\}\neq\varnothing$.
\end{itemize}
\end{definition}

Regular sets of networks are one of the main ingredients in Theorem \ref{characterization-net}. In Section \ref{sub-4} we show that many significant sets of networks are in fact regular. For the moment observe, for instance, that $\mathcal{N}(A)$ is regular. Note also that, since for every  $\mathcal{D}\subseteq \mathcal{N}(A)$, we have $\mathbb{Z}\mathcal{D}\supseteq \{N(0)\}$, condition $(b)$ implies that every regular set of networks is nonempty.

The next proposition shows a property that is particularly meaningful, especially in the light of Theorem \ref{characterization-net}, namely that each regular domain contains an element that is not pseudo-symmetric.

\begin{proposition}\label{no-pseudo}
Let  $\mathcal{D}\subseteq \mathcal{N}(A)$ be regular. Then $\mathcal{D}\not\subseteq \mathcal{PS}(A)$.
\end{proposition}

\begin{proof}
Assume by contradiction that $\mathcal{D}\subseteq \mathcal{PS}(A)$. Since $\mathcal{PS}(A)$ is a subspace of $\mathcal{N}(A)$ containing $\mathcal{R}(A)$, we have that $\mathbb{Q}\mathcal{D}+\mathcal{R}(A)\subseteq \mathcal{PS}(A)$.
As a consequence, since $\{N_x:x\in A\}\cap \mathcal{PS}(A)=\varnothing$, we get $\big(\mathbb{Q}\mathcal{D}+\mathcal{R}(A)\big)\cap \{N_x:x\in A\}=\varnothing$. That contradicts condition $(d)$ in Definition \ref{regular-net}.
\end{proof}

By Proposition \ref{Osoddisfa}, if $\mathcal{D}$ is CVP and CA, we know that $\mathscr{O}_{|\mathcal{D}}$, $\mathscr{I}_{|\mathcal{D}}$ and $\mathscr{T}_{|\mathcal{D}}$ satisfy neutrality, consistency and cancellation. Theorem \ref{characterization-net} below shows that the regularity of $\mathcal{D}$ guarantees that $\mathscr{O}_{|\mathcal{D}}$, $\mathscr{I}_{|\mathcal{D}}$ and $\mathscr{T}_{|\mathcal{D}}$ are the unique {\sc ns}s  on $\mathcal{D}$ fulfilling those properties.

\begin{theorem}\label{characterization-net}
Let  $\mathcal{D}\subseteq \mathcal{N}(A)$ be regular and $\mathscr{F}$ be a {\sc ns} on $\mathcal{D}$ satisfying neutrality, consistency and cancellation.
Then, $\mathscr{F}\in\{\mathscr{O}_{|\mathcal{D}},\mathscr{I}_{|\mathcal{D}},\mathscr{T}_{|\mathcal{D}}\}$.  Moreover, for every $N\in \mathcal{D}\setminus \mathcal{PS}(A)$, we have that
\begin{itemize}
\item[$(i)$] if $\mathscr{F}(N)=\mathscr{O}(N)$, then $\mathscr{F}=\mathscr{O}_{|\mathcal{D}}$;
\item[$(ii)$] if $\mathscr{F}(N)=\mathscr{I}(N)$, then $\mathscr{F}=\mathscr{I}_{|\mathcal{D}}$;
\item[$(iii)$] if $\mathscr{F}(N)=A$, then $\mathscr{F}=\mathscr{T}_{|\mathcal{D}}$.
\end{itemize}
\end{theorem}

\begin{proof} 
Set $\mathcal{V}=\mathbb{Q}\mathcal{D}+\mathcal{R}(A)$. Since $\mathcal{D}$ is regular, we know that 
$\mathcal{V}\cap \mathcal{PS}(A)\in\{\mathcal{R}(A),\mathcal{PS}(A)\}$ and 
$\mathcal{V}\cap \{N_x\in x\in A\}\neq\varnothing$. By Theorem \ref{ext-finale}, there exists $\mathscr{F}':\mathcal{V}\to P_*(A)$ extension of $\mathscr{F}$ satisfying neutrality, consistency and cancellation. 
Applying then Theorem \ref{subspace} to $\mathscr{F}'$ we deduce that 
$\mathscr{F}'\in\{\mathscr{O}_{|\mathcal{V}},\mathscr{I}_{|\mathcal{V}},\mathscr{T}_{|\mathcal{V}}\}$. 
We conclude then that $\mathscr{F}=\mathscr{F}'_{|\mathcal{D}}\in\{\mathscr{O}_{|\mathcal{D}},\mathscr{I}_{|\mathcal{D}},\mathscr{T}_{|\mathcal{D}}\}$.

Let us consider now $N\in \mathcal{D}\setminus \mathcal{PS}(A)$. By Proposition \ref{flat-o}, we have that $|\{\mathscr{O}(N),\mathscr{I}(N),A\}|=3$. As a consequence, if $\mathscr{F}(N)=\mathscr{O}(N)$, then we have $\mathscr{F}(N)\neq\mathscr{I}(N)$ and 
$\mathscr{F}(N)\neq\mathscr{T}(N)$, and hence $\mathscr{F}=\mathscr{O}_{|\mathcal{D}}$. 
If $\mathscr{F}(N)=\mathscr{I}(N)$, then we have $\mathscr{F}(N)\neq\mathscr{O}(N)$ and 
$\mathscr{F}(N)\neq\mathscr{T}(N)$, and hence $\mathscr{F}=\mathscr{I}_{|\mathcal{D}}$.
If $\mathscr{F}(N)=\mathscr{T}(N)$, then we have $\mathscr{F}(N)\neq\mathscr{O}(N)$ and 
$\mathscr{F}(N)\neq\mathscr{I}(N)$, and hence $\mathscr{F}=\mathscr{T}_{|\mathcal{D}}$. That proves $(i)$, $(ii)$ and $(iii)$.
\end{proof}

We emphasize that, by Proposition \ref{no-pseudo}, when $\mathcal{D}$ is regular there always exists a network $N^*\in\mathcal{D}\setminus \mathcal{PS}(A)$. Thus, by Theorem \ref{characterization-net}, if $\mathscr{F}$ is a {\sc ns} on $\mathcal{D}$ satisfying neutrality, consistency and cancellation, we can completely decide the nature of $\mathscr{F}$, just by looking at $\mathscr{F}(N^*)$. In particular, by Theorem \ref{characterization-net}, we can easily characterize any restriction of the net-outdegree network solution to a regular subset of $\mathcal{N}(A)$, once the following further property is considered.

\begin{definition}
Let  $\mathcal{D}\subseteq \mathcal{N}(A)$ and $\mathscr{F}$ be a {\sc ns} on $\mathcal{D}$. $\mathscr{F}$ satisfies $\mathscr{O}$-coherence if there exists $N\in \mathcal{D}\setminus \mathcal{PS}(A)$ such that $\mathscr{F}(N)=\mathscr{O}(N)$.
\end{definition}

Note that, in the light of Proposition \ref{flat-o}, a {\sc ns} $\mathscr{F}$ on $\mathcal{D}$ satisfies $\mathscr{O}$-coherence if and only if there exists an element in $\mathcal{D}$ at which $\mathscr{F}$ coincides with $\mathscr{O}$ and at which $\mathscr{F}$ differs from $\mathscr{I}$. We are now ready to state and prove the announced characterization theorem.

\begin{theorem}\label{crucial-net}
Let $\mathcal{D}\subseteq \mathcal{N}(A)$ be regular. Then $\mathscr{O}_{|\mathcal{D}}$ is the unique {\sc ns} on $\mathcal{D}$ satisfying neutrality, consistency, cancellation and $\mathscr{O}$-coherence.
\end{theorem}

\begin{proof}
Since $\mathcal{D}$ is CPV and CA, by Proposition \ref{Osoddisfa} we know that $\mathscr{O}_{|\mathcal{D}}$ satisfies neutrality, consistency and cancellation. By Proposition \ref{no-pseudo} we know that there exists $N\in \mathcal{D}\setminus\mathcal{PS}(A)$ and hence we immediately get that $\mathscr{O}_{|\mathcal{D}}$ satisfies $\mathscr{O}$-coherence, as well. Let now $\mathscr{F}$ be a {\sc ns} on $\mathcal{D}$ that satisfies neutrality, consistency, cancellation and $\mathscr{O}$-coherence. Since $\mathscr{F}$ satisfies $\mathscr{O}$-coherence, we know that there exists $N\in \mathcal{D}\setminus\mathcal{PS}(A)$ such that $\mathscr{F}(N)=\mathscr{O}(N)$. Since $\mathcal{D}$ is regular and $\mathscr{F}$ satisfies neutrality, consistency and cancellation, by Theorem \ref{characterization-net}, we conclude that $\mathscr{F}=\mathscr{O}_{|\mathcal{D}}$.
\end{proof}

\subsection{Regular sets of networks}\label{sub-4}

Because of Theorem \ref{characterization-net}, it becomes very important to identify the regular subsets of $\mathcal{N}(A)$. 
Indeed, each regular set of networks leads to a specification of Theorem \ref{characterization-net}.
The following theorem shows that there are many interesting regular sets of networks.

\begin{theorem}
Let $\mathbb{X}\in \{ \mathbb{N}_0, \mathbb{Z}, \mathbb{Q}_+,\mathbb{Q}\}$. Then the sets
$\mathcal{N}_{\mathbb{X}}(A)$, $\mathcal{N}_{\mathbb{X}}(A)\cap\mathcal{B}(A)$ and $\mathbb{X}\{N_x:x\in A\}$
are regular subsets of $\mathcal{N}(A)$.
\end{theorem}

\begin{proof}
Let us prove first that $\mathcal{D}\coloneq \mathcal{N}_{\mathbb{X}}(A)$ is regular. We have to show that it satisfies conditions $(a)$, $(b)$, $(c)$ and $(d)$ of Definition \ref{regular-net}. Clearly, we have that $\mathcal{D}$ is CVP and CA and thus $\mathcal{D}$ satisfies condition $(a)$. Pick now $N\in \mathbb{Z}\mathcal{D}$ and let $k=\min_{(x,y)\in A^2_*} c^N(x,y)$. Consider $t\in\mathbb{N}$ such that $t+k>0$. Thus, we have that $N(t)\in \mathcal{D}\cap \mathcal{R}(A)$ and $N+N(t)\in \mathcal{D}$ and thus condition $(b)$ is satisfied. Note that, for every $(x,y)\in A^2_*$, we have that $N_{xy}\in \mathcal{N}_{\mathbb{N}_0}(A)$. Thus, $\mathbb{Q}\mathcal{N}_{\mathbb{N}_0}(A)=\mathcal{N}(A)$ and so $\mathbb{Q}\mathcal{N}_{\mathbb{N}_0}(A)+ \mathcal{R}(A)=\mathcal{N}(A)$. Since $\mathcal{N}_{\mathbb{N}_0}(A)\subseteq \mathcal{D}$, we also have that $\mathbb{Q}\mathcal{D}+ \mathcal{R}(A)=\mathcal{N}(A)$. Thus,
$(\mathbb{Q}\mathcal{D}+\mathcal{R}(A))\cap \mathcal{PS}(A)=\mathcal{PS}(A)$ and $(c)$ is satisfied.
For every $x\in A$, we have that $N_x\in \mathcal{D}$ and so $(d)$ is satisfied.

Let us prove now that $\mathcal{D}\coloneq \mathcal{N}_{\mathbb{X}}(A)\cap\mathcal{B}(A)$ satisfies conditions $(a)$, $(b)$, $(c)$ and $(d)$ of Definition \ref{regular-net}.  Clearly, we have that $\mathcal{D}$ is CVP and CA and hence $\mathcal{D}$ satisfies condition $(a)$. 
Pick now $N\in \mathbb{Z}\mathcal{D}$ and let $k=\min_{(x,y)\in A^2_*} c^N(x,y)$. Consider $t\in\mathbb{N}$ such that $t+k>0$. Thus, we have that $N(t)\in \mathcal{D}\cap \mathcal{R}(A)$ and $N+N(t)\in \mathcal{D}$ and hence condition $(b)$ is satisfied. Given $(x,y)\in A^2_*$, let 
\[
R=\sum_{z\in A\setminus \{x,y\}}(N_{xz}+N_{zx})+\sum_{z\in A\setminus \{x,y\}}(N_{yz}+N_{zy})+K_{A\setminus \{x,y\}}.
\]
We have that $R\in \mathcal{R}(A)$ and that $2N_{xy}+R\in \mathcal{B}_{\mathbb{N}_0}(A)$.
Since $N_{xy}=\frac{1}{2}(2N_{xy}+R)-\frac{1}{2}R$, we deduce that $N_{xy}\in \mathbb{Q}\mathcal{B}_{\mathbb{N}_0}(A)+\mathcal{R}(A)$ and since $\{N_{xy}: (x,y)\in A^2_*\}$ is a basis for $\mathcal{N}(A)$, we get
$\mathcal{N}(A)\subseteq \mathbb{Q}\mathcal{B}_{\mathbb{N}_0}(A)+\mathcal{R}(A)$. 
By $\mathcal{B}_{\mathbb{N}_0}(A)\subseteq \mathcal{D}$, we have then that $\mathcal{N}(A)\subseteq \mathbb{Q}\mathcal{D}+\mathcal{R}(A)$ and, since the opposite inclusion is trivial, we conclude that $\mathbb{Q}\mathcal{D}+\mathcal{R}(A)=\mathcal{N}(A)$. 
Thus, $(\mathbb{Q}\mathcal{D}+\mathcal{R}(A))\cap \mathcal{PS}(A)=\mathcal{PS}(A)$ and  $(c)$ is satisfied.
Given $x\in A$, we have that $2N_x+R\in \mathcal{D}$, where $R=K_{A\setminus\{x\}}\in \mathcal{R}(A)$. Thus, 
$N_x=\frac{1}{2}(2N_x+R)-\frac{1}{2}R\in \mathbb{Q}\mathcal{D}+\mathcal{R}(A)$ and $(c)$ is satisfied.

Let us finally prove that $\mathcal{D}\coloneq \mathbb{X}\{N_x:x\in A\}$ satisfies conditions $(a)$, $(b)$, $(c)$ and $(d)$ of Definition \ref{regular-net}.  
Note first that $\{N(k):k\in \mathbb{N}_0\}\subseteq \mathcal{D}\cap \mathcal{R}(A)$  because, for every $k\in \mathbb{N}_0$, 
we have $N(k)=\sum_{x\in A}kN_x\in \mathcal{D}\cap \mathcal{R}(A)$. Clearly, we have that $\mathcal{D}$ is CVP and CA and hence $\mathcal{D}$ satisfies condition $(a)$. Pick now $N\in \mathbb{Z}\mathcal{D}$ and let $k=\min_{(x,y)\in A^2_*} c^N(x,y)$. Consider $t\in\mathbb{N}$ such that $t+k>0$. Thus, we have that $N(t)\in \mathcal{D}\cap \mathcal{R}(A)$ and $N+N(t)\in \mathcal{D}$ and hence condition $(b)$ is satisfied.
Observe further that $\mathbb{Q}\mathcal{D}=\mathcal{O}(A)$. Thus, by Proposition \ref{generale-reti}$(iii.g)$, we have that $(c)$ holds.
Given now $x\in A$, we have that $N_x\in \mathcal{D}$ and hence condition $(d)$ is satisfied.
\end{proof}

\section{Preference profiles}\label{sezione-profili}

A relation $R$ on $A$ is a subset of $A^2$, that is, an element of $P(A^2)$. The set of relations on $A$ is denoted by $\mathbf{R}(A)$.
Let $R\in\mathbf{R}(A)$. 
Given $x,y\in A$, we sometimes write $x\succeq_R y$ instead of $(x,y)\in R$;
$x\succ_R y$ instead of $(x,y)\in R$ and $(y,x)\not\in R$; $x\sim_R y$ instead of $(x,y)\in R$ and $(y,x)\in R$; $x\perp_R y$ instead of $(x,y)\notin R$ and $(y,x)\notin R$. For every $x\in A$, we define the sets
\begin{equation}\label{LU}
L(R,x)\coloneq \{y\in A: x\succ_R y\},\quad U(R,x)\coloneq \{y\in A: y\succ_R x\},
\end{equation}
\[
I(R,x)\coloneq \{y\in A: x\sim_R y\},\quad IN(R,x)\coloneq \{y\in A: x\perp_R y\}.
\]
Note that $L(R,x)$, $U(R,x)$, $I(R,x)$ and $IN(R,x)$ are pairwise disjoint and $L(R,x)\cup U(R,x)\cup I(R,x)\cup IN(R,x)=A$.
Given $\psi\in\mathrm{Sym}(A)$, we define $R^\psi=\{(x,y)\in A^2:(\psi^{-1}(x),\psi^{-1}(y))\in R\}$. We also set 
$R^r=\{(x,y)\in A^2:(y,x)\in R\}$. Moreover, let
\[
\mathrm{Max}(R)\coloneq \{x\in A: \forall y\in A,\; y\not\succ_R x\}, \quad \mathrm{Min}(R)\coloneq \{x\in A: \forall y\in A,\; x\not\succ_R y\},
\]
be the sets of maximal and minimal elements of $R$, respectively.

From now on, we interpret $A$ as the set of alternatives. Consider a countably infinite set $V$ whose elements are to be interpreted as potential voters. For simplicity, we assume $V=\mathbb{N}$. Elements of $\mathbf{R}(A)$ are interpreted as preference relations on $A$. Let us consider the set
\[
\mathbf{R}(A)^*=\bigcup_{\substack{
I\subseteq V\\
 I\neq\varnothing\, \mathrm{finite}
}}
 \mathbf{R}(A)^I.
\]
An element of $\mathbf{R}(A)^*$ is called preference profile. 
Thus, a preference profile is a function from a finite nonempty subset of $V$ to $\mathbf{R}(A)$. 
Given $p\in \mathbf{R}(A)^*$, we denote by $\mathrm{Dom}(p)$ the domain of $p$ and, for every $i\in \mathrm{Dom}(p)$, 
$p(i)\in  \mathbf{R}(A)$ is interpreted as the preference relation of voter $i$ on the set of alternatives $A$.

If $p\in \mathbf{R}(A)^*$, we denote by $p^r$ the element of $\mathbf{R}(A)^*$ such that $\mathrm{Dom}(p^r)=\mathrm{Dom}(p)$ and for every $i\in \mathrm{Dom}(p)$, $p^r(i)=p(i)^r$. If $p\in \mathbf{R}(A)^*$ and $\psi \in \mathrm{Sym}(A)$, we denote by $p^\psi$ the element of $\mathbf{R}(A)^*$ such that $\mathrm{Dom}(p^{\psi})=\mathrm{Dom}(p)$ and, for every $i\in \mathrm{Dom}(p)$, $p^\psi(i)=p(i)^\psi$.

If $p,p'\in \mathbf{R}(A)^*$, we say that $p, p'$ are disjoint if $\mathrm{Dom}(p)\cap \mathrm{Dom}(p')=\varnothing$; we say that $p'$ is a clone of $p$  if $|\mathrm{Dom}(p)|= |\mathrm{Dom}(p')|$ and there exists a bijection $\varphi:\mathrm{Dom}(p)\rightarrow \mathrm{Dom}(p')$ such that, for every $i\in \mathrm{Dom}(p)$, we have  $p'(\varphi(i))=p(i)$; we say that $p'$ is a disjoint clone of $p$  if $p'$ is a clone of $p$ and $p$ and $p'$ are disjoint. If $p, p'\in \mathbf{R}(A)^*$ are disjoint, we denote by $p+p'$ the element of $\mathbf{R}(A)^*$ such that $\mathrm{Dom}(p+p')=\mathrm{Dom}(p)\cup\mathrm{Dom}(p')$ and such that, for every $i\in \mathrm{Dom}(p)$, we have $(p+p')(i)=p(i)$ and, for every $i\in \mathrm{Dom}(p')$, we have $(p+p')(i)=p'(i)$.

Note that if $p, p'\in \mathbf{R}(A)^*$ are disjoint, then $p+p'=p'+p$. Moreover, if $p, p', p''\in \mathbf{R}(A)^*$ are pairwise disjoint,
then $p+p'$ and $p''$ are disjoint, $p$ and $p'+p''$ are disjoint and $(p+p')+p''=p+(p'+p'')$. Thus we can write $p+p'+p''$ without ambiguity. A similar observation holds for any number of pairwise disjoint preference profiles.

Given $R\in \mathbf{R}(A)$, let $N(R)=(A,A^2_*,c_R)$ be the network in $\mathcal{N}_{\mathbb{N}_0}(A)$ whose capacity $c_R$ is defined, for every $(x,y)\in A^2_*$, by $c_R(x,y)=1$ if $(x,y)\in R$ and $c_R(x,y)=0$ if $(x,y)\notin R$. Note that $R\subseteq A_d^2$ if and only if $N(R)=N(0).$

Consider then the function $N:\mathbf{R}(A)^*\rightarrow \mathcal{N}_{\mathbb{N}_0}(A)$ defined, for every  $p\in \mathbf{R}(A)^*$,  by
\[
N(p)\coloneq \sum_{i\in\mathrm{Dom}(p)} N(p(i)).
\] 
Given $p\in \mathbf{R}(A)^*$, we call $N(p)$ the network associated with $p$ and we denote its capacity by $c_p$. Thus, we have that
\[
c_p= \sum_{i\in \mathrm{Dom}(p)} c_{p(i)}.
\]
For every $\mathbf{D}\subseteq \mathbf{R}(A)^*$, we set $N(\mathbf{D})=\{N(p): p\in \mathbf{D}\}.$  We call $N(\mathbf{D})$ the set of networks of  $\mathbf{D}$.

Let $p,p'\in \mathbf{R}(A)^*$  and $\psi\in\mathrm{Sym}(A)$. If $p,p'$ are disjoint, then the following properties hold true 
\begin{equation}\label{prop-reti-profili}
N(p^r)=N(p)^r, \qquad
N(p^\psi)=N(p)^\psi, \qquad  N(p+p')=N(p)+N(p').
\end{equation}
Note also that if $p$ and $p'$ are clones, then $N(p)=N(p')$. As a consequence, if $q$ is a clone of $p^r$ disjoint from $p$ then, using \eqref{prop-reti-profili}, we obtain
\begin{equation}\label{-np}
N(p+q)=N(p)+N(q)=N(p)+N(p^r)=N(p)+N(p)^r\in \mathcal{R}(A).
\end{equation}

Let us introduce now some properties a set of preference profiles may meet.

\begin{definition}
Let $\mathbf{D}\subseteq \mathbf{R}(A)^*$. 
\begin{itemize}
	\item $\mathbf{D}$ is {\rm closed under permutation of alternatives} {\rm (CPA)} if, for every $p\in \mathbf{D}$ and $\psi\in\mathrm{Sym}(A)$, $p^\psi\in \mathbf{D}$;
	\item $\mathbf{D}$ is {\rm closed under addition} {\rm (CA)} if, for every $p, p'\in \mathbf{D}$ disjoint, $p+p'\in \mathbf{D}$; 
	\item $\mathbf{D}$ is {\rm coherent with clones} {\rm (CWC)} if, for every $k\in\mathbb{N}$ and $p_1,\ldots,p_k\in \mathbf{D}$, there exist $q_1,\ldots,q_k\in \mathbf{D}$ pairwise disjoint such that, for every $i,j\in[k]$, $q_i$ is a clone of $p_i$ disjoint from $p_j$; 
	\item $\mathbf{D}$ is {\rm coherent with reversal symmetry} {\rm (CRS)} if, for every $p\in \mathbf{D}$, there exists $q\in \mathbf{D}$ disjoint from $p$ such that $N(p+q)\in\mathcal{R}(A)$.
\end{itemize}
\end{definition}

The next proposition establishes how the properties of a set of preference profiles $\mathbf{D}$ affect the properties of its set of networks.

\begin{proposition}\label{nuovo-lemma} Let $\mathbf{D}\subseteq \mathbf{R}(A)^*$. Then the following facts hold true.
\begin{itemize}
  \item[$(i)$] If $\mathbf{D}$ is {\rm CPA}, then $N(\mathbf{D})$ is {\rm CPV}.
	\item[$(ii)$] If $\mathbf{D}$ is {\rm CA} and {\rm CWC}, then $N(\mathbf{D})$ is {\rm CA}.
	\item[$(iii)$] If $\mathbf{D}$ is {\rm CPA}, {\rm CA} and {\rm CWC}, then, for every $p\in\mathbf{D}$, there exists $q\in \mathbf{D}$ disjoint from $p$ such that 
  $N(p+q)\in N(\mathbf{D})\cap \mathcal{C}(A)$. In particular, $\mathbf{D}$ is {\rm CRS}.
	\item[$(iv)$] If $\mathbf{D}$ is {\rm CPA}, {\rm CA}, {\rm CWC} and such that there exists $p\in \mathbf{D}$ with $N(p)\neq N(0)$, then, for every $k\in\mathbb{N}$, there exists $t\in\mathbb{N}$ with $t\ge k$ such that $N(t)\in N(\mathbf{D})\cap \mathcal{C}(A)$.
\end{itemize}
\end{proposition}

\begin{proof} 
$(i)$ Assume that $\mathbf{D}$ is {\rm CPA} and show that $N(\mathbf{D})$ is {\rm CPV}.
Let $N\in N(\mathbf{D})$ and $\psi\in \mathrm{Sym}(A)$. Then there exists $p\in \mathbf{D}$ such that $N=N(p)$. Since $\mathbf{D}$ is {\rm CPA} we have that $p^\psi\in\mathbf{D}$. Thus, $N^\psi=N(p)^\psi=N(p^\psi)\in N(\mathbf{D})$. 

$(ii)$ Assume that $\mathbf{D}$ is {\rm CA} and {\rm CWC} and show that $N(\mathbf{D})$ is {\rm CA}.
Let $N,M\in N(\mathbf{D})$. Then there exist $p,p'\in \mathbf{D}$ such that $N=N(p)$ and $M=N(p')$. Since $\mathbf{D}$ is CWC
we can find $\hat{p},\hat{p}'\in \mathbf{D}$ disjoint such that $\hat{p}$ is a clone of $p$, $\hat{p}'$ is a clone of $p'$.
Since $\mathbf{D}$ is {\rm CA} we have that $\hat{p}+\hat{p}'\in\mathbf{D}$. 
Thus, $N+M=N(\hat{p})+N(\hat{p}')=N(\hat{p}+\hat{p}')\in N(\mathbf{D})$.

$(iii)$ Assume that $\mathbf{D}$ is {\rm CPA}, {\rm CA} and {\rm CWC}. Consider $p\in \mathbf{D}$. Let $\psi_1,\ldots,\psi_k$ be an enumeration of the $k=m!-1$ elements of  $\mathrm{Sym}(A)\setminus \{id\}$. Since $\mathbf{D}$ is CPA, for every $i\in[k]$, $p^{\psi_i}\in \mathbf{D}$.
Since $\mathbf{D}$ is {\rm CWC},
there exist $\hat{p}_1,\ldots, \hat{p}_k\in \mathbf{D}$ pairwise disjoint such that, for every $i\in[k]$, $\hat{p}_i$ is a clone of $p^{\psi_i}$ and $\hat{p}_i$ is disjoint from $p$.
Consider then $q\coloneq \hat{p}_1+\ldots+\hat{p}_k$ and note that $q$ is disjoint from $p$.
Since $\mathbf{D}$ is CA, we have that $q\in \mathbf{D}$ and $p+q\in \mathbf{D}$ and hence $N(p+q)\in N(\mathbf{D})$. We complete the proof showing that 
$N(p+q)\in\mathcal{C}(A)$.
By \eqref{prop-reti-profili}, we have that
\[
N(p+q)=N\left(p+\sum_{i=1}^k\hat{p}_i\right)=N(p)+\sum_{i=1}^kN(\hat{p}_i)=N(p)+\sum_{i=1}^kN(p^{\psi_i})
\]
\[
=\sum_{\psi\in \mathrm{Sym}(A)}N(p^{\psi})=\sum_{\psi\in \mathrm{Sym}(A)}N(p)^{\psi}.
\]
Thus, for every $\sigma\in \mathrm{Sym}(A)$,
\[
N(p+q)^\sigma=\left(\sum_{\psi\in \mathrm{Sym}(A)}N(p)^{\psi}\right)^\sigma=\sum_{\psi\in \mathrm{Sym}(A)}\left(N(p)^{\psi}\right)^\sigma
\] 
\[
=\sum_{\psi\in \mathrm{Sym}(A)}N(p)^{\sigma\psi}=\sum_{\psi\in \mathrm{Sym}(A)}N(p)^{\psi}=N(p+q).
\]
By Proposition \ref{costanti-sim}, we conclude that $N(p+q)\in\mathcal{C}(A)\leq \mathcal{R}(A)$.

$(iv)$ Assume that $\mathbf{D}$ is {\rm CPA}, {\rm CA} and {\rm CWC} and that there exists $p\in \mathbf{D}$ such that $N(p)\neq N(0)$.
By $(iii)$ we know that there exists $q\in \mathbf{D}$ disjoint from $p$ such that $N(p+q)\in N(\mathbf{D})\cap \mathcal{C}(A)$. Let $p'\coloneq p+q$.
Since $\mathbf{D}$ is CA, we have that $p'\in \mathbf{D}$. Moreover, $N(p')=N(p)+N(q)$ and, since $N(p),N(q)\in \mathcal{N}_{\mathbb{N}_0}(A)$
and $N(p)\neq N(0)$, we deduce that $N(p')=N(r)$ for some $r\in\mathbb{N}$. Fixed now $k\in\mathbb{N}$, there exists $h\in\mathbb{N}$ such that $t\coloneq hr\ge k$. Now, due to the fact that $\mathbf{D}$ is {\rm CWC}, there exist $p_1,\ldots,p_h\in \mathbf{D}$ pairwise disjoint clones of $p'$. Since $\mathbf{D}$ is {\rm CA}, we have that $p''\coloneq p_1+\ldots+p_h\in \mathbf{D}$ and $N(p'')=N(t)\in N(\mathbf{D})$, which completes the proof. 
\end{proof}

\section{Social choice correspondences} 

Let $\mathbf{D}\subseteq \mathbf{R}(A)^*$. A social choice correspondence ({\sc scc}) on $\mathbf{D}$ is a function from $\mathbf{D}$ to $P_*(A)$. Thus, a {\sc scc} on $\mathbf{D}$ can be interpreted as a method for selecting a nonempty set of alternatives when the voters and their preferences generate a preference profile belonging to $\mathbf{D}$.

The following three {\sc scc}s are crucial for our purposes.
The net-outdegree {\sc scc}, denoted by $O$, is the {\sc scc} on $\mathbf{R}(A)^*$ 
that associates with every $p\in\mathbf{R}(A)^*$ the set
\[
O(p)\coloneq  \mathscr{O}(N(p));
\]
the net-indegree {\sc scc}, denoted by $I$, is the {\sc scc} on $\mathbf{R}(A)^*$ 
that associates with every $p\in\mathbf{R}(A)^*$ the set
\[
I(p)\coloneq  \mathscr{I}(N(p));
\]
the total {\sc scc}, denoted by $T$, is the {\sc scc} on $\mathbf{R}(A)^*$ 
that associates with any $p\in\mathbf{R}(A)^*$ the set $A$. 

In the next proposition we show that the aforementioned {\sc scc}s coincide on those preference profiles whose associated network is pseudo-symmetric, while they have different outcomes on the other preference profiles. The proof is an immediate consequence of 
Proposition \ref{flat-o} and then omitted.

\begin{proposition}\label{flat-o-prof} Let $p\in\mathbf{R}(A)^*$. Then the following facts hold true.
\begin{itemize}
\item[$(i)$] If $N(p)\in \mathcal{PS}(A)$, then $O(p)=I(p)=A$.
\item[$(ii)$] If $N(p)\not\in \mathcal{PS}(A)$, then $O(p)\cap I(p)=\varnothing$ and $|\{O(p),I(p),A\}|=3$.
\end{itemize}
\end{proposition}

Let us introduce now some properties that can be used to assess and compare {\sc scc}s.

\begin{definition}\label{defs-scc}
Let $\mathbf{D}\subseteq \mathbf{R}(A)^*$ and $F$ be a {\sc scc} on $\mathbf{D}$. 
\begin{itemize}
	\item $F$ satisfies {\rm neutrality} if $\mathbf{D}$ is {\rm CPA} and, for every $p\in \mathbf{D}$ and $\psi\in\mathrm{Sym}(A)$, $F(p^\psi)=\psi(F(p))$;
	\item $F$ satisfies  {\rm consistency} if $\mathbf{D}$ is {\rm CA} and, for every $p,p'\in \mathbf{D}$ disjoint  with $F(p)\cap F(p')\neq\varnothing$, $F(p+p')=F(p)\cap F(p')$;
	\item $F$ satisfies {\rm cancellation} if, for every $p\in \mathbf{D}$ such that $N(p)\in\mathcal{R}(A)$, $F(p)=A$;
\item $F$ satisfies {\rm anonymity} if, for every $p,p'\in \mathbf{D}$ with $\mathrm{Dom}(p)=\mathrm{Dom}(p')$ and $p'$ clone of $p$, we have $F(p)=F(p').$
\item $F$  satisfies {\rm Fishburn-anonymity} if, for every $p,p'\in \mathbf{D}$, $N(p)=N(p')$ implies $F(p)=F(p')$.
\end{itemize}
\end{definition}
The interpretation of those properties is natural. Neutrality means that all the alternatives are equally treated or, in other words, that no group of alternatives has an exogenous advantage over the others; consistency requires that if two disjoint sets of voters select two sets of alternative that are not disjoint, then the common alternatives must be the ones selected by the voters of the two sets, provided that the individual preferences on the alternatives are kept unchanged; cancellation means that all the alternatives must be selected when, for every pair of alternatives, the number of individuals preferring the first alternative to the second one equals the number of individuals preferring the second alternative to the first one; anonymity means that all the voters involved in the voting procedure equally contribute to the choice; Fishburn-anonymity means that the only information used to select the alternatives is the family of numbers obtained by counting, for every pair of alternatives, how many individuals prefer the first alternative to the second one and the second alternative to the first one.  Recall that, when  $\mathbf{D}=\mathbf{L}(A)^*$,  Fishburn-anonymity corresponds to the property C2 in Fishburn (1977). Cleary, Fishburn-anonymity  implies anonymity. Let us also stress that stating that a {\sc scc} satisfies neutrality and consistency implies to give an information on the properties of the domain of the {\sc scc}.

The next proposition shows that any restriction of the net-outdegree {\sc scc}, the net-indegree {\sc scc} and the total {\sc scc} satisfies the aforementioned properties, provided the chosen domain satisfies suitable conditions.

\begin{proposition}\label{Prop-OIT}
Let $\mathbf{D}\subseteq \mathbf{R}(A)^*$. Then the following facts hold true.
\begin{itemize}
\item[$(i)$] If $\mathbf{D}$ is {\rm CPA}, then $O_{|\mathbf{D}}$, $I_{|\mathbf{D}}$ and $T_{|\mathbf{D}}$ satisfy neutrality.
\item[$(ii)$] If $\mathbf{D}$ is {\rm CA} and {\rm CWC}, then $O_{|\mathbf{D}}$, $I_{|\mathbf{D}}$ and $T_{|\mathbf{D}}$ satisfy consistency.
\item[$(iii)$] $O_{|\mathbf{D}}$, $I_{|\mathbf{D}}$ and $T_{|\mathbf{D}}$ satisfy cancellation.
\item[$(iv)$] $O_{|\mathbf{D}}$, $I_{|\mathbf{D}}$ and $T_{|\mathbf{D}}$ satisfy Fishburn-anonymity.
\end{itemize}
\end{proposition}

\begin{proof}
We prove $(i)$, $(ii)$, $(iii)$ and $(iv)$ for $O_{\mathbf{D}}$ only. The proofs for $I_{|\mathbf{D}}$ are analogous, while the ones for $T_{|\mathbf{D}}$ are  straightforward.

$(i)$ Assume that $\mathbf{D}$ is {\rm CPA}. Let us prove first that $O_{|\mathbf{D}}$ satisfies neutrality. Let $p\in \mathbf{D}$ and $\psi\in\mathrm{Sym}(A)$. We have to prove that $O(p^\psi)=\psi(O(p))$. By Propositions \ref{nuovo-lemma}$(i)$ and \ref{Osoddisfa}$(i)$, we know that $N(\mathbf{D})$ is CPV and that $\mathscr{O}_{|N(\mathbf{D})}$ satisfies neutrality. Thus, we have that
\[
O(p^\psi)=\mathscr{O}(N(p^\psi))=\mathscr{O}(N(p)^\psi)=\psi(\mathscr{O}(N(p)))=\psi(O(p)).
\]

$(ii)$ Assume that $\mathbf{D}$ is {\rm CA} , and {\rm CWC}. Let us prove first that $O_{|\mathbf{D}}$ satisfies consistency. Let $p,p'\in \mathbf{D}$ disjoint with $O(p)\cap O(p')\neq \varnothing$. We have to prove that $O(p+p')=O(p)\cap O(p')$. By Propositions \ref{nuovo-lemma}$(ii)$ and \ref{Osoddisfa}$(ii)$, we know that $N(\mathbf{D})$ is CA and that $\mathscr{O}_{|N(\mathbf{D})}$ satisfies consistency. Recalling also that $\mathscr{O}(N(p))\cap \mathscr{O}(N(p'))=O(p)\cap O(p')\neq \varnothing$, we have that
\[
O(p+p')=\mathscr{O}(N(p+p'))=\mathscr{O}(N(p)+N(p'))=\mathscr{O}(N(p))\cap\mathscr{O}(N(p'))=O(p)\cap O(p').
\]

$(iii)$ Let us first show that $O_{|\mathbf{D}}$ satisfies cancellation. Let $p\in \mathbf{D}$ be such that $N(p)\in\mathcal{R}(A)$. We have to prove that $O(p)=A$. By Proposition \ref{Osoddisfa}$(iii)$, we know that $\mathscr{O}_{|N(\mathbf{D})}$ satisfies cancellation. Thus, we have that $O(p)=\mathscr{O}(N(p))=A$.

$(iv)$ Let us first that $O_{|\mathbf{D}}$ satisfies Fishburn-anonymity. Let $p,p'\in \mathbf{D}$ be such that $N(p)=N(p')$. Then we immediately have that $O(p)=\mathscr{O}(N(p))=\mathscr{O}(N(p'))=O(p')$.
\end{proof}

We present now an auxiliary result for {\sc scc}s satisfying consistency and cancellation inspired by Lemma 1 in Young (1974).
It will play a crucial role in proving the main characterization result for {\sc scc}, namely Theorem \ref{characterization-profiles}.

\begin{proposition}\label{passaggio-reti} Let $\mathbf{D}\subseteq \mathbf{R}(A)^*$ be {\rm CA}, {\rm CWC} and {\rm CRS}, and let $F$ be a {\sc scc} on $\mathbf{D}$ satisfying consistency and cancellation. Then $F$ satisfies Fishburn-anonymity. In particular, $F$ is  anonymous.
\end{proposition}

\begin{proof} 
Let us prove first that if $p,p'\in \mathbf{D}$ are disjoint clones, then $F(p)=F(p')$.
Let $p,p'\in \mathbf{D}$ be disjoint clones. Since $\mathbf{D}$ is CRS, there exists $q\in \mathbf{D}$ such that $q$ is disjoint from $p$ and $N(p+q)\in\mathcal{R}(A)$.  Since $\mathbf{D}$ is CWC, there exists $s\in \mathbf{D}$ such that $s$ is a clone of $q$ disjoint from $p$ and $p'$. We have then that $p+q$, $p+s$ and $p'+s$ are clones and hence $N(p+q)=N(p+s)=N(p'+s)\in\mathcal{R}(A)$. By cancellation of $F$, we have that $F(p+s)=A$ and $F(p'+s)=A$.
Using now consistency of $F$,  we get
$F(p)=F(p)\cap A=F(p)\cap F(p'+s)=F(p+(p'+s))=F(p'+(p+s))=F(p')\cap F(p+s)=F(p')\cap A =F(p').$

We prove now that $F$ satisfies Fishburn anonymity. Let $p,p'\in \mathbf{D}$ be such that $N(p)=N(p')$. 
Since $\mathbf{D}$ satisfies CWC, there exist $\hat{p},\hat{p}'\in\mathbf{D}$ disjoint such that $\hat{p}$ is a clone of $p$,  $\hat{p}'$ is a clone of $p'$ and $\hat{p}$ and $\hat{p}'$ are both disjoint from $p$ and $p'$.
Since $\mathbf{D}$ satisfies CRS, there exist $q\in\mathbf{D}$ such that $q$ is disjoint 
from $p$  and $N(p+q)\in\mathcal{R}(A)$.
Since $\mathbf{D}$ satisfies CWC, there exist $\hat{q}\in\mathbf{D}$ such that $\hat{q}$ is a clone of $q$ and $\hat{q}$ is disjoint from 
$p$, $q$, $\hat{p}$ and $\hat{p}'$.
By $N(p)=N(p')$, properties \eqref{prop-reti-profili} and the properties of clones, we get 
\[
N(\hat{p}+\hat{q})=N(\hat{p})+N(\hat{q})=N(p)+N(q)=N(p+q)\in \mathcal{R}(A),
\]
\[
N(\hat{p}'+\hat{q})=N(\hat{p}')+N(\hat{q})=N(p')+N(q)=N(p)+N(q)=N(p+q)\in \mathcal{R}(A).
\]
Thus, by cancellation of $F$, we obtain that $F(\hat{p}+\hat{q})=A$ and $F(\hat{p}'+\hat{q})=A$.
Applying now consistency of $F$, we obtain 
\[
F(\hat{p})=F(\hat{p})\cap A=F(\hat{p})\cap F(\hat{p}'+\hat{q})=F(\hat{p}+(\hat{p}'+\hat{q}))
\]
\[
=
F(\hat{p}'+(\hat{p}+\hat{q}))=F(\hat{p}')\cap F(\hat{p}+\hat{q})=F(\hat{p}')\cap A=F(\hat{p}').
\]
Since $\hat{p}$ is a disjoint clone of $p$ and $\hat{p}'$ is a disjoint clone of $p'$, by the first part of the proof, we finally obtain 
$F(p)=F(\hat{p})=F(\hat{p}')=F(p')$.
\end{proof}

\section{General characterization theorem for {\sc scc}s}

Let us introduce now a crucial property for sets of preference profiles. That property, called regularity, is strongly inspired to the concept of regularity for sets of networks presented in Definition \ref{regular-net}.

\begin{definition}\label{regular}
Let $\mathbf{D}\subseteq \mathbf{R}(A)^*$. We say that $\mathbf{D}$ is {\rm regular} if
\begin{itemize}
	\item[$(\alpha)$] $\mathbf{D}$ is {\rm CPA}, {\rm CA} and {\rm CWC};
	\item[$(\beta)$] for every $N\in \mathbb{Z} N(\mathbf{D})$, there exist $p,p'\in \mathbf{D}$ such that
$N(p)=N+N(p')$ with $N(p')\in \mathcal{R}(A)$;
	\item[$(\gamma)$] $\big(\mathbb{Q}N(\mathbf{D})+\mathcal{R}(A)\big)\cap \mathcal{PS}(A)\in\{\mathcal{R}(A),\mathcal{PS}(A)\}$;
	\item[$(\epsilon)$] $\big(\mathbb{Q}N(\mathbf{D})+\mathcal{R}(A)\big)\cap \{N_x:x\in A\}\neq\varnothing$.
\end{itemize}
\end{definition}
Regular sets of preference profiles are one of the main ingredients of Theorem \ref{characterization-profiles}, the main characterization result for social choice correspondences. In Section \ref{regset}, we will exhibit several remarkable examples of regular subsets of $\mathbf{R}(A)^*$. 

The next proposition helps apply the theory developed for networks to the framework of social choice correspondences. It also raises an interesting point in reference to the second part of Theorem \ref{characterization-profiles} by explaining that any regular set of preference profiles contains an element whose associated network is not pseudo-symmetric.

\begin{proposition}\label{reg-scc-ns}
Let $\mathbf{D}\subseteq \mathbf{R}(A)^*$ be regular. Then $N(\mathbf{D})$ is regular and $N(\mathbf{D})\not\subseteq \mathcal{PS}(A)$.
\end{proposition}

\begin{proof}
In order to prove that $N(\mathbf{D})$ is regular, we have to show that conditions $(a)$, $(b)$, $(c)$ and $(d)$
in Definition \ref{regular-net} are satisfied by $N(\mathbf{D})$. First note that, since $\mathbf{D}$ is CPA, CA and CWC, by Proposition \ref{nuovo-lemma},  
we have that $N(\mathbf{D})$ is CVP and CA. Thus, $N(\mathbf{D})$ satisfies condition  $(a)$. 
Since $\mathbf{D}$ satisfies conditions $(\beta)$, $(\gamma)$ and $(\epsilon)$, we immediately also have that $N(\mathbf{D})$ satisfies conditions $(b)$, $(c)$ and $(d)$. The fact that $N(\mathbf{D})\not\subseteq \mathcal{PS}(A)$ follows then from Proposition \ref{no-pseudo}.
\end{proof}

By Proposition \ref{Prop-OIT}, if $\mathbf{D}$ is CPA, CA and CWC, we know that $O_{|\mathbf{D}}$, $I_{|\mathbf{D}}$ and $T_{|\mathbf{D}}$ satisfy neutrality, consistency and cancellation. The next result shows that the regularity of $\mathbf{D}$ guarantees that $O_{|\mathbf{D}}$, $I_{|\mathbf{D}}$ and $T_{|\mathbf{D}}$ are the unique {\sc scc}s  on $\mathbf{D}$ fulfilling those properties.

\begin{theorem}\label{characterization-profiles}
Let $\mathbf{D}\subseteq \mathbf{R}(A)^*$ be regular and $F$ be a {\sc scc} on $\mathbf{D}$ satisfying neutrality, consistency and cancellation. Then 
$F\in \{O_{|\mathbf{D}}, I_{|\mathbf{D}}, T_{|\mathbf{D}}\}$. Moreover, for every $p\in \mathbf{D}$ such that $N(p)\not\in\mathcal{PS}(A)$, we have that
\begin{itemize}
\item[$(i)$] if $F(p)=O(p)$, then $F=O_{|\mathbf{D}}$;
\item[$(ii)$] if $F(p)=I(p)$, then $F=I_{|\mathbf{D}}$;
\item[$(iii)$]  if $F(p)=A$, then $F=T_{|\mathbf{D}}$.
\end{itemize}
\end{theorem}

\begin{proof} 
Let $\mathcal{D}\coloneq N(\mathbf{D})\subseteq \mathcal{N}(A)$.
By Proposition \ref{reg-scc-ns}, we know that $\mathcal{D}$ is regular.
By Proposition \ref{nuovo-lemma}$(iii)$, we also know that $\mathbf{D}$ is {\rm CA}, {\rm CWC} and {\rm CRS}. 

Consider then the network solution $\mathscr{F}:\mathcal{D}\to P_*(A)$, defined, for every $N\in \mathcal{D}$, by $\mathscr{F}(N)=F(p)$, where $p$ is any element in $\mathbf{D}$ such that $N(p)=N$. We show that $\mathscr{F}$ is well defined. Indeed, since $\mathbf{D}$ is CA, CWC and CRS and $F$ satisfies consistency and cancellation, by Proposition \ref{passaggio-reti}, we have that $F$ is Fishburn-anonymous. Thus, if $N\in \mathcal{D}$ and $p,p'\in \mathbf{D}$ are such that $N(p)=N(p')=N$, then we have $F(p)=F(p')$. Note in particular that, for every $p\in\mathbf{D}$, $F(p)=\mathscr{F}(N(p))$.

We show now that $\mathscr{F}$ satisfies neutrality. Let $N\in \mathcal{D}$ and $\psi\in\mathrm{Sym}(A)$. Consider $p\in \mathbf{D}$ such that $N(p)=N$. Following the definition of $\mathscr{F}$, using formulas \eqref{prop-reti-profili} and recalling that $F$ satisfies neutrality, we have 
\[
\mathscr{F}(N^{\psi})=\mathscr{F}(N(p)^{\psi})=\mathscr{F}(N(p^{\psi}))=F(p^{\psi})=\psi F(p)=\psi \mathscr{F}(N(p))=\psi \mathscr{F}(N).
\]
We also have that $\mathscr{F}$ satisfies cancellation. Indeed, let $N\in \mathcal{D}$ such that $N\in\mathcal{R}(A)$. Consider then $p\in \mathbf{D}$ such that $N(p)=N$. Since $F$ satisfies cancellation, we get $\mathscr{F}(N)=F(p)=A$.
Let us check now that $\mathscr{F}$ satisfies consistency too. Let $N,M\in \mathcal{D}$ be such that $\mathscr{F}(N)\cap \mathscr{F}(M)\neq \varnothing$. Consider $p,p'\in \mathbf{D}$ such that $N=N(p)$ and $M=N(p')$. Since $\mathbf{D}$ is CWC there are $\hat{p},\hat{p}'\in \mathbf{D}$ disjoint such that $\hat{p}$ is a clone of $p$ and $\hat{p}'$ is a clone of $p'$. Thus, $N=N(\hat{p})$ and $M=N(\hat{p}')$. Then
\[
F(\hat{p})\cap F(\hat{p}')=\mathscr{F}(N(\hat{p}))\cap \mathscr{F}(N(\hat{p}')) =\mathscr{F}(N)\cap \mathscr{F}(M)\neq\varnothing,
\]
and hence, using consistency of $F$,  we get 
\[
\mathscr{F}(N+M)=\mathscr{F}(N(\hat{p})+N(\hat{p}'))=\mathscr{F}(N(\hat{p}+\hat{p}'))=F(\hat{p}+\hat{p}')=F(\hat{p})\cap F(\hat{p}')=\mathscr{F}(N)\cap \mathscr{F}(M).
\]
By Theorem \ref{characterization-net}, we deduce that $\mathscr{F}\in\{\mathscr{O}_{|\mathcal{D}},\mathscr{I}_{|\mathcal{D}},\mathscr{T}_{|\mathcal{D}}\}$. 
Assume now that $\mathscr{F}=\mathscr{O}_{|\mathcal{D}}$. We have that, for every $p\in\mathbf{D}$, $F(p)=\mathscr{F}(N(p))=\mathscr{O}(N(p))=O(p)$. Thus, $F=O_{|\mathbf{D}}$. An analogous argument proves that if  $\mathscr{F}=\mathscr{I}_{|\mathcal{D}}$, then $F=I_{|\mathbf{D}}$ and that if 
$\mathscr{F}=\mathscr{T}_{|\mathcal{D}}$, then $F=T_{|\mathbf{D}}$. We conclude then that $F\in \{O_{|\mathbf{D}}, I_{|\mathbf{D}}, T_{|\mathbf{D}}\}$.

By Proposition \ref{reg-scc-ns}, there exists $p\in \mathbf{D}$ such that $N(p)\not\in\mathcal{PS}(A)$. By Proposition \ref{flat-o-prof}, we have that $|\{O(p),I(p),A\}|=3$. As a consequence, 
if $F(p)=O(p)$, then we have $F(p)\neq I(p)$ and $F(p)\neq T(p)$ and hence $F=O_{|\mathbf{D}}$; 
if $F(p)=I(p)$, then we have $F(p)\neq O(p)$ and $F(p)\neq T(p)$ and hence $F=I_{|\mathbf{D}}$;
if $F(p)=T(p)$, then we have $F(p)\neq O(p)$ and $F(p)\neq I(p)$ and hence $F=T_{|\mathbf{D}}$.
That proves $(i)$, $(ii)$ and $(iii)$.
\end{proof}

\subsection{Some regular subsets of $\mathbf{R}(A)^*$}\label{regset}

Because of Theorem \ref{characterization-profiles}, we can clearly understand the importance of determining families of regular subsets of $\mathbf{R}(A)^*$. 
Indeed, each regular set of networks leads to a specific characterization theorem. Theorem \ref{reg-scc} at the end of this section shows that many interesting subsets of $\mathbf{R}(A)^*$ are actually regular. In order to present that result we need to introduce some notable classes of preference relations.

Let $R\in\mathbf{R}(A)$. We say that $R$ is reflexive if $A^2_d\subseteq R$; complete if, for every $x,y\in A$, $(x,y)\in R$ or $(y,x)\in R$; transitive if, for every $x,y,z\in A$, $(x,y)\in R$ and $(y,z)\in R$ implies $(x,z)\in R$; antisymmetric if, for every $x,y\in A$, $(x,y)\in R$ and $(y,x)\in R$ implies $x=y$. Note that a complete relation is necessarily reflexive. 

We say that $R$ is an order if $R$ is complete and transitive; a linear order if $R$ is complete, transitive and antisymmetric; a partial order if $R$ is reflexive, antisymmetric and transitive. The set of orders on $A$ is denoted by $\mathbf{O}(A)$; the set of linear orders on $A$ is denoted by $\mathbf{L}(A)$; the set of partial orders on $A$ is denoted by $\mathbf{P}(A)$. Of course, we have $ \mathbf{L}(A)\neq \varnothing$, $\mathbf{L}(A)\subseteq \mathbf{O}(A)$ and $ \mathbf{L}(A)\subseteq \mathbf{P}(A)$. 
Note that, if $R\in\mathbf{O}(A)$, then $IN(R,x)=\varnothing$; if $R\in\mathbf{P}(A)$, then $I(R,x)=\{x\}$ (see \eqref{LU}).

Let $R\in\mathbf{O}(A)$. Then we have that  
\begin{equation}\label{maxorder}
\mathrm{Max}(R)= \{x\in A: \forall y\in A,\; x\succeq_R y\}, \quad \mathrm{Min}(R)=\{x\in A: \forall y\in A,\; y\succeq_R x\};
\end{equation}
$\mathrm{Max}(R)$ and $\mathrm{Min}(R)$ are nonempty; for every $x\in \mathrm{Max}(R)$, $\mathrm{Max}(R)=I(R,x)$; for every $x\in \mathrm{Min}(R)$, $\mathrm{Min}(R)=I(R,x)$; $\mathrm{Max}(R)\cap \mathrm{Min}(R)\neq\varnothing$ is equivalent to $\mathrm{Max}(R)=\mathrm{Min}(R)=A$ that, in turn, is equivalent to $R=A^2$; for every $x\in \mathrm{Max}(R)$ and $y\in A\setminus \mathrm{Max}(R)$, $x\succ_R y$; for every $x\in \mathrm{Min}(R)$ and $y\in A\setminus \mathrm{Min}(R)$, $y\succ_R x$; if $R\in\mathbf{L}(A)$, then $|\mathrm{Max}(R)|=|\mathrm{Min}(R)|=1$. We say that $R$ is a dichotomous order if  $\mathrm{Max}(R)\cup \mathrm{Min}(R)=A$;
a top-truncated order if, for every $x\in A\setminus \mathrm{Min}(R)$, $|I(R,x)|=1$. The set of dichotomous orders on $A$ is denoted by $\mathbf{Di}(A)$; the set of top-truncated orders on $A$ is denoted by $\mathbf{T}(A)$. Note that $\mathbf{L}(A)\subseteq \mathbf{T}(A)$. In particular, $ \mathbf{T}(A)\neq \varnothing.$

For every $t\in [m]$, we set $\mathbf{Di}_t(A)\coloneq \{R\in \mathbf{Di}(A): |\mathrm{Max}(R)|=t\}$ and,  given $X\subseteq [m]$, we set 
\[
\mathbf{Di}_X(A)\coloneq \bigcup_{t\in X} \mathbf{Di}_t(A).
\]
For every $s\in [m-1]_0$, we set $\mathbf{T}_s(A)\coloneq \{R\in \mathbf{T}(A): |A\setminus \mathrm{Min}(R)|=s\}$ and given $Y\subseteq [m-1]_0$, we set 
\[
\mathbf{T}_Y(A)\coloneq \bigcup_{s\in Y} \mathbf{T}_s(A).
\]
Note that, for every $t\in [m]$, $\mathbf{Di}_t(A)=\mathbf{Di}_{\{t\}}(A) \neq\varnothing$ and, for every $s\in [m-1]_0$,  $\mathbf{T}_s(A)=\mathbf{T}_{\{s\}}(A) \neq\varnothing$. Moreover, $\mathbf{Di}_m(A)=\mathbf{T}_0(A)=\{A^2\}$,  $\mathbf{Di}_1(A)= \mathbf{T}_1(A)$, $\mathbf{T}_{m-1}(A)=\mathbf{L}(A)$.

Let $\mathbf{D}(A)\subseteq \mathbf{R}(A)$. We say that $\mathbf{D}(A)$ is closed under permutation of alternatives if, for every $R\in \mathbf{D}(A)$ and $\psi\in\mathrm{Sym}(A)$, $R^\psi\in \mathbf{D}(A)$. It is easily checked that the sets $\mathbf{R}(A)$, $\mathbf{P}(A)$, $\mathbf{O}(A)$, $\mathbf{L}(A)$, $\mathbf{Di}_X(A)$ and $\mathbf{T}_Y(A)$ are closed under permutation of alternatives.

 We set 
\[
\mathbf{D}(A)^*\coloneq \bigcup_{\substack{
I\subseteq V\\
 I\neq\varnothing\, \mathrm{finite}
}}
 \mathbf{D}(A)^I.
\]
Note that $\mathbf{D}(A)^*\subseteq \mathbf{R}(A)^*$ and that  $N(\mathbf{D}(A)^*)=\{N(0)\}$ if and only if $\mathbf{D}(A)\subseteq P(A^2_d)$.

The next theorem, whose proof is in the appendix, shows that the set $\mathbf{D}(A)^*$ is regular for many remarkable qualifications of $\mathbf{D}(A)$ typically used in social choice theory. 

\begin{theorem}\label{reg-scc}
The following facts hold true.
\begin{itemize}
\item[$(i)$] Let $ \mathbf{D}(A)\subseteq \mathbf{R}(A)$ be closed under permutation of alternatives and such that $\mathbf{L}(A)\subseteq \mathbf{D}(A)$. Then $\mathbf{D}(A)^*$ is regular. In particular, $\mathbf{R}(A)^*$, $\mathbf{P}(A)^*$, $\mathbf{O}(A)^*$, $\mathbf{T}(A)^*$ and $\mathbf{L}(A)^*$ are regular.
\item[$(ii)$]Let $X\subseteq [m]$ with $X\cap [m-1]\neq\varnothing$. Then
$\mathbf{Di}_X(A)^*$ is regular. In particular, $\mathbf{Di}(A)^*$ is regular and, for every $t\in[m-1]$, $\mathbf{Di}_t(A)^*$ and $\mathbf{Di}_{[t]}(A)^*$ are regular.
\item[$(iii)$]Let $Y\subseteq [m-1]_0$ with $Y\cap [m-1]\neq\varnothing$. Then
$\mathbf{T}_Y(A)^*$ is regular. In particular,
$\mathbf{T}(A)^*$ is regular and, for every $s\in[m-1]$, $\mathbf{T}_s(A)^*$ and  $\mathbf{T}_{[s]}(A)^*$ are regular.
\end{itemize}
\end{theorem}

\section{New and old characterization results for classic {\sc scc}s}\label{ultima}

In this section we exploit the power and the generality of our approach to prove within a unified framework new and known characterization results for some classic {\sc scc}s. In order to get the goal, we need first to recognize that the net-oudegree {\sc scc} can be seen as the common root of many well-known {\sc scc}s (Section \ref{Ouguale}). Then we need to introduce and discuss some further properties for {\sc scc}s (Section \ref{faoc}).

\subsection{Classic {\sc scc}s and the net-outdegree {\sc scc}}\label{Ouguale}

In this section we show that many well-know and classic {\sc scc}s, namely the Borda rule, the Partial Borda rule, the Averaged Borda rule, the Approval Voting, the Plurality rule and the anti-Plurality rule, correspond to suitable restrictions of the net-outdegree {\sc scc}. Let us begin by recalling the definition of the aforementioned rules.

Let us recall first the definition of the Borda rule on $\mathbf{O}(A)^*$ (Mas-Colell et al., 1997, Example 21.C.1; Vorsatz, 2008), a definition that extends to 
$\mathbf{O}(A)^*$ the classic version of the rule defined on $\mathbf{L}(A)^*$. For every $R\in\mathbf{O}(A)$ and $x\in A$, the Borda score of $x$ in $R$ is defined by
\[
b(R,x)\coloneq \frac{1}{|I(R,x)|}\sum_{i=0}^{|I(R,x)|-1} (|L(R,x)|+i)= |L(R,x)|+ \frac{|I(R,x)|}{2}-\frac{1}{2}.
\]
For every $p\in \mathbf{O}(A)^*$ and $x\in A$, the Borda score of $x$ in $p$ is defined by
\[
b(p,x)\coloneq \sum_{i\in \mathrm{Dom}(p)}b(p(i),x).
\]
The Borda rule is then the {\sc scc} defined, for every $p\in \mathbf{O}(A)^*$, by 
\[
BOR(p)\coloneq \underset{x\in A}{\mathrm{argmax}}\; b(p,x).
\] 

The Partial Borda rule is a {\sc scc} on $\mathbf{P}(A)^*$ introduced by Cullinan et al. (2014) and defined as follows.
Given $R\in\mathbf{P}(A)$, the partial Borda score of $x$ in $R$ is defined by
\[
pb(R,x)\coloneq 2|L(R,x)|+ |IN(R,x)|.
\]
For every $p\in \mathbf{P}(A)^*$ and $x\in A$, the partial Borda score of $x$ in $p$ is defined by
\[
pb(p,x)\coloneq \sum_{i\in \mathrm{Dom}(p)}pb(p(i),x).
\]
The Partial Borda rule is then the {\sc scc} defined, for every $p\in \mathbf{P}(A)^*$, by 
\[
PBOR(p)\coloneq \underset{x\in A}{\mathrm{argmax}}\; pb(p,x).
\]

The Averaged Borda rule is a {\sc scc} on $\mathbf{T}(A)^*$ introduced by Dummett (1997). Such a {\sc scc}, here denoted by $ABOR$, coincides with $BOR_{|\mathbf{T}(A)^*}$.

The Approval Voting is a {\sc scc } on $\mathbf{Di}(A)^*$ defined by Brams and Fishburn (1978) as follows.
For every $p\in \mathbf{Di}(A)^*$ and $x\in A$, the approval score of $x$ in $p$ is defined as
\begin{equation}\label{appr-score}
av(p,x)\coloneq |\{i\in \mathrm{Dom}(p): x\in \mathrm{Max}(p(i))\}|.
\end{equation}
The Approval Voting is then the {\sc scc} defined, for every $p\in \mathbf{Di}(A)^*$, by 
\[
AV(p)\coloneq \underset{x\in A}{\mathrm{argmax}}\; av(p,x).
\]
Finally, the Plurality rule, here denoted by $PLU$, is the {\sc scc} on $\mathbf{Di}_1(A)^*$ given by to $AV_{|\mathbf{Di}_1(A)^*}$ and 
the anti-Plurality rule, here denoted by $APLU$, is the {\sc scc} on $\mathbf{Di}_{m-1}(A)^*$ given by to $AV_{|\mathbf{Di}_{m-1}(A)^*}$.

In the rest of the section we prove that each of above considered the rules always selects the same alternatives as the net-outdegree {\sc scc}.
In what follows, for every $p\in \mathbf{R}(A)^*$ and $x\in A$, we set
\begin{equation}\label{nuova-O}
o(p,x)\coloneq \sum_{i\in \mathrm{Dom}(p)}|L(p(i),x)|-\sum_{i\in \mathrm{Dom}(p)}|U(p(i),x)|.
\end{equation}
The next proposition explains the strong relation between the quantity $o(p,x)$ and the net-outdegree {\sc scc} and constitutes a very useful preliminary result.

\begin{proposition}\label{O-interpret}
Let $p\in \mathbf{R}(A)^*$ and $x\in A$. Then $\delta^{N(p)}(x)=o(p,x)$. In particular, 
\[
O(p)=\underset{x\in A}{\mathrm{argmax}}\; o(p,x).
\]
\end{proposition}

\begin{proof}
We have that
\[
\delta^{N(p)}(x)=\sum_{y\in A\setminus \{x\}}c_p(x,y)-\sum_{y\in A\setminus \{x\}}c_p(y,x)
=\sum_{i\in \mathrm{Dom}(p)}\sum_{y\in A\setminus \{x\}}\left(c_{p(i)}(x,y)-c_{p(i)}(y,x)\right).
\]
Clearly, for every $i\in  \mathrm{Dom}(p)$,  the alternatives in $I(p(i),x)\cup IN(p(i),x)$ contribute $0$ to the above sum. It follows that 
\[
\delta^{N(p)}(x)=\sum_{i\in \mathrm{Dom}(p)}\left(\sum_{y\in L(p(i),x)}1-\sum_{y\in U(p(i),x)}1\right)=\sum_{i\in \mathrm{Dom}(p)}\Big(|L(p(i),x)|-|U(p(i),x)|\Big)=o(p,x).
\]
\end{proof}

We are now ready to prove all the announced results.

\begin{proposition}\label{Borda-equal-O}
$BOR=O_{|\mathbf{O}(A)^*}$. In particular, $BOR_{|\mathbf{L}(A)^*}=O_{|\mathbf{L}(A)^*}$.
\end{proposition}

\begin{proof} Let $R\in\mathbf{O}(A)$ and $x\in A$. Since $IN(R,x)=\varnothing$, we have that $A$ is the disjoint union of $L(R,x), U(R,x), I(R,x)$. It follows that $m=|L(R,x)|+|U(R,x)|+|I(R,x)|.$
Thus, we have
\begin{equation*}\label{Borda1}
b(R,x)=|L(R,x)|+\frac{1}{2}(m-|L(R,x)|-|U(R,x)|)-\frac{1}{2}=
\frac{1}{2}\left(|L(R,x)|-|U(R,x)|+m-1 \right).
\end{equation*}
As a consequence, recalling \eqref{nuova-O}, for every $p\in \mathbf{O}(A)^*$ and $x\in A$ we get
\[
b(p,x)=\sum_{i\in \mathrm{Dom}(p)}\frac{1}{2}\left(|L(p(i),x)|-|U(p(i),x)|+m-1\right)=\frac{1}{2} o(p,x) + \frac{m-1}{2}|\mathrm{Dom}(p)|.
\] 
Thus, by Proposition \ref{O-interpret}, for every $p\in \mathbf{O}(A)^*,$ we have
\[
BOR(p)= \underset{x\in A}{\mathrm{argmax}}\;\left(\frac{1}{2} o(p,x) + \frac{m-1}{2}|\mathrm{Dom}(p)|\right)=\underset{x\in A}{\mathrm{argmax}}\; o(p,x)=O(p).
\]
\end{proof}

\begin{proposition}\label{PBorda-equal-O}
$PBOR=O_{|\mathbf{P}(A)^*}$. 
\end{proposition}

\begin{proof}
We first claim that, for every $R\in\mathbf{P}(A)$ and $x\in A$, we have 
\begin{equation}\label{PBorda1}
pb(R,x)=|L(R,x)|-|U(R,x)|+m-1.
\end{equation}
Indeed, recalling that $|I(R,x)|=1$, we have
\[
pb(R,x)=2|L(R,x)|+|IN(R,x)|=|L(R,x)|+|L(R,x)|+|IN(R,x)|
\]
\[
=|L(R,x)|+(m-U(R,x)-I(R,x))=|L(R,x)|-|U(R,x)|+m-1.
\]
Consider now $p\in \mathbf{P}(A)^*$ and $x\in A$. By \eqref{PBorda1}, we get
\[
pb(p,x)=\sum_{i\in \mathrm{Dom}(p)}\left(|L(p(i),x)|-|U(p(i),x)|+m-1\right)=o(p,x) + (m-1)|\mathrm{Dom}(p)|.
\] 
Thus, by Proposition \ref{O-interpret}, for every $p\in \mathbf{P}(A)^*$, we have 
\[
PBOR(p)= \underset{x\in A}{\mathrm{argmax}}(o(p,x) + (m-1)|\mathrm{Dom}(p)|)=\underset{x\in A}{\mathrm{argmax}}\; o(p,x)=O(p).
\]
\end{proof}

\begin{proposition}\label{ABorda-equal-O}
$ABOR=O_{|\mathbf{T}(A)^*}$. 
\end{proposition}

\begin{proof}
It immediately follows from the fact that $ABOR=BOR_{|\mathbf{T}(A)^*}$ and Proposition \ref{Borda-equal-O}.
\end{proof}

\begin{proposition}\label{AV-equal-O}
$AV=O_{|\mathbf{Di}(A)^*}$.
\end{proposition}

\begin{proof}
Let $p\in \mathbf{Di}(A)^*$ and $x\in A$. Then, by \eqref{nuova-O}, we have
\[
o(p,x)=\sum_{i\in \mathrm{Dom}(p)}(|L(p(i),x)|-|U(p(i),x)|)=\sum_{\substack{i\in \mathrm{Dom}(p),\\ x\in \mathrm{Max}(p(i))}}(m-| \mathrm{Max}(p(i))|)+\sum_{\substack {i\in \mathrm{Dom}(p),\\ x\notin  \mathrm{Max}(p(i))}}-| \mathrm{Max}(p(i))|
\]
\[
=m\cdot av(p,x)-\sum_{i\in \mathrm{Dom}(p)}| \mathrm{Max}(p(i))|.
\]
Thus, for every $p\in \mathbf{Di}(A)^*$, we have 
\[
O(p)= \underset{x\in A}{\mathrm{argmax}}\;o(p,x)=\underset{x\in A}{\mathrm{argmax}}\; av(p,x)=AV(p).
\]
\end{proof}

\begin{proposition}\label{PLU}
$PLU=O_{|\mathbf{Di}_1(A)^*}$.
\end{proposition}
\begin{proof}
Simply note that $\mathbf{Di}_1(A)^*\subseteq \mathbf{D}(A)^*$ and apply Proposition \ref{AV-equal-O}.
\end{proof}

\begin{proposition}\label{APLU}
$APLU=O_{|\mathbf{Di}_{m-1}(A)^*}$.
\end{proposition}
\begin{proof}
Simply note that $\mathbf{Di}_{m-1}(A)^*\subseteq \mathbf{D}(A)^*$ and apply Proposition \ref{AV-equal-O}.
\end{proof}

It is worth noting that, since $\mathbf{Di}(A)^*\subseteq \mathbf{O}(A)^*$, by Proposition \ref{Borda-equal-O} we get the known equality $AV=BOR_{|\mathbf{Di}(A)^*}$ (Vorsatz, 2008, Proposition 1). Moreover, since $\mathbf{Di}_1(A)^*\subseteq \mathbf{O}(A)^*$ and $\mathbf{Di}_{m-1}(A)^*\subseteq \mathbf{O}(A)^*$, by Propositions \ref{PLU} and \ref{APLU} we also deduce the equalities $PLU=BOR_{|\mathbf{Di}_1(A)^*}$ and $APLU=BOR_{|\mathbf{Di}_{m-1}(A)^*}$.

\subsection{Faithfulness, averseness, $O$-coherence and Fishburn-cancellation}\label{faoc}

In the next definition we extend the classic property of faithfulness introduced by Young (1974) and Fishburn (1979) for {\sc scc}s defined on $\mathbf{L}(A)^*$ and $\mathbf{Di}(A)^*$ to the case of {\sc scc}s defined on $\mathbf{R}(A)^*$.  Moreover, we also extend to the case of {\sc scc}s defined on $\mathbf{R}(A)^*$ a property, that we call averseness, considered by Cullinan et al. (2014) in the context of {\sc scc} defined on $\mathbf{P}(A)^*$. The term averseness also appears in Kurihara (2018) referred to {\sc scc} defined in $\mathbf{Di}_{m-1}(A)^*$. We finally introduce a new and weak property called $O$-coherence.

\begin{definition}\label{def-faith}
Let $\mathbf{D}\subseteq \mathbf{R}(A)^*$ and $F$ be a {\sc scc} on $\mathbf{D}$. 
\begin{itemize}
\item $F$ satisfies {\rm averseness} if, for every $p\in \mathbf{D}$ such that $\mathrm{Dom}(p)=\{i\}$ for some $i\in V$ and $\mathrm{Max}(p(i))\neq\varnothing$, we have $F(p)\subseteq \mathrm{Max}(p(i))$;
\item $F$ satisfies {\rm faithfulness} if, for every $p\in \mathbf{D}$ such that $\mathrm{Dom}(p)=\{i\}$ for some $i\in V$ and $\mathrm{Max}(p(i))\neq\varnothing$, we have $F(p)=\mathrm{Max}(p(i))$;
\item $F$ satisfies {\rm $O$-coherence} if there exists $p\in \mathbf{D}$ such that $N(p)\not\in \mathcal{PS}(A)$ and $F(p)=O(p)$.
\end{itemize}
\end{definition}

In the next proposition we propose conditions on $\mathbf{D}\subseteq \mathbf{R}(A)^*$ that guarantee that the net-outdegree {\sc scc} restricted to $\mathbf{D}$ satisfies the properties  in Definition \ref{def-faith}.

\begin{proposition}\label{newO}
Let $\mathbf{D}\subseteq \mathbf{R}(A)^*$. Then the following facts hold true.
\begin{itemize}
\item[$(i)$] If  $\mathbf{D}\subseteq \mathbf{P}(A)^*$, then $O_{|\mathbf{D}}$ satisfies averseness.
\item[$(ii)$]  If $\mathbf{D}\subseteq \mathbf{O}(A)^*$, then $O_{|\mathbf{D}}$ satisfies faithfulness. 
\item[$(iii)$] If $N(\mathbf{D})\not \subseteq \mathcal{PS}(A)$, then $O_{|\mathbf{D}}$ satisfies $O$-coherence. In particular, if $\mathbf{D}$ is regular, then 
 $O$ satisfies $O$-coherence.
\end{itemize}
\end{proposition}

\begin{proof}
$(i)$  Let $\mathbf{D}\subseteq \mathbf{P}(A)^*$ and $p\in \mathbf{D}$ with  $\mathrm{Dom}(p)=\{i\}$. Then $p(i)\in \mathbf{P}(A)$ and $\mathrm{Max}(p(i))\neq\varnothing$. We show that $O(p)\subseteq \mathrm{Max}(p(i))$.
Let $x^*\in O(p)=\underset{x\in A}{\mathrm{argmax}}\; o(p,x)$. Assume, by contradiction, that $x^*\notin\mathrm{Max}(p(i)).$ Then, there exists $y\in A$ such that $y\succ_{p(i)}x^*$. Then, by the transitivity of $p(i)$, we get $L(p(i), x^*)\subsetneq L(p(i),y)$. Observe that that inclusion is proper because $x^*\in L(p(i),y)\setminus L(p(i), x^*).$ Moreover, we have 
$U(p(i),y)\subseteq U(p(i), x^*)$. As a consequence, we deduce
\[
o(p, y)=|L(p(i), y)|-|U(p(i), y)|> |L(p(i),  x^*)|-|U(p(i),  x^*)|=o(p, x^*),
\]
a contradiction.

$(ii)$ Let $\mathbf{D}\subseteq \mathbf{O}(A)^*$ and $p\in \mathbf{D}$ with  $\mathrm{Dom}(p)=\{i\}$. Since $\mathbf{O}(A)^*\subseteq \mathbf{P}(A)^*$, by $(i)$, we have that $O(p)\subseteq \mathrm{Max}(p(i))$. Hence, it is enough to show that $\mathrm{Max}(p(i))\subseteq O(p)$. Since $p(i)\in \mathbf{O}(A)$, by \eqref{maxorder}, we have 
$\mathrm{Max}(p(i))= \{x\in A: \forall y\in A,\; x\succeq_{p(i)} y\}.$
Let $x^*\in \mathrm{Max}(p(i))$. Then, for every $y\in A$, $x^*\succeq_{p(i)} y$. We show that, for every $y\in A$, we have  $o(p, y)\leq o(p, x^*)$. Let $y\in A.$ Then, by the transitivity of $p(i)$, we have $L(p(i), y)\subseteq L(p(i),x^*)$ and $U(p(i),x^*)\subseteq U(p(i), y)$. As a consequence, 
\[
o(p, y)=|L(p(i), y)|-|U(p(i), y)|\leq |L(p(i),  x^*)|-|U(p(i),  x^*)|=o(p, x^*).
\]
We conclude then that $x^*\in O(p)=\underset{x\in A}{\mathrm{argmax}}\; o(p,x)$.

$(iii)$ It immediately follows from the definition of $O$-coherence and Proposition \ref{reg-scc-ns}.
\end{proof}

We observe that, in general, $O$ does not satisfy faithfulness if restricted to $\mathbf{P}(A)^*$. For instance, consider $A=[3]$ and $p\in \mathbf{P}(A)^*$, where $\mathrm{Dom}(p)=\{1\}$ and $p(1)=\{(1,2)\}\in \mathbf{P}(A)$; then it is immediately checked that $O(p)=\{1\}\neq \{1, 3\}= \mathrm{Max}(p(1))$.

The next simple proposition describes the relation among faithfulness, averseness and $O$-coherence.

\begin{proposition}\label{comp-fai-ave-oco}
Let $\mathbf{D}\subseteq \mathbf{R}(A)^*$ and $F$ be a {\sc scc} on $\mathbf{D}$. Then the following facts hold true.
\begin{itemize}
\item[$(i)$] If $F$ satisfies faithfulness, then $F$ satisfies averseness.
\item[$(ii)$] If $\mathbf{D}\subseteq \mathbf{O}(A)^*$ and $F$ satisfies averseness and neutrality, then $F$ satisfies faithfulness.
\item[$(iii)$] If there exist $p\in \mathbf{D}$ and $i\in V$ such that $\mathrm{Dom}(p)=\{i\}$ and $p(i)\in \mathbf{O}(A)\setminus\{A^2\}$ and $F$ satisfies faithfulness, then $F$ satisfies $O$-coherence. 
\item[$(iv)$] If there exist $p\in \mathbf{D}$ and $i\in V$ such that $\mathrm{Dom}(p)=\{i\}$, $p(i)\in \mathbf{O}(A)$, $|\mathrm{Max}(p(i))|=1$ and $F$ satisfies averseness, then $F$ satisfies $O$-coherence.
\end{itemize}
\end{proposition}

\begin{proof}
$(i)$ Straightforward.

$(ii)$ Assume that  $\mathbf{D}\subseteq \mathbf{O}(A)^*$ and that $F$ satisfies averseness and neutrality. Let $p\in \mathbf{D}$ be such that $\mathrm{Dom}(p)=\{i\}$ for some $i\in V$. We have to show that $F(p)=\mathrm{Max}(p(i))$. By averseness we know that $F(p)\subseteq \mathrm{Max}(p(i))$. Thus, we are left with proving that $\mathrm{Max}(p(i))\subseteq F(p)$. Let $y\in \mathrm{Max}(p(i))$ and prove that $y\in F(p)$. Since $F(p)\neq\varnothing$, there exists $x\in F(p)$. Since $x\in \mathrm{Max}(p(i))$, there exists $\psi\in \mathrm{Sym}(A)$ such that $\psi(x)=y$ and fixing the alternatives in $ A\setminus \mathrm{Max}(p(i))$. Since $p(i)\in\mathbf{O}(A)$, we have that $p^{\psi}=p$. Thus, using the neutrality of $F$, we obtain $y=\psi(x)\in \psi F(p)=F(p^{\psi})=F(p)$, as desired.

$(iii)$ Assume that there exist $p\in \mathbf{D}$ and $i\in V$ such that $\mathrm{Dom}(p)=\{i\}$, $p(i)\in \mathbf{O}(A)\setminus\{A^2\}$ and assume that $F$ satisfies averseness. Observe first that, by Proposition \ref{newO}, we have $O(p)=\mathrm{Max}(p(i))$.
Let us prove now that $N(p)\not\in \mathcal{PS}(A)$. Indeed, assume by contradiction that $N(p)\in \mathcal{PS}(A)$. By Proposition \ref{flat-o-prof}, we have $\mathrm{Max}(p(i))=O(p)=A$; as a consequence, $p(i)=A^2$, against the assumption. Now it is enough to observe that, by faithfulness of $F$, we have $F(p)=\mathrm{Max}(p(i))=O(p)$.

$(iv)$ Assume that there exist $p\in \mathbf{D}$ and $i\in V$ such that $\mathrm{Dom}(p)=\{i\}$, $p(i)\in \mathbf{O}(A)$ and $|\mathrm{Max}(p(i))|=1$  and assume that $F$ satisfies averseness. Observe first that, by Proposition \ref{newO}, we have $O(p)=\mathrm{Max}(p(i))$.
Since $|\mathrm{Max}(p(i))|=1$, we have that  $p(i)\neq A^2$. As observed in the proof of $(iii)$, the fact that $p(i)\in \mathbf{O}(A)\setminus\{A^2\}$ implies $N(p)\not\in \mathcal{PS}(A)$. Since $F$ satisfies averseness and $|\mathrm{Max}(p(i))|=1$, we have that $F(p)=\mathrm{Max}(p(i))=O(p)$. 
\end{proof}

Let us introduce now an interesting concept of cancellation introduced by Fishburn (1979) for {\sc scc}s defined on a subset of 
$\mathbf{Di}(A)^*$. Recall that the number $av(p,x)$ is the approval score defined in \eqref{appr-score}.

\begin{definition}
Let $\mathbf{D}\subseteq \mathbf{Di}(A)^*$ and $F$ be a {\sc scc} on $\mathbf{D}$. We say that $F$ satisfies {\rm Fishburn-cancellation} if, for every $p\in \mathbf{D}$ such that $av(p,x)=av(p,y)$ for all $x,y\in A$, we have that $F(p)=A$.
\end{definition}

The next proposition explains that for {\sc scc}s defined on a subset of $\mathbf{Di}(A)^*$ the property of cancellation coincides with the property of Fishburn-cancellation. We observe that Fishburn-cancellation is certainly much more intuitive than cancellation when dichotomous orders are involved.

\begin{proposition}\label{canc-fishburn}
Let $\mathbf{D}\subseteq \mathbf{Di}(A)^*$ and $F$ be a {\sc scc} on $\mathbf{D}$. Then $F$ satisfies cancellation if and only if $F$ satisfies Fishburn-cancellation.
\end{proposition}

\begin{proof} Let $p\in \mathbf{Di}(A)^*$. We prove first that the following facts are equivalent:
\begin{itemize}
	\item[$(a)$] $N(p)\in\mathcal{R}(A)$;
	\item[$(b)$] for every $x,y\in A$, $av(p,x)=av(p,y)$.
\end{itemize}
First of all, observe that $(a)$ and $(b)$ are respectively equivalent to
\begin{itemize}
	\item[$(a')$] for every $x,y\in A$ with $x\neq y$, $c_p(x,y)=c_p(y,x)$;
	\item[$(b')$] for every $x,y\in A$ with $x\neq y$, $av(p,x)=av(p,y)$.
\end{itemize}
It is then enough to prove that $(a')$ is equivalent to $(b')$.
Observe that, for every $x,y\in A$ with $x\neq y$, we have that
\[
av(p,x)=|\{i\in\mathrm{Dom}(p): x\in \mathrm{Max}(p(i)),y\in \mathrm{Max}(p(i))\}|
+|\{i\in\mathrm{Dom}(p): x\in \mathrm{Max}(p(i)),y\not\in \mathrm{Max}(p(i))\}|
\]
and
\[
c_p(x,y)=|\mathrm{Dom}(p)|-|\{i\in\mathrm{Dom}(p): x\not\in \mathrm{Max}(p(i)),y\in \mathrm{Max}(p(i))\}|.
\]
Thus, for every $x,y\in A$ with $x\neq y$, we have that
\[
av(p,x)-av(p,y)=(|\mathrm{Dom}(p)|-c_p(y,x))-(|\mathrm{Dom}(p)|-c_p(x,y))=c_p(x,y)-c_p(y,x).
\]
As a consequence, we deduce that $(a')$ is equivalent to $(b')$ and the proof is completed. 
\end{proof}

\subsection{Characterization theorems for {\sc scc}s}

Let us begin with the following characterization theorem for $O_{|\mathbf{D}}$ when  $\mathbf{D}$ is a regular subset of $\mathbf{R}(A)^*$.

\begin{theorem}\label{crucial}
Let $\mathbf{D}\subseteq \mathbf{R}(A)^*$ be regular. Then $O_{|\mathbf{D}}$ is the unique {\sc scc} on $\mathbf{D}$ satisfying neutrality, consistency, cancellation and $O$-coherence.
\end{theorem}

\begin{proof}
Since $\mathbf{D}$ is CPA, CA and CWC, we know by Proposition \ref{Prop-OIT} that $O_{|\mathbf{D}}$ satisfies neutrality, consistency, cancellation. By Proposition \ref{newO} we know that $O_{|\mathbf{D}}$ satisfies $O$-coherence. Let now $F$ be a {\sc scc} on $\mathbf{D}$ satisfies neutrality, consistency, cancellation and $O$-coherence. Since $\mathbf{D}$ is regular and $F$ satisfies neutrality, consistency and cancellation, by Theorem \ref{characterization-profiles}, we know that $F\in \{O_{|\mathbf{D}},I_{|\mathbf{D}},T_{|\mathbf{D}}\}$. Since $F$ satisfies $O$-coherence, we know there exists $p\in \mathbf{D}$ such that $N(p)\not\in\mathcal{PS}(A)$
and $F(p)=O(p)$. By Theorem \ref{characterization-profiles}, we conclude then that $F=O_{|\mathbf{D}}$.
\end{proof}

By means of Theorem \ref{crucial} and the results proved in Sections \ref{Ouguale} and \ref{faoc} we are now in the position to easily prove a variety of characterization results for well-known {\sc scc}s. Those results are collected in Theorem \ref{lista} below. Each statement of that theorem presents one or two lists of at least two properties written between brackets. That must be interpreted as follows: picking any 
property in each list determines a true statement. Thus, Theorem \ref{lista} contains $29$ characterization results.

\begin{theorem}\label{lista}
The following facts hold true.
\begin{itemize}
\item[$(i)$]$BOR_{|\mathbf{L}(A)^*}$ is the unique {\sc scc} on $\mathbf{L}(A)^*$ satisfying neutrality, consistency, cancellation and [$O$-coherence; faithfulness; averseness].
\item[$(ii)$]$BOR$ is the unique {\sc scc} on $\mathbf{O}(A)^*$ satisfying neutrality, consistency, cancellation and [$O$-coherence; faithfulness; averseness].
\item[$(iii)$]$PBOR$ is the unique  {\sc scc} on $\mathbf{P}(A)^*$ satisfying neutrality, consistency, cancellation and [$O$-coherence; averseness]. 
\item[$(iv)$]Let $Y\subseteq [m-1]_0$ with $Y\cap [m-1]\neq\varnothing$. Then
$ABOR_{|\mathbf{T}_Y(A)^*}$ is the unique {\sc scc} on $\mathbf{T}_Y(A)^*$ satisfying neutrality, consistency, cancellation and 
[$O$-coherence; faithfulness; averseness]. 
\item[$(v)$]Let $X\subseteq [m]$ with $X\cap [m-1]\neq \varnothing$. Then 
$AV_{|\mathbf{Di}_X(A)^*}$ is the unique  {\sc scc} on $\mathbf{Di}_X(A)^*$ satisfying neutrality, consistency, [cancellation; Fishburn-cancellation] and [$O$-coherence; faithfulness; averseness].
\item[$(vi)$]$PLU$ is the unique  {\sc scc} on $\mathbf{Di}_1(A)^*$ satisfying neutrality, consistency, [cancellation; Fishburn-cancellation] and [$O$-coherence; faithfulness; averseness].
\item[$(vii)$]$APLU$ is the unique  {\sc scc} on $\mathbf{Di}_{m-1}(A)^*$ satisfying neutrality, consistency, [cancellation; Fishburn-cancellation] and [$O$-coherence; faithfulness; averseness]. 
\end{itemize}
\end{theorem}

\begin{proof}
Let $\Sigma_1$, $\Sigma_2$ and $\Sigma_3$ be the sets of subsets of $\mathbf{R}(A)^*$ defined as
\begin{eqnarray*}
\Sigma_1&=&\{\mathbf{O}(A)^*, \mathbf{L}(A)^*\}\cup\{\mathbf{T}_Y(A)^*: Y\subseteq [m-1]_0, Y\cap[m-1]\neq\varnothing\},\\
\Sigma_2&=&\{\mathbf{Di}_X(A)^*:X\subseteq [m],X\cap[m-1]\neq\varnothing\},\\
\Sigma_3&=&\{\mathbf{P}(A)^*\}.
\end{eqnarray*}
Consider now $\mathbf{D}\in \Sigma_1\cup\Sigma_2\cup\Sigma_3$.  By Theorems \ref{tutte-new}, \ref{dicotomici-all} and \ref{truncated-all}, we know that $\mathbf{D}$ is regular. Hence, by Proposition \ref{newO}, $O_{\mid \mathbf{D}}$ satisfies $O$-coherence.
Observe now that the following observations hold true.
\begin{itemize}
\item[$(c1)$] Assume that  $\mathbf{D}\in \Sigma_1\cup \Sigma_2$. Then we have $\mathbf{D}\subseteq \mathbf{O}(A)^*$ and there exists $p\in \mathbf{D}$  such that $\mathrm{Dom}(p)=\{1\}$, $p(1)\in \mathbf{O}(A)\setminus\{A^2\}$. Then, by Proposition \ref{comp-fai-ave-oco}, if a {\sc scc} $F$ on $\mathbf{D}$ satisfies neutrality and one between faithfulness and averseness, then $F$ satisfies $O$-coherence. Moreover, by Proposition \ref{newO}, $O_{|\mathbf{D}}$ satisfies faithfulness and then averseness.
\item[$(c2)$] Assume that $\mathbf{D}\in \Sigma_3$. Then there exists $p\in \mathbf{D}$ such that $\mathrm{Dom}(p)=\{1\}$, $p(1)\in \mathbf{O}(A)$ and $|\mathrm{Max}(p(1))|=1$. Thus, by Proposition \ref{comp-fai-ave-oco}, if a {\sc scc} $F$ on $\mathbf{D}$ satisfies averseness, then $F$ satisfies $O$-coherence.
Moreover, by Proposition \ref{newO}, $O_{|\mathbf{D}}$ satisfies averseness. 
\item[$(c3)$] Assume that $\mathbf{D}\in \Sigma_2$, by Proposition \ref{canc-fishburn}, we have that  a {\sc scc} $F$ on $\mathbf{D}$ satisfies cancellation if and only if it satisfies Fishburn-cancellation.
\end{itemize}
Considering now Theorem \ref{crucial} and Propositions \ref{Borda-equal-O}-\ref{APLU}, 
 we have that 
$(c1)$ implies $(i)$, $(ii)$ and $(iv)$;
$(c1)$ and $(c3)$ imply $(v)$, $(vi)$ and $(vii)$;
$(c2)$ implies $(iii)$.
\end{proof}

It is important to note that most of the characterization results in Theorem \ref{lista} are new. Some of them instead correspond or immediately imply well-known results. Such results are listed below.

\begin{theorem}{\rm [Young, 1974, Theorem 1]}\label{young-theorem}
$BOR_{|\mathbf{L}(A)^*}$ is the unique {\sc scc} on $\mathbf{L}(A)^*$ satisfying neutrality, consistency, cancellation and faithfulness.
\end{theorem}

\begin{theorem}{\rm [Young, 1974]}\label{borda-young}
$BOR$ is the unique {\sc scc} on $\mathbf{O}(A)^*$ satisfying neutrality, consistency, cancellation and faithfulness.
\end{theorem}

\begin{theorem}{\rm [Cullinan et al., 2014, Theorem 2]}\label{culli-theorem}
$PBOR$ is the unique  {\sc scc} on $\mathbf{P}(A)^*$ satisfying neutrality, consistency, cancellation and averseness.
\end{theorem}

\begin{theorem}{\rm [Fishburn, 1979, Theorem 4]}\label{Fishburn}
Let $X\subseteq [m]$ with $X\cap [m-1]\neq \varnothing$. Then $AV_{|\mathbf{Di}_X(A)^*}$ is the unique  {\sc scc} on $\mathbf{Di}_X(A)^*$ satisfying neutrality, consistency, Fishburn-cancellation and faithfulness. 
\end{theorem}

Note that Theorem \ref{young-theorem} corresponds to one of the statements in Theorem \ref{lista}$(i)$; 
Theorem \ref{borda-young} corresponds to one of the statements in Theorem \ref{lista}$(ii)$;
Theorem \ref{culli-theorem} corresponds to one of the statements in Theorem \ref{lista}$(iii)$;
Theorem \ref{Fishburn} corresponds to one of the statements in Theorem \ref{lista}$(v)$. It is worth noticing that Young (1974) states that the content of Theorem \ref{borda-young} can be easily proved adapting the proof of Theorem \ref{young-theorem}. However, it does not seem obvious how one could adapt Young's reasoning to the new framework since it appears, in many parts, guided by the specific properties of linear orders.

From Theorem \ref{lista}, it is possible to easily get further results present in the literature with little effort.
First of all, let us show how we can deduce an interesting characterization result by Terzopoulou and Endriss (2021). In order to show that fact we first need to recall the definition of monotonicity considered in Terzopoulou and Endriss (2021). Given $x,y\in A$ with $x\neq y$, let $(x\,y)\in\mathrm{Sym}(A)$ be the permutation that switches $x$ and $y$ and leaves all the other alternatives fixed. A {\sc scc} $F$ on $\mathbf{T}(A)^*$ is monotonic if the following condition holds true: given $p,p'\in \mathbf{T}(A)^*$ with $\mathrm{Dom}(p)=\mathrm{Dom}(p')$, $i^*\in \mathrm{Dom}(p)$ and $x,y\in A$, if $p(i)=p'(i)$ for all $i\in \mathrm{Dom}(p)\setminus \{i^*\}$, $x\succ_{p(i^*)}y$, $y\in F(p)$ and $p'(i^*)=p(i^*)^{(x\, y)}$, then $F(p')=\{y\}$.

\begin{theorem}{\rm [Terzopoulou and Endriss, 2021, Theorem 6]}\label{terzo-endriss}
$ABOR$ is the unique {\sc scc} on $\mathbf{T}(A)^*$ satisfying neutrality, consistency, cancellation and monotonicity.
\end{theorem}

\begin{proof} We know that $ABOR$ satisfies neutrality, consistency and cancellation. It is easily proved that $ABOR$ satisfies monotonicity, as well. Let now $F$ be a {\sc scc} on $\mathbf{T}(A)^*$ satisfying neutrality, consistency, cancellation and monotonicity. We have to prove that $F=ABOR$.
By Theorem \ref{lista}$(iv)$, it is enough to show that  $F$ satisfies $O$-coherence. 
Let $x\in A$ and consider $p:\{1\}\to\mathbf{Di}_1(A)$ such that $\mathrm{Max}(p(1))=\{x\}$. Then $p\in \mathbf{Di}_1(A)^*\subseteq \mathbf{T}(A)^*$. 
Since  $\delta^{N(p)}(x)=m-1\neq 0$, we have that  $N(p)\not\in \mathcal{PS}(A)$. By Proposition \ref{comp-fai-ave-oco}, we have that $O(p)=\{x\}$. We want to show that $F(p)=\{x\}$. Since $F(p)\neq \varnothing$, it is enough to show that $F(p)\cap (A\setminus\{x\})=\varnothing.$
Assume, by contradiction, that there exists $y\in F(p)$ such that $y\neq x$. Of course, $x\succ_{p(1)}y$. Thus, by monotonicity  and by neutrality of $F$, we deduce 
$$\{y\}=F(p^{(x\,y)})=(x\,y)F(p)\supseteq \{x\},$$
against $y\neq x$.
Hence, $F$ satisfies $O$-coherence.
\end{proof}

Finally, we show that Theorem \ref{lista} allows to easily get a characterization result for the Plurality rule by Sekiguchi (2012)
and a characterization result for the anti-Plurality rule by Kurihara (2018).\footnote{Theorem 1 in Sekiguchi (2012) actually refers to the Plurality rule defined in $\mathbf{L}(A)^*$ and involves a further property called tops-only. However, because of the tops-only property, that result is equivalent to our Theorem \ref{Seki}. Analogously, Theorem 1 in Kurihara (2018) actually refers to the anti-Plurality rule defined in $\mathbf{L}(A)^*$ and involves a further property called bottoms-only. However, because of the bottoms-only property, that result is equivalent to our Theorem \ref{kuri}.} In order  to do that we need a simple preliminary lemma. In the sequel, for $p\in \mathbf{R}(A)^*$, $\varphi\in \mathrm{Sym}(\mathrm{Dom}(p))$ and $\psi\in \mathrm{Sym}(A)$, we denote by $p^{(\varphi, \psi)}$ the element of $\mathbf{R}(A)^*$ such that $\mathrm{Dom}(p^{(\varphi,\psi)})=\mathrm{Dom}(p)$ and, for every $i\in \mathrm{Dom}(p)$, $p^{(\varphi,\psi)}(i)=p(\varphi^{-1}(i))^\psi$.\footnote{Details on the symbol $p^{(\varphi, \psi)}$ and its properties can be found in Bubboloni and Gori (2014) and (2015).}

\begin{lemma}\label{an+neut} 
Let $\mathbf{D}\subseteq \mathbf{R}(A)^*$, $F$ be {\sc scc} on $\mathbf{D}$ satisfying anonymity and neutrality  and $p\in \mathbf{D}$. Assume that, for every $\psi\in \mathrm{Sym}(A)$, there exists $\varphi\in \mathrm{Sym}( \mathrm{Dom}(p))$ such that $p^{(\varphi,\psi)}=p$. Then $F(p)=A$.
\end{lemma}

\begin{proof} Let $y\in A$ and prove that $y\in F(p)$.
Since $F(p)\neq \varnothing$, there exists $x^*\in F(p)$. Consider then $\psi\in \mathrm{Sym}(A)$ such that $y=\psi(x^*)$
and let $\varphi\in \mathrm{Sym}( \mathrm{Dom}(p))$ be such that $p^{(\varphi,\psi)}=p$.
Since $F$ is anonymous and neutral and since $x^*\in F(p)$, we have that $y=\psi(x^*)\in \psi F(p)=F(p^{(\varphi,\psi)})=F(p)$.
\end{proof}

\begin{theorem}\label{Seki}{\rm [Sekiguchi, 2012, Theorem 1]} $PLU$ is the unique  {\sc scc} on $\mathbf{Di}_1(A)^*$ satisfying anonymity, neutrality, consistency and faithfulness. 
\end{theorem}
\begin{proof} We know that $PLU$ satisfies anonymity, neutrality, consistency and faithfulness. Let now $F$ be a {\sc scc} satisfying anonymity, neutrality, consistency and faithfulness. We have to prove that $F=APLU$. By Theorem \ref{lista}$(vi)$, it is enough to show that  $F$ satisfies 
Fishburn-cancellation. Let  $p\in \mathbf{Di}_1(A)^*$ be such that  $av(p,x)=av(p,y)$ for all $x,y\in A$. Then we have $|\{i\in \mathrm{Dom}(p): \mathrm{Max}(p(i))=\{x\}\}|=|\{i\in \mathrm{Dom}(p):  \mathrm{Max}(p(i))=\{y\}\}|$ for all $x,y\in A$. As a consequence, for every $\psi\in \mathrm{Sym}(A)$, there exists $\varphi\in \mathrm{Sym}( \mathrm{Dom}(p))$ such that $p^{(\varphi,\psi)}=p.$ Since $F$ is anonymous and neutral, by Lemma \ref{an+neut} we deduce $F(p)=A.$
\end{proof}

\begin{theorem}\label{kuri}{\rm [Kurihara, 2018, Theorem 1]} $APLU$ is the unique  {\sc scc} on $\mathbf{Di}_{m-1}(A)^*$ satisfying 
anonymity, neutrality, consistency and averseness.
\end{theorem}
\begin{proof} We know that $APLU$ satisfies anonymity, neutrality, consistency and averseness. Let now $F$ be a {\sc scc} satisfying anonymity, neutrality, consistency and averseness. We have to prove that $F=APLU$. 
By Theorem \ref{lista}$(vii)$, it is enough to show that $F$ satisfies Fishburn-cancellation. Let  $p\in \mathbf{Di}_{m-1}(A)^*$ be such that 
$av(p,x)=av(p,y)$ for all $x,y\in A$. Then we have  $|\{i\in \mathrm{Dom}(p): x\notin \mathrm{Min}(p(i))\}|=|\{i\in \mathrm{Dom}(p): y\notin \mathrm{Min}(p(i))\}|$ for all $x,y\in A$. As a consequence we also have
$|\{i\in \mathrm{Dom}(p): x\in \mathrm{Min}(p(i))\}|=|\{i\in \mathrm{Dom}(p): y\in \mathrm{Min}(p(i))\}|$ for all $x,y\in A$. 
Hence, for every $\psi\in \mathrm{Sym}(A)$, there exists $\varphi\in \mathrm{Sym}( \mathrm{Dom}(p))$ such that $p^{(\varphi,\psi)}=p.$ Since $F$ is anonymous and neutral, by Lemma \ref{an+neut} we deduce $F(p)=A$.
\end{proof}

\vspace{10mm}
\noindent {\Large{\bf Acknowledgements}}
\vspace{2mm}

\noindent Daniela Bubboloni is supported by GNSAGA
of INdAM (Italy) and by the national project PRIN 2022- 2022PSTWLB - Group Theory and Applications - CUP B53D23009410006.  We thank the participants to the Workshop on Collective Decisions: Current Trends (Granada, 2023) for their valuable comments on a preliminary version of the paper.

%
%

\vspace{10mm}
\noindent {\Large{\bf References}}
\vspace{2mm}

\noindent Al\'os-Ferrer, C., 2006. A simple characterization of approval voting. 
Social Choice and Welfare 27, 621--625. 
\vspace{2mm}

\noindent Barber\`a, S., Bossert, W., 2023. Opinion aggregation: Borda and Condorcet revisited. Journal of Economic Theory 210, 105--654.
\vspace{2mm}

\noindent Belkin, A.R., Gvozdik. A.A., 1989. The flow model for ranking objects. 
Problems of Cybernetics, Decision Making and Analysis of Expert Information, Eds. Dorofeyuk, A., Litvak, B., Tjurin, Yu., Moscow, 109--117 (in Russian).
\vspace{2mm}

\noindent Bouyssou, D., 1992.  Ranking methods based on valued preference relations: a characterization of the net flow method. 
European Journal of Operational Research 60, 61--67.\vspace{2mm}

\noindent Brams, S.J., Fishburn, P.C., 1978. Approval Voting. 
The American Political Science Review 72, 831--847.
\vspace{2mm}

\noindent Brandl, F., Peters, D., 2019. An axiomatic characterization of the Borda mean rule. Social Choice and Welfare 52, 685--707.
\vspace{2mm}

\noindent Brans, J.P., Vincke, Ph., Mareschal, B., 1986. How to select and how to rank projects: The Promethee method.
European Journal of Operational Research 24, 228--238.\vspace{2mm}

\noindent Bubboloni, D., Gori, M., 2014. Anonymous and neutral majority rules. 
Social Choice and Welfare 43, 377--401.
\vspace{2mm}

\noindent Bubboloni, D., Gori, M., 2015. Symmetric majority rules. 
Mathematical Social Sciences 76, 73--86.
\vspace{2mm}

\noindent Bubboloni, D., Gori, M., 2018. The flow network method. 
Social Choice and Welfare 51, 621--656.
\vspace{2mm}

\noindent Cullinan, J., Hsiao, S.K., Polett, D., 2014. A borda count for partially ordered ballots. 
Social Choice and Welfare 42, 913--926.
\vspace{2mm}

\noindent Duddy, C., Piggins, A., Zwicker, W.S., 2016. Aggregation of binary evaluations: a Borda-like approach.
Social Choice and Welfare 46,  301--333.
\vspace{2mm}

\noindent Dummett, M., 1997. {\it Principles of Electoral Reform}. Oxford University Press.
\vspace{2mm}

\noindent Fishburn, P.C., 1979. Symmetric and consistent aggregation with dichotomous voting. 
In: J-J. Laffont (ed.), \textsl{Aggregation and revelation of preferences}, North Holland.\vspace{2mm}

\noindent Gvozdik, A.A., 1987. Flow interpretation of the tournament model for ranking. 
Abstracts of the VI-th Moscow Conference of young scientists on cybernetics and computing, Moscow: Scientific Council on Cybernetics of RAS, 1987, p.56 (in Russian).
\vspace{2mm}

\noindent Kurihara, T., 2018. A simple characterization of the anti-plurality rule. 
Economics Letters 168, 110--111.
\vspace{2mm}

\noindent Henriet, D., 1985. The Copeland choice function - An axiomatic characterization. 
Social Choice and Welfare 2, 49--64.
\vspace{2mm}

\noindent Hansson, B., Sahlquist, H., 1976. A proof technique for social choice with variable electorate. 
Journal of Economic Theory 13, 193--200.
\vspace{2mm}

\noindent Langville, A.N., Meyer, C.D., 2012. {\it Who Is $\#1$? The Science of Rating and Ranking}. 
Princeton University Press.
\vspace{2mm}

\noindent Laslier, J.-F., 1997. {\it Tournament solutions and majority voting}. 
Studies in Economic Theory, Volume 7. Springer.
\vspace{2mm}

\noindent Mas-Colell, A., Whinston, M.D., Green, J., 1995. {\it Microeconomic Theory}. 
New York: Oxford University Press.
\vspace{2mm}

\noindent Moulin, H., 1988. Condorcet’s principle implies the no show paradox. 
Journal of Economic Theory 45, 53--64.
\vspace{2mm}

\noindent Myerson, R.B., 1995. Axiomatic derivation of scoring rules without the ordering assumption. 
Social Choice and Welfare 12, 59--74.
\vspace{2mm}

\noindent Rubinstein, A., 1980. Ranking the participants in a tournament. 
SIAM Journal on Applied Mathematics 38, 108--111.\vspace{2mm}

\noindent Sekiguchi, Y., 2012. A characterization of the Plurality rule. 
Economic Letters 116, 330--332.\vspace{2mm}

\noindent Schulze, M., 2011. A new monotonic, clone-independent, reversal symmetric, and Condorcet-consistent single-winner election method. 
Social Choice and Welfare  36, 267--303.
\vspace{2mm}

\noindent Terzopoulou, Z., Endriss, U., 2021. The Borda class - An axiomatic study of the Borda rule on top-truncated preferences.
Journal of Mathematical Economics 92, 31--40.
\vspace{2mm}

\noindent van den Brink, R., Gilles, R.P., 2009. The outflow ranking method for weighted directed graphs. 
European Journal of Operational Research 193, 484--491.
\vspace{2mm}

\noindent Vorsatz, M., 2008. Scoring rules on dichotomous preferences. 
Social Choice and Welfare 31, 151--162.
\vspace{2mm}

\noindent Young, H.P., 1974. An axiomatization of Borda's rule. 
Journal of Economic Theory 9, 43--52.
\vspace{2mm}

\noindent Young, H.P., 1975. Social choice scoring functions. 
SIAM Journal on Applied Mathematics 28, 824--838.
\vspace{2mm}

\noindent Young, H.P., 1988. Condorcet's theory of voting. 
American Political Science Review 82, 1231--1244.
\vspace{2mm}

\noindent Zwicker, W.S., 1991. The voters' paradox, spin, and the Borda count. 
Mathematical Social Sciences 22, 187--227.

\appendix

\section{Proof of Proposition \ref{generale-reti}}

Throughout the proof, given $N\in \mathcal{N}(A)$, we denote the capacity of $N$ by $c^N$.

$(i)$ The fact that $\mathcal{C}(A)$, $\mathcal{R}(A)$, $\mathcal{B}_0(A)$, $\mathcal{B}(A)$, $\mathcal{PS}(A)$ and $\mathcal{O}(A)$ are subspaces of 
$\mathcal{N}(A)$ follows easily by their definition.

$(ii)$ The fact that $\mathrm{dim}\,\mathcal{C}(A)=1$ immediately follows by observing that, for every $k\in \mathbb{Q}$, we have that $N(k)=kN(1)$.
It is easily observed that the set $R=\{N_{xy}+N_{yx}: (x,y)\in A^2_*\}$ is a basis for $\mathcal{R}(A)$ and that there is a natural bijection between $R$ and $P_2(A)$. Thus, $\mathrm{dim}\,\mathcal{R}(A)=\frac{m(m-1)}{2}$.

Let us consider now the set $\mathcal{B}(A)$.
We claim that the set $B=\{N_{xy}-N_{yx}:(x,y)\in A^2_*\}\cup\{N(1)\}$ generates $\mathcal{B}(A)$. Given $N\in \mathcal{N}(A)$, let us associate with $N$ the network
\begin{equation}\label{rev-part}
\widehat{N}\coloneq\sum_{(x,y)\in A^2_*}  [c^N(x,y)-c^N(y,x)] \left(N_{xy}-N_{yx}\right).
\end{equation}
For every $(x,y)\in A^2_*$, we have that $c^{\widehat{N}}(x,y)=2(c^N(x,y)-c^N(y,x))$; therefore,  $\widehat{N}\in \mathcal{B}_0(A)$.
Let us prove now that if $N\in \mathcal{B}_k(A)$ with $k\in \mathbb{Q}$, then 
\begin{equation}\label{ug}
2N=kN(1)+\frac{1}{2}\widehat{N}.
\end{equation}
In order to do that, we need to show that the networks on the two sides of \eqref{ug} have the same capacity.
Indeed, for every $(x,y)\in A^2_*$, we have that
\[
c^{kN(1)+\frac{1}{2}\widehat{N}}(x,y)=k+c^N(x,y)-c^N(y,x)=2c^N(x,y)=c^{2N}(x,y).
\]
That implies that $N$ is a linear combination of networks in $B$. As a consequence, we deduce that $B$ generates $\mathcal{B}(A)$.

Consider now the function $L:\mathcal{N}(A)\rightarrow \mathcal{N}(A)$ defined, for every $N\in \mathcal{N}(A)$, by $L(N)=\widehat{N}$. 
Clearly, $L$ is linear, $\mathrm{Ker}(L)=\mathcal{R}(A)$, and $\mathrm{Im}(L)\subseteq \mathcal{B}_0(A)$. Moreover, as a  particular case of \eqref{ug}, we have that if $N\in \mathcal{B}_0(A)$, then $N=\frac{1}{4}\widehat{N}$. Thus, we deduce that
\begin{equation}\label{Limage}
L( \mathcal{N}(A))=\mathcal{B}_0(A). 
\end{equation}
Since $\mathrm{dim}(\mathcal{N}(A))=\mathrm{dim}\,(\mathcal{R}(A))+\mathrm{dim}(\mathcal{B}_0(A))$, we get the equality $m(m-1)=\frac{m(m-1)}{2}+\mathrm{dim}\,(\mathcal{B}_0(A))$. As a consequence, we conclude that $\mathrm{dim}\,(\mathcal{B}_0(A))=\frac{m(m-1)}{2}$.

By \eqref{ug} and \eqref{Limage}, we also have $\mathcal{B}(A)\leq \mathcal{C}(A)+\mathcal{B}_0(A)$. On the other hand, it is clear that $\mathcal{C}(A)+\mathcal{B}_0(A)\leq \mathcal{B}(A).$ Thus $\mathcal{B}(A)= \mathcal{C}(A)+\mathcal{B}_0(A)$. Since it is immediately checked that $ \mathcal{C}(A)\cap\mathcal{B}_0(A)=\{N(0)\}$, we deduce that 
\begin{equation}\label{somma1}
\mathcal{B}(A)= \mathcal{C}(A)\oplus \mathcal{B}_0(A)
\end{equation}
and then $\mathrm{dim}\,\mathcal{B}(A)=\frac{m(m-1)}{2}+1.$ 

The fact  that $\mathcal{PS}(A)=(m-1)^2$ is the content of Proposition \ref{ker-delta}. The fact that $\mathcal{O}(A)=m$ follows from the fact that the outstar networks generates  $\mathcal{O}(A)$ and are known to be independent.

Observe now that $\mathcal{R}(A)\cap\mathcal{B}_0(A)=\{N(0)\}$. Since $\mathrm{dim}(\mathcal{R}(A) +\mathcal{B}_0(A))=\mathrm{dim}(\mathcal{R}(A))+\mathrm{dim}(\mathcal{B}_0(A))=m(m-1),$ we deduce that
\begin{equation}\label{oplus}
\mathcal{N}(A)=\mathcal{R}(A)\oplus\mathcal{B}_0(A).
\end{equation}
As a consequence, $\mathcal{N}(A)=\mathcal{PS}(A) +\mathcal{B}(A)$. Thus,
\[
m(m-1)=(m-1)^2+\frac{m(m-1)}{2}+1-\mathrm{dim}\, \left(\mathcal{PS}(A)\cap \mathcal{B}(A)\right),
\]
and therefore $\mathrm{dim}\, \left(\mathcal{PS}(A)\cap \mathcal{B}(A)\right)=\frac{(m-1)(m-2)}{2}+1.$

$(iii)$ The two inclusions $\mathcal{R}(A)\subseteq \mathcal{PS}(A)$ and  $\mathcal{C}(A)\subseteq \mathcal{R}(A)\cap \mathcal{B}(A)$ are immediate. In order to get $\mathcal{C}(A)=\mathcal{R}(A)\cap \mathcal{B}(A)$, we need to show that $\mathcal{R}(A)\cap \mathcal{B}(A)\subseteq \mathcal{C}(A)$. 
Let $N\in \mathcal{R}(A)\cap \mathcal{B}(A)$. Then, for every $(x,y)\in A^2_*$, we have that $c^N(x,y)=c^N(y,x)$ and $c^N(x,y)+c^N(y,x)=k$. As a consequence, for every $(x,y)\in A^2_*$, we have that $c^N(x,y)=\frac{k}{2}$ and thus $N\in \mathcal{C}(A)$. 

By \eqref{somma1} and \eqref{oplus}, we know that $\mathcal{B}(A)= \mathcal{C}(A)\oplus \mathcal{B}_0(A)$ and $\mathcal{N}(A)=\mathcal{R}(A)\oplus\mathcal{B}_0(A)$. Since $\mathcal{B}_0(A)\leq \mathcal{B}(A)$, we also have $\mathcal{N}(A)=\mathcal{R}(A)+\mathcal{B}(A)$.

Since $\delta(\mathcal{O}(A))=\mathbb{Q}\{\delta(N_x):x\in A\}$, by Proposition \ref{im-delta-1}, we deduce that
$\delta(\mathcal{O}(A))=\delta(\mathcal{N}(A))$. Thus, it follows that $\mathrm{dim\, } \delta(\mathcal{O}(A))= m-1$ and therefore
$m-1=m-\mathrm{dim\, }(\mathrm{Ker}(\delta)\cap \mathcal{O}(A))$. As a consequence, $\mathrm{dim\, }(\mathcal{PS}(A)\cap \mathcal{O}(A))=1$.
Since $N(1)=\sum_{x\in A}N_x\in \mathcal{PS}(A)\cap\mathcal{O}(A)$, we conclude that
$\mathcal{C}(A)=\mathbb{Q}N(1)=\mathcal{PS}(A)\cap \mathcal{O}(A)$.

In order to prove the equality $\mathcal{R}(A)=(\mathcal{O}(A)+\mathcal{R}(A))\cap \mathcal{PS}(A)$, 
we only need to prove the inclusion
$(\mathcal{O}(A)+\mathcal{R}(A))\cap \mathcal{PS}(A)\subseteq\mathcal{R}(A)$ because the other inclusion is trivial.
Let $N\in (\mathcal{O}(A)+\mathcal{R}(A))\cap \mathcal{PS}(A)$. Then, there exist $R\in \mathcal{R}(A)$ and $q_x\in \mathbb{Q}$ for all $x\in A$,  such that
\begin{equation}\label{formaN}
N=\sum_{x\in A}q_xN_x+R.
\end{equation}
Recalling that $\mathcal{R}(A)\subseteq \mathcal{PS}(A)=\mathrm{Ker}(\delta)$, we have that
\[
0=\delta(N)=\delta\left(\sum_{x\in A}q_xN_x\right).
\]
Thus, $\sum_{x\in A}q_xN_x\in \mathcal{O}(A)\cap\mathcal{PS}(A)=\mathcal{C}(A)\subseteq \mathcal{R}(A)$.
From \eqref{formaN}, we conclude that
$N\in \mathcal{R}(A).$

$(iv)$ Straightforward.

\section{Proof of Theorem \ref{ext-finale}}

The proof of Theorem \ref{ext-finale} is obtained through some intermediate propositions. 

\begin{proposition}\label{ext1} 
Let $\mathcal{D}\subseteq\mathcal{N}(A)$ and $\mathscr{F}$ a {\sc ns} on $\mathcal{D}$ 
 satisfying neutrality, consistency and cancellation.
Assume that, for every $N\in \mathbb{Z}\mathcal{D}$, there exists $R\in \mathcal{D}\cap \mathcal{R}(A)$ such that $N+R\in\mathcal{D}$. Then there exists $\mathscr{F}':\mathbb{Z}\mathcal{D}\to P_*(A)$ extension of $\mathscr{F}$ satisfying neutrality, consistency and cancellation. 
\end{proposition}

\begin{proof}
By assumption we know that, for every $N\in \mathbb{Z}\mathcal{D}$, there exists $R\in \mathcal{D}\cap \mathcal{R}(A)$ such that $N+R\in\mathcal{D}$.
Consider then $\mathscr{F}':\mathbb{Z}\mathcal{D}\to P_*(A)$ be defined, for every $N\in \mathbb{Z}\mathcal{D}$, by
\[
\mathscr{F}'(N)=\mathscr{F}(N+R),
\]
where $R\in\mathcal{D}\cap \mathcal{R}(A)$ is such that  $N+R\in\mathcal{D}$. We are going to prove that $\mathscr{F}'$ is a well defined extension of $\mathscr{F}$ and satisfies neutrality, consistency and cancellation.

Let us prove first that $\mathscr{F}'$ is well defined. Consider $N\in\mathbb{Z}\mathcal{D}$ and assume that $R,R'\in\mathcal{D}\cap \mathcal{R}(A)$ are such that $N+R\in\mathcal{D}$ and $N+R'\in\mathcal{D}$. Since $\mathscr{F}$ satisfies consistency and cancellation, we can apply twice Proposition \ref{sum-R} and deduce that $\mathscr{F}(N+R)=\mathscr{F}(N+R+R')=\mathscr{F}(N+R')$.


We next show that $\mathscr{F}'$ is an extension of $\mathscr{F}$. Let $N\in\mathcal{D}$ and pick $R\in\mathcal{D}\cap \mathcal{R}(A)$ such that $N+R\in \mathcal{D}.$ Then, by the definition of $\mathscr{F}'$ and Proposition \ref{sum-R}, we have $\mathscr{F}'(N)=\mathscr{F}(N+R)=\mathscr{F}(N)$.

Let us prove now that $\mathscr{F}'$ satisfies neutrality. We first show that $\mathbb{Z}\mathcal{D}$ is CVP. Let $N\in \mathbb{Z}\mathcal{D}$, $\psi\in \mathrm{Sym}(A)$ and prove that $N^\psi\in \mathbb{Z}\mathcal{D}$. By definition of $\mathbb{Z}\mathcal{D}$, we have that there exist $k\in \mathbb{N}$, $N_1,\ldots, N_k\in \mathcal{D}$ and $z_1,\ldots, z_k\in\mathbb{Z}$ such that $N=\sum_{i=1}^kz_iN_i$. Thus, we have that
\[
N^\psi=\left(\sum_{i=1}^kz_iN_i\right)^\psi=\sum_{i=1}^kz_iN_i^\psi.
\] 
Since $\mathscr{F}$ satisfies neutrality we know that $\mathcal{D}$ is CVP.
We conclude then that $N^\psi\in \mathbb{Z}\mathcal{D}$, since, for every $i\in[k]$, $N_i^\psi\in\mathcal{D}$.
Consider now $N\in \mathbb{Z}\mathcal{D}$ and $\psi\in \mathrm{Sym}(A)$ and prove that $\mathscr{F}'(N^\psi)=\psi(\mathscr{F}'(N))$.
Indeed, consider $R\in\mathcal{D}\cap \mathcal{R}(A)$ such that $N+R\in\mathcal{D}$. Since $\mathcal{D}$
and $\mathcal{R}(A)$ are CVP, we have that $R^\psi\in \mathcal{D}\cap \mathcal{R}(A)$ and $(N+R)^\psi=N^\psi+R^\psi\in \mathcal{D}$.
Thus, by the neutrality of $\mathscr{F}$ we have
\[
\mathscr{F}'(N^\psi)=\mathscr{F}(N^\psi+R^\psi)=\mathscr{F}((N+R)^\psi)=\psi(\mathscr{F}(N+R))=\psi(\mathscr{F}'(N)).
\]
That finally proves the neutrality of $\mathscr{F}'$.

Let us prove now that $\mathscr{F}'$ is consistent. It is clear that $\mathbb{Z}\mathcal{D}$ is CA. Consider now 
$N,M\in\mathbb{Z}\mathcal{D}$ be such that $\mathscr{F}'(N)\cap \mathscr{F}'(M)\neq\varnothing$. We want to show that 
$\mathscr{F}'(N+M)=\mathscr{F}'(N)\cap \mathscr{F}'(M)$. Consider $S,T\in\mathcal{D}\cap \mathcal{R}(A)$ such that $N+S,M+T\in\mathcal{D}$. Thus,
$\mathscr{F}'(N)=\mathscr{F}(N+S)$, $\mathscr{F}'(M)=\mathscr{F}(M+T)$ and hence
$\mathscr{F}(N+S)\cap \mathscr{F}(M+T)\neq\varnothing$. By consistency of $\mathscr{F}$, we have that 
\[
\mathscr{F}((N+S)+(M+T))=\mathscr{F}(N+S)\cap \mathscr{F}(M+T).
\]
Observe now that $(N+S)+(M+T)=(N+M)+(S+T)\in\mathcal{D}$, $N+M\in\mathbb{Z}\mathcal{D}$ and $S+T\in\mathcal{D}\cap \mathcal{R}(A)$ since both $\mathcal{D}$ and $\mathcal{R}(A)$ are CA. Thus, by the definition of $\mathscr{F}'$, we get
\[
\mathscr{F}'(N+M)=\mathscr{F}(N+M+S+T)=\mathscr{F}(N+S)\cap \mathscr{F}(M+T)=\mathscr{F}'(N)\cap \mathscr{F}'(M),
\]
and consistency of $\mathscr{F}'$ is proved.

Let us finally prove that $\mathscr{F}'$ satisfies cancellation. Let $N\in\mathbb{Z}\mathcal{D}\cap \mathcal{R}(A)$ and prove that $\mathscr{F}'(N)=A$. Indeed, let 
$R\in\mathcal{D}\cap \mathcal{R}(A)$ be such that $N+R\in \mathcal{D}$. Then, since $\mathcal{R}(A)$ is CA, we also have that $N+R\in \mathcal{D}\cap \mathcal{R}(A)$.
Thus, by cancellation of $\mathscr{F}$, we get
$
\mathscr{F}'(N)=\mathscr{F}(N+R)=A.
$
\end{proof}

\begin{proposition}\label{ext2}
Let $\mathcal{D}\subseteq\mathcal{N}(A)$ such that $\mathbb{Z}\mathcal{D}=\mathcal{D}$ and $\mathscr{F}$ a {\sc ns} on $\mathcal{D}$  satisfying neutrality, consistency and cancellation.
Then there exists $\mathscr{F}':\mathbb{Q}\mathcal{D}\to P_*(A)$ extension of $\mathscr{F}$ satisfying neutrality, consistency and cancellation.
\end{proposition}

\begin{proof}Recall that $\mathbb{Q}\mathcal{D}$ is a vector space.
First of all, note that, for every $N\in \mathbb{Q}\mathcal{D}$, there exists $k\in\mathbb{N}$ such that $kN\in \mathcal{D}$. Indeed, there exist $s\in\mathbb{N}$, $q_1,\ldots,q_s\in \mathbb{Q}$ and $N_1,\ldots,N_s\in\mathcal{D}$ such that $N=\sum_{i=1}^sq_iN_i$. Consider then $k\in\mathbb{N}$ such that, for every $i\in[s]$,
$kq_i\in\mathbb{Z}$. Then we have
\[
kN=\sum_{i=1}^s(kq_i)N_i\in\mathbb{Z}\mathcal{D}=\mathcal{D}.
\]
Given $N\in \mathbb{Q}\mathcal{D}$, we have then that 
\[
K(N)\coloneq \{k\in\mathbb{N}: kN\in \mathcal{D}\}\neq\varnothing.
\]
Define then $k_N\coloneq \min K(N)$ and consider $\mathscr{F}':\mathbb{Q}\mathcal{D}\to P_*(A)$ be defined, for every $N\in  \mathbb{Q}\mathcal{D}$, by 
$\mathscr{F}'(N)\coloneq \mathscr{F}(k_NN).$ We are going to prove that $\mathscr{F}'$ is an extension of $\mathscr{F}$ and satisfies neutrality, consistency, cancellation.

We have that $\mathscr{F}'$ is an extension of $\mathscr{F}$ since, for every $N\in \mathcal{D}$, $k_N=1$ and so $\mathscr{F}'(N)=\mathscr{F}(k_NN)=\mathscr{F}(N)$.

Let us prove that $\mathscr{F}'$ satisfies cancellation. Consider $N\in \mathbb{Q}\mathcal{D}\cap \mathcal{R}(A)$. Since $\mathcal{R}(A)$ is a subspace of 
$\mathcal{N}(A)$, we have that $k_NN\in \mathcal{D}\cap \mathcal{R}(A)$. By cancellation of $\mathscr{F}$ we have then that
$
\mathscr{F}'(N)=\mathscr{F}(k_N N)=A,
$
and cancellation of $\mathscr{F}'$ is proved.

We now prove that $\mathscr{F}'$ satisfies neutrality. Let $N\in \mathbb{Q}\mathcal{D}$ and $\psi\in \mathrm{Sym}(A)$. We first show that $N^\psi\in \mathbb{Q}\mathcal{D}$. We know that there are $s\in \mathbb{N}$, $N_1,\ldots, N_s\in \mathcal{D}$ and $q_1,\ldots, q_s\in\mathbb{Q}$ such that $N=\sum_{i=1}^sq_iN_i$. Thus, we have that
\[
N^\psi=\left(\sum_{i=1}^sq_iN_i\right)^\psi=\sum_{i=1}^sq_iN_i^\psi.
\] 
Since $\mathscr{F}$ satisfies neutrality, we know that $\mathcal{D}$ is CVP.
We conclude then that, for every $i\in[s]$, $N_i^\psi\in\mathcal{D}$ and thus $N^\psi\in \mathbb{Q}\mathcal{D}$.
Let us prove now that $k_N=k_{N^\psi}$. 
We know that $k_N N\in\mathcal{D}$. Since $\mathcal{D}$ is CVP, $(k_N N)^\psi=k_N N^\psi\in\mathcal{D}$. Thus, $k_N\in K(N^\psi)$ and thus $k_{N^\psi}\le k_N$. Moreover, we also know that $k_{N^\psi} N^\psi\in\mathcal{D}$. Since $\mathcal{D}$ is CVP, $(k_{N^\psi} N^\psi)^{\psi^{-1}}=k_{N^\psi} (N^\psi)^{\psi^{-1}}=k_{N^\psi} N\in\mathcal{D}$. Thus, $k_{N^\psi}\in K(N)$ and thus $k_N\le k_{N^\psi}$.
We conclude then that $k_N=k_{N^\psi}$, as desired. Now, by neutrality of $\mathscr{F}$,  we get
\[
\mathscr{F}'(N^\psi)=\mathscr{F}(k_{N^\psi}N^\psi)=\mathscr{F}(k_{N}N^\psi)=\mathscr{F}( (k_{N} N)^\psi)=\psi(\mathscr{F}(k_NN))=\psi(\mathscr{F}'(N))
\]
and neutrality of $\mathscr{F}'$ is proved.

Let us prove that $\mathscr{F}'$ satisfies consistency. Let $N,M\in \mathbb{Q}\mathcal{D}$ and assume that $\mathscr{F}'(N)\cap \mathscr{F}'(M)\neq\varnothing$. It is clear that $N+M\in \mathbb{Q}\mathcal{D}$. We know that $k_NN, k_MM,k_{N+M}(N+M)\in\mathcal{D}$. Let us define $t=k_Nk_Mk_{N+M}$. 
Since $\mathscr{F}$ satisfies consistency, by Proposition \ref{scalar-vector-q} and by the definition of $\mathscr{F}'$, we have that 
\[
\mathscr{F}'(N)=\mathscr{F}(k_NN)=\mathscr{F}(tN),\quad \mathscr{F}'(M)=\mathscr{F}(k_MM)=\mathscr{F}(tM)
\]
and therefore $\mathscr{F}(tN)\cap \mathscr{F}(tM)\neq\varnothing$.
Thus, again by consistency of $\mathscr{F}$ and Proposition \ref{scalar-vector-q}, we get
\[
\mathscr{F}'(N+M)=\mathscr{F}(k_{N+M}(N+M))=\mathscr{F}(t(N+M))=\mathscr{F}(tN+tM)=
\]
\[
\mathscr{F}(tN)\cap \mathscr{F}(tM)=\mathscr{F}'(N)\cap \mathscr{F}'(M)
\]
and consistency of $\mathscr{F}'$ is proved.
\end{proof}

\begin{proposition}\label{ext3}
Let $\mathcal{D}$ be a subspace of $\mathcal{N}(A)$ and $\mathscr{F}$ a {\sc ns} on $\mathcal{D}$ satisfying neutrality, consistency and  cancellation.
Then there exists $\mathscr{F}':\mathcal{D}+\mathcal{R}(A)\to P_*(A)$ extension of $\mathscr{F}$ satisfying neutrality, consistency and cancellation.
\end{proposition}

\begin{proof}
Let $\mathscr{F}':\mathcal{D}+\mathcal{R}(A)\to P_*(A)$ be defined, for every $N'\in \mathcal{D}+\mathcal{R}(A)$, by $\mathscr{F}'(N')=\mathscr{F}(N),$
where $N'=N+R$ with $N\in\mathcal{D}$ and $R\in\mathcal{R}(A)$. We are going to prove that $\mathscr{F}'$ is a well defined extension of $\mathscr{F}$ and satisfies neutrality, consistency and cancellation.

First we  prove that $\mathscr{F}'$ is well defined. Let $N'\in \mathcal{D}+\mathcal{R}(A)$ and assume that
$N'=N_1+R_1=N_2+R_2$ where $N_1,N_2\in\mathcal{D}$ and $R_1,R_2\in\mathcal{R}(A)$. Observe then that $N_1-N_2=R_2-R_1$. Since $\mathcal{D}$ and $\mathcal{R}(A)$ are subspaces of $\mathcal{N}(A)$ we have that $N_1-N_2\in \mathcal{D}$ and $R_2-R_2\in \mathcal{R}(A)$ and thus
$N_1-N_2\in \mathcal{D}\cap\mathcal{R}(A)$. Since $\mathscr{F}$ satisfies consistency and cancellation and $N_1=N_2+(N_1-N_2)$, by Proposition \ref{sum-R}, we obtain that $\mathscr{F}(N_1)=\mathscr{F}(N_2)$, as desired.

In order to prove that $\mathscr{F}'$ extends $\mathscr{F}$ simply note that $N(0)\in \mathcal{D}\cap\mathcal{R}(A)$ because $\mathcal{D}$ is a subspace of $\mathcal{N}(A)$ and so, given $N\in \mathcal{D}$, we have 
$N=N+N(0)\in \mathcal{D}+\mathcal{R}(A)$ and $\mathscr{F}'(N)=\mathscr{F}(N)$.

Let us prove that $\mathscr{F}'$ satisfies neutrality. Let  $N'\in \mathcal{D}+\mathcal{R}(A)$ and $\psi\in\mathrm{Sym}(A)$. Then 
$N'=N+R$ where $N\in\mathcal{D}$ and $R\in\mathcal{R}(A)$. We have that $R^\psi\in \mathcal{R}(A)$ and, since $\mathscr{F}$ is neutral, we also have that $N^\psi\in\mathcal{D}$. Finally, again using the neutrality of $\mathscr{F}$, we have that
\[
\mathscr{F}'((N')^\psi)=\mathscr{F}'((N+R)^\psi)=\mathscr{F}'(N^\psi+R^\psi)=\mathscr{F}(N^\psi)=\psi(\mathscr{F}(N))=\psi(\mathscr{F}'(N')),
\]
that proves the neutrality of $\mathscr{F}'$.

Let us prove that $\mathscr{F}'$ satisfies consistency. Let  $N', M'\in \mathcal{D}+\mathcal{R}(A)$ be such that 
$\mathscr{F}'(N')\cap \mathscr{F}'(M')\neq\varnothing$. We have that $N'=N+S$ and $M'=M+T$ where $N,M\in\mathcal{D}$ and $S,T\in\mathcal{R}(A)$.
Since $\mathscr{F}'(N')=\mathscr{F}(N)$ and $\mathscr{F}'(M')=\mathscr{F}(M)$, we deduce that $\mathscr{F}(N)\cap \mathscr{F}(M)\neq\varnothing$.
Note also that $N'+M'=(N+M)+(S+T)$ where $N+M\in\mathcal{D}$ and $S+T\in\mathcal{R}(A)$.
Then, by the consistency of $\mathscr{F}$ and Proposition \ref{sum-R}, we obtain that
\[
\mathscr{F}'(N'+M')=\mathscr{F}((N+M)+(S+T))=\mathscr{F}(N+M)=\mathscr{F}(N)\cap \mathscr{F}(M)=\mathscr{F}'(N')\cap \mathscr{F}'(M'),
\]
that proves the consistency of $\mathscr{F}'$.

Let us prove that $\mathscr{F}'$ satisfies cancellation. Consider $N'\in (\mathcal{D}+\mathcal{R}(A))\cap \mathcal{R}(A)=\mathcal{R}(A)$. We have that $N'=N+R$ for some $N\in\mathcal{D}$ and $R\in\mathcal{R}(A)$. We deduce then that $N=N'-R\in \mathcal{D}\cap\mathcal{R}(A)$. Since $\mathscr{F}$ satisfies cancellation, we conclude then that $\mathscr{F}'(N')=\mathscr{F}(N)=A$. That proves that $\mathscr{F}'$ satisfies cancellation.
\end{proof}

\begin{proof}[Proof of Theorem \ref{ext-finale}]
By Proposition \ref{ext1}, there exists $\mathscr{F}_1:\mathbb{Z}\mathcal{D}\to P_*(A)$, extension of $\mathscr{F}$ and satisfying neutrality, consistency and cancellation.
Since $\mathbb{Z}(\mathbb{Z}\mathcal{D})=\mathbb{Z}\mathcal{D}$, we can apply Proposition \ref{ext2} to $\mathscr{F}_1$ obtaining the existence of $\mathscr{F}_2:\mathbb{Q}\mathcal{D}\to P_*(A)$, extension of $\mathscr{F}_1$ satisfying neutrality, consistency and cancellation. Thus, in particular, $\mathscr{F}_2$ is an extension of $\mathscr{F}$.
Since $\mathbb{Q}\mathcal{D}$ is a subspace of $\mathcal{N}(A)$, we can apply Proposition \ref{ext3} to $\mathscr{F}_2$ obtaining the existence of $\mathscr{F}':\mathbb{Q}\mathcal{D}+\mathcal{R}(A)\to P_*(A)$, extension of $\mathscr{F}_2$ satisfying neutrality, consistency and cancellation. Thus, $\mathscr{F}'$ is also an extension of $\mathscr{F}$ having the required properties. 
\end{proof}

\section{Proof of Theorem \ref{reg-scc}}

Theorem \ref{reg-scc} follows from Theorems \ref{tutte-new}, \ref{dicotomici-all} and \ref{truncated-all}, which are respectively proved in Sections \ref{1}, \ref{2} and \ref{3}. Some preliminary results are presented in Section \ref{0}.

\subsection{Some preliminary propositions} \label{0}

When $\mathbf{D}(A)$ is closed under permutation of alternatives, the set $\mathbf{D}(A)^*$ satisfies many useful properties, as described in the following propositions. In particular, Proposition \ref{beta} turns out to be a crucial short-cut for testing condition $(\beta)$ of Definition \ref{regular}.

\begin{proposition}\label{new-prop}
Let $\mathbf{D}(A)\subseteq \mathbf{R}(A)$. The following facts hold true.
\begin{itemize}
\item [$(i)$]$\mathbf{D}(A)^*$ is {\rm CA} and {\rm CWC} and $N(\mathbf{D}(A)^*)$ is {\rm CA}. 
\item[$(ii)$] If $\mathbf{D}(A)$ is closed under permutation of alternatives, then $\mathbf{D}(A)^*$ is {\rm CPA} and {\rm CRS} and $N(\mathbf{D}(A)^*)$ is {\rm CPV}.
\item[$(iii)$] If $\mathbf{D}(A)$ is closed under permutation of alternatives and there exists $R\in \mathbf{D}(A)$ such that $R\not\subseteq A^2_d$, then, for every $k\in\mathbb{N}$, there exists $t\in \mathbb{N}$, with $t\ge k$ such that $N(t)\in N(\mathbf{D}(A)^*)\cap\mathcal{C}(A)$.
\end{itemize}
\end{proposition}

\begin{proof}
It is immediately observed that $\mathbf{D}(A)^*$ is {\rm CA} and {\rm CWC}. Thus, by Proposition \ref{nuovo-lemma}$(ii)$, we get $(i)$. If $\mathbf{D}(A)$ is closed under permutation of alternatives, it is immediate to show that $\mathbf{D}(A)^*$ is {\rm CPA}. 
Thus, by Proposition \ref{nuovo-lemma}$(i)$-$(iii)$, we get $(ii)$. If $\mathbf{D}(A)$ is closed under permutation of alternatives and there exists $R\in \mathbf{D}(A)$ such that $R\not\subseteq A^2_d$,
we can consider $p\in \mathbf{D}(A)^*$ such that $\mathrm{Dom}(p)=\{1\}$ and $p(1)=R$. Then $N(p)\neq N(0)$ and so, applying Proposition \ref{nuovo-lemma}$(iv)$, we get $(iii)$. 
\end{proof}

\begin{proposition}\label{beta}
Let $\mathbf{D}(A)\subseteq \mathbf{R}(A)$ be nonempty. If $\mathbf{D}(A)$ is closed under permutation of alternatives, then $\mathbf{D}(A)^*$ satisfies condition $(\beta)$ of Definition {\rm \ref{regular}}.
\end{proposition}

\begin{proof} Note that, since  $\mathbf{D}(A)\neq \varnothing$, we have that $\mathbf{D}(A)^*\neq \varnothing$. We need to show that, for every $N\in \mathbb{Z}N(\mathbf{D}(A)^*)$,  
\begin{equation}\label{tesi1}
\mbox{there exist $p,p'\in \mathbf{D}(A)^*$ with $N(p')\in \mathcal{R}(A)$ such that $N(p)=N+N(p')$.}
\end{equation}

Assume first that $\mathbf{D}(A)\subseteq P(A^2_d)$. Then we have $N(\mathbf{D}(A)^*)=\{N(0)\}$  and hence also  $\mathbb{Z} N(\mathbf{D})=\{N(0)\}$. Let $N\in \mathbb{Z} N(\mathbf{D})$. Picking any $p,p'\in \mathbf{D}(A)^*$, we have that \eqref{tesi1} is trivially satisfied since $N=N(p)=N(p')=N(0)\in \mathcal{R}(A)$. 

Assume next that $\mathbf{D}(A)\not\subseteq P(A^2_d)$.
Then, by Proposition \ref{new-prop}, we have that $\mathbf{D}(A)^*$ is {\rm CA}, {\rm CWC}, {\rm CPA} and {\rm CRS},
$N(\mathbf{D}(A)^*)$ is {\rm CA} and there exists $\overline{p}\in \mathbf{D}(A)^*$ such that $N(\overline{p})\in \mathcal{C}(A)\subseteq \mathcal{R}(A)$. 
Let $N\in \mathbb{Z} N(\mathbf{D})$ and prove \eqref{tesi1}. 

If $N=N(0)$, we get \eqref{tesi1} choosing $p=p'=\overline{p}$.
If $N\neq N(0)$, there exist $k\in \mathbb{N}$, $p_1,\ldots,p_k\in \mathbf{D}(A)^*$ and $m_1,\ldots,m_k\in \mathbb{Z}$, 
with  $m_{i}\neq 0$ for some $i\in[k]$,  such that
\[
N=\sum_{i\in [k]}m_iN(p_i).
\]
Assume first that $m_i\geq 0$ for all $i\in [k]$. By the fact that $N(\mathbf{D}(A)^*)$ is CA, we get that $N=N(q)$ for some $q\in \mathbf{D}(A)^*$. Since $\mathbf{D}(A)^*$ is CWC, there exists $p' \in \mathbf{D}(A)^*$ clone of $\overline{p}$ disjoint from $q$. Moreover, let $p=q+p'$. Thus, $N(p')\in \mathcal{C}(A)$ and 
$N(p)=N(q+p')=N(q)+N(p')=N+N(p')$ gives \eqref{tesi1}.

Assume now that $m_i< 0$ for some $i\in [k]$.
Then 
\begin{equation}\label{formula}
N=\sum_{i\in [k], m_i>0} m_i N(p_i)- \sum_{i\in [k], m_i<0}(-m_i)N(p_i).
\end{equation}
Since $N(\mathbf{D}(A)^*)$ is CA, there exists $q\in \mathbf{D}(A)^*$ such that $\sum_{i\in [k], m_i<0}(-m_i)N(p_i)=N(q).$ There are two cases to discuss.

Suppose first that the first sum in \eqref{formula} is empty. Then we have $N=-N(q)$. Since $\mathbf{D}(A)^*$ is CRS there exists $q'\in \mathbf{D}(A)^*$ disjoint from $q$ such that $N(q)+N(q')=N(q+q')\in \mathcal{R}(A)$.
It follows that $N=N(q')-N(q+q')$ and thus $N+N(q+q')=N(q')$, which gives \eqref{tesi1} by choosing $p=q'$ and $p'=q+q'$.

Suppose next that the first sum in \eqref{formula} is not empty. Since $N(\mathbf{D}(A)^*)$ is CA and $\mathbf{D}(A)^*$ is CWC, 
there exists $s\in \mathbf{D}(A)^*$  with $s$ disjoint from $q$
such that $N=N(s)-N(q)$. Since $\mathbf{D}(A)^*$ is CRS and CWC, there exists $q'\in \mathbf{D}(A)^*$ disjoint from $q$ and $s$ such that
$N(q)+N(q')=N(q+q')\in \mathcal{R}(A)$. Thus, we get
$N=N(q')+N(s)-N(q+q')$ and so $N+N(q+q')=N(q')+N(s)=N(q'+s)$,
which gives \eqref{tesi1} by choosing $p=q'+s$ and $p'=q+q'$.
\end{proof}

\subsection{Linear orders}\label{1}

Let us consider now sets of preference profiles of the type $\mathbf{D}(A)^*$, where $\mathbf{L}(A)\subseteq \mathbf{D}(A)\subseteq \mathbf{R}(A)$.

\begin{proposition}\label{gamma-lin}
Let $\mathbf{D}(A)\subseteq \mathbf{R}(A)$ with $\mathbf{L}(A)\subseteq \mathbf{D}(A)$. Then
$\mathbb{Q} N(\mathbf{D}(A)^*)+\mathcal{R}(A)=\mathcal{N}(A)$.
\end{proposition}

\begin{proof} Since $\mathbb{Q} N(\mathbf{L}(A)^*)+\mathcal{R}(A)\subseteq \mathbb{Q} N(\mathbf{D}(A)^*)+\mathcal{R}(A)\subseteq \mathcal{N}(A)$, it is sufficient to show that $\mathbb{Q} N(\mathbf{L}(A)^*)+\mathcal{R}(A)\supseteq \mathcal{N}(A)$.

We claim first that, for every $x,y\in A$ with $x\neq y$, there exists $R_{xy}\in \mathcal{R}(A)$ such that 
\begin{equation}\label{young1}
2N_{xy}+R_{xy}\in N(\mathbf{L}(A)^*).
\end{equation}
Let $x,y\in A$ with $x\neq y$. Let $A=\{x_1,\dots, x_m\}$, where $x_1,\ldots, x_m$ are distinct, $x_1=x$ and $x_2=y.$ 
Let $p_{xy}\in \mathbf{L}(A)^*$ be such that $\mathrm{Dom}(p_{xy})=\{1,2\}$, $p_{xy}(1)=\{(x_i,x_j)\in A^2: i,j\in[m], i\le j\}$ and $p_{xy}(2)=(p_{xy}(1)^r)^{(x\,y)}$.\footnote{Using the standard informal way to represent linear orders, we have that if $m=2$, then $p_{xy}(1)=p_{xy}(2)=x_1x_2$;  if $m\geq 3$, then $p_{xy}(1)=x_1x_2 \cdots x_m$ and $p_{xy}(2)=x_mx_{m-1}\cdots x_3x_1x_2$.} 
Then, it is immediately checked that 
$N(p)=2N_{xy}+R_{xy}$, where 
\begin{equation}\label{Rxyrev}
R_{xy}\coloneq \sum_{i\in[m]}\sum_{j\in [m]\setminus\{1,2\} }(N_{x_ix_j}+N_{x_jx_i})\in\mathcal{R}(A).
\end{equation}
Note that $R_{xy}=N(0)$ when $m=2$. That proves the claim.

By \eqref{young1}, we have that, for every $x,y\in A$ with $x\neq y$, the networks $2N_{xy}+R_{xy}\in N(\mathbf{L}(A)^*)$ and, by \eqref{Rxyrev}, $R_{xy}\in \mathcal{R}(A).$
It follows that $N_{xy}\in \mathbb{Q}N(\mathbf{L}(A)^*)+\mathcal{R}(A)$. Since the networks $N_{xy}$ are a basis for $\mathcal{N}(A)$, we deduce that $\mathbb{Q}N(\mathbf{L}(A)^*)+\mathcal{R}(A)\supseteq\mathcal{N}(A)$. 
\end{proof}

\begin{theorem}\label{tutte-new} Let $ \mathbf{D}(A)\subseteq \mathbf{R}(A)$ be closed under permutation of alternatives and such that $\mathbf{L}(A)\subseteq \mathbf{D}(A)$. Then $\mathbf{D}(A)^*$ is regular. 
\end{theorem}

\begin{proof} 
We prove that $\mathbf{D}(A)^*$ satisfies the conditions $(\alpha)$, $(\beta)$, $(\gamma)$ and $(\varepsilon)$ of Definition \ref{regular}.

Since $\mathbf{D}(A)$ is closed under permutation of alternatives, by Proposition \ref{new-prop}, we have that  $\mathbf{D}(A)^*$ is CPA, CA, CWC and CRS. Thus, $(\alpha)$ in Definition \ref{regular} holds.
The fact that  $(\beta)$ in Definition \ref{regular} holds
 follows from Proposition \ref{beta}, because $\mathbf{D}(A)$ is closed under permutation of alternatives and $\mathbf{D}(A)\neq \varnothing.$
In order to prove that $\mathbf{D}(A)^*$ satisfies $(\gamma)$ in Definition \ref{regular}, observe first that by Proposition \ref{gamma-lin} we know that
\begin{equation*}\label{all}
\mathbb{Q}N(\mathbf{D}(A)^*)+\mathcal{R}(A)=\mathcal{N}(A).
\end{equation*}
As a consequence, we have
\[
\big(\mathbb{Q}N(\mathbf{D}(A)^*)+\mathcal{R}(A)\big)\cap \mathcal{PS}(A)=\mathcal{N}(A)\cap \mathcal{PS}(A)=\mathcal{PS}(A),
\]
and 
\[
\big(\mathbb{Q}N(\mathbf{D}(A)^*)+\mathcal{R}(A)\big)\cap \{N_x:x\in A\}=\{N_x:x\in A\}\neq\varnothing,
\]
 and thus $(\gamma)$ and $(\epsilon)$ hold true.
\end{proof}

\subsection{Dichotomous orders}\label{2}

\begin{definition}\label{caso-dicotomici}
Let $\varnothing\neq X\subseteq A$. Denote by $p_X$ the element of $\mathbf{Di}(A)^*$ such that $\mathrm{Dom}(p_X)=\{1\}$ and $\mathrm{Max}(p_X(1))=X$. 
\end{definition}
It is easily checked that 
\begin{equation}\label{singolo-dico}
N(p_X)=K_{A\setminus X}+\sum_{x\in X}N_x,
\end{equation}
where $K_{A\setminus X}$ is the complete network on $A\setminus X$ (see \eqref{Bcompleto}).
In particular, for every $x\in A$, since $K_{\{x\}}=N(0),$ we have
\begin{equation}\label{dico-speciale}
N(p_{\{x\}})=K_{A\setminus \{x\}}+N_x,\quad N(p_{A\setminus \{x\}})=\sum_{y\in A\setminus \{x\}}N_y.
\end{equation}

\begin{proposition}\label{dicotomici-CON}
Let $\mathbf{D}(A)\subseteq \mathbf{R}(A)$ with $\mathbf{Di}_t(A)\subseteq \mathbf{D}(A)$ for some $t\in[m-1]$. Then 
$\big(\mathbb{Q}N(\mathbf{D}(A)^*)+\mathcal{R}(A)\big)\cap \{N_x:x\in A\}\neq\varnothing$.
\end{proposition}

\begin{proof} Let $t\in [m-1]$ be such that $\mathbf{Di}_t(A)\subseteq \mathbf{D}(A)$. Since $\mathbf{Di}_t(A)^*\subseteq \mathbf{D}(A)^*$, it is sufficient to show that $\big(\mathbb{Q}N(\mathbf{Di}_t(A)^*)+\mathcal{R}(A)\big)\cap \{N_x:x\in A\}\neq\varnothing$.
 
If $m=2$, then $t=1$ and $\mathbf{Di}_t(A)^*=\mathbf{L}(A)^*$ and the desired property follows by Proposition \ref{gamma-lin}. 
Assume then that $m\geq 3$. Let $a\in A$ and let $r\coloneq \binom{m-1}{t-1}$ be the number of subsets of $A$ having size $t$ and containing $a$. Let $X_1,\dots, X_r$ be those subsets. Let  $p\in\mathbf{Di}_t(A)^*$ be such that $\mathrm{Dom}(p)=[r]$ and, for every $i\in [r]$, $\mathrm{Max}(p(i))\coloneq X_i$. If $t=1$, then we have $r=1$ and, by \eqref{singolo-dico}, we have $N(p)=N(p_{\{a\}})=K_{A\setminus \{a\}}+N_x$, as desired. 
Assume next $t\geq 2$.
By \eqref{singolo-dico}, we have 
\[
N(p)=\sum_{i\in[r]}N(p_{X_i})=\sum_{i\in[r]}\left(K_{A\setminus X_i}+\sum_{x\in X_i}N_x\right)=\sum_{i\in[r]}K_{A\setminus X_i}+ \sum_{i\in[r]}\sum_{x\in X_i}N_x
\]
\[
=\sum_{i\in[r]}K_{A\setminus X_i}+rN_a+\sum_{i\in[r]}\sum_{x\in X_i\setminus\{a\}}N_x=\sum_{i\in[r]}K_{A\setminus X_i}+rN_a+\binom{m-2}{t-2}\sum_{x\in A\setminus\{a\}}N_x
\]
\[
=\sum_{i\in[r]}K_{A\setminus X_i}+rN_a+\binom{m-2}{t-2}\left(-N_a + \sum_{x\in A}N_x \right)
\]
\[
=\left(r-\binom{m-2}{t-2}\right)N_a+\sum_{i\in[r]}K_{A\setminus X_i}+\binom{m-2}{t-2}N(1)
\]
Note that $R\coloneq \sum_{i\in[r]}K_{A\setminus X_i}+\binom{m-2}{t-2}N(1)\in \mathcal{R}(A)$ and that, since $m\geq 3$, we have
\[
k\coloneq \binom{m-1}{t-1}-\binom{m-2}{t-2}=\binom{m-2}{t-1}\in \mathbb{N}.
\]
Thus, we have proved that $N(p)=kN_a+R$. As a consequence, $N_a=\frac{1}{k}N(p)-\frac{1}{k}R\in \big(\mathbb{Q}N(\mathbf{Di}_t(A)^*)+\mathcal{R}(A)\big)\cap \{N_x:x\in A\}\neq\varnothing$. 
\end{proof}

\begin{proposition}\label{gamma-dic}
Let $\mathbf{D}(A)\subseteq \mathbf{Di}(A)$. Then $\left(\mathbb{Q} N(\mathbf{D}(A)^*)+\mathcal{R}(A)\right)\cap  \mathcal{PS}(A)=\mathcal{R}(A).$
\end{proposition}

\begin{proof} We first show that 
\begin{equation}\label{spazio-dico}
\mathbb{Q}N(\mathbf{Di}(A)^*)+\mathcal{R}(A)=\mathcal{O}(A)+\mathcal{R}(A).
\end{equation}
Recall that $\mathcal{O}(A)=\mathbb{Q}\{N_x:x\in A\}$. By \eqref{dico-speciale}, for every $x\in A$, we have $N_x+K_{A\setminus\{x\}}=N(p_{\{x\}})\in N(\mathbf{Di}(A)^*)$. Since $K_{A\setminus\{x\}}\in \mathcal{R}(A)$, we deduce that $N_x\in \mathbb{Q}N(\mathbf{Di}(A)^*)+\mathcal{R}(A)$.  As a consequence, we have that $\mathcal{O}(A)+\mathcal{R}(A)\subseteq \mathbb{Q}N(\mathbf{Di}(A)^*)+\mathcal{R}(A).$

Let $p\in \mathbf{Di}(A)^*$ and, for every $i\in\mathrm{Dom}(p)$, let $X_i=\mathrm{Max}(p(i))$. Using \eqref{singolo-dico}, we have that
\[
N(p)=\sum_{i\in\mathrm{Dom}(p)} N(p_{X_i})=\sum_{i\in\mathrm{Dom}(p)} \left(K_{A\setminus X_i}+\sum_{x\in X_i}N_x\right)=\sum_{i\in\mathrm{Dom}(p)} K_{A\setminus X_i}+\sum_{i\in\mathrm{Dom}(p)}\sum_{x\in X_i}N_x.
\]
It follows that $N(\mathbf{Di}(A)^*)\subseteq \mathcal{O}(A)+\mathcal{R}(A)$ and thus $\mathbb{Q}N(\mathbf{Di}(A)^*)+\mathcal{R}(A)\subseteq\mathcal{O}(A)+\mathcal{R}(A).$ That shows \eqref{spazio-dico}.

Let now $\mathbf{D}(A)\subseteq \mathbf{Di}(A)$. Then, by \eqref{spazio-dico}, we have
\[
\mathbb{Q}\mathcal{N}(\mathbf{D}(A)^*)+\mathcal{R}(A)\subseteq \mathbb{Q}\mathcal{N}(\mathbf{Di}(A)^*)+\mathcal{R}(A)=\mathcal{O}(A)+\mathcal{R}(A),
\]
and hence, by Proposition  \ref{generale-reti}$(iii.g)$, we get
\[
(\mathbb{Q}\mathcal{N}(\mathbf{D}(A)^*)+\mathcal{R}(A))\cap \mathcal{PS}(A)\subseteq (\mathcal{O}(A)+\mathcal{R}(A))\cap  \mathcal{PS}(A)=\mathcal{R}(A).
\]
Since the other inclusion trivially holds we deduce that $(\mathbb{Q}\mathcal{N}(\mathbf{D}(A)^*)+\mathcal{R}(A))\cap \mathcal{PS}(A)=\mathcal{R}(A).$
\end{proof}

\begin{theorem}\label{dicotomici-all} Let $X\subseteq [m]$ with $X\cap [m-1]\neq\varnothing$. Then
$\mathbf{Di}_X(A)^*$ is regular. 
\end{theorem}

\begin{proof} We prove that $\mathbf{Di}_X(A)^*$ satisfies the conditions $(\alpha)$, $(\beta)$, $(\gamma)$ and $(\epsilon)$ of Definition \ref{regular}.
Since $\mathbf{Di}_X(A)$ is closed under permutation of alternatives, by Proposition \ref{new-prop}, we have that 
$\mathbf{Di}_X(A)^*$ is CPA, CA and CRS. Thus, condition $(\alpha)$  is satisfied. 
By Proposition \ref{beta} we deduce that condition $(\beta)$ is satisfied. Since clearly $\mathbf{Di}_X(A)\subseteq \mathbf{Di}(A)$, by Proposition \ref{gamma-dic} we deduce that condition $(\gamma)$ is satisfied.
Since $X\cap [m-1]\neq\varnothing$, we have that $\mathbf{Di}_X(A)\supseteq \mathbf{Di}_t(A)$ for some $t\in [m-1]$. Thus, by Proposition \ref{dicotomici-CON} we deduce that condition $(\epsilon)$ is satisfied. 
\end{proof}

\subsection{Top-truncated orders}\label{3}

\begin{proposition}\label{s-truncated-CON}
Let $\mathbf{D}(A)\subseteq \mathbf{R}(A)$ with $\mathbf{T}_s(A)\subseteq \mathbf{D}(A)$ for some $s\in[m-1]$. Then $\big(\mathbb{Q}N(\mathbf{D}(A)^*)+\mathcal{R}(A)\big)\cap \{N_x:x\in A\}\neq\varnothing$.
\end{proposition}

\begin{proof}  
Let $x\in A$. Let us prove that there exist $p\in\mathbf{T}_s(A)^*$, $k\in \mathbb{N}$  and $R\in \mathcal{R}(A)$ such that $N(p)=kN_x+R$. That allows to complete the proof since $N_x=\frac{1}{k}N(p)-\frac{1}{k}R\in\big(\mathbb{Q}N(\mathbf{D}(A)^*)+\mathcal{R}(A)\big)\cap \{N_x:x\in A\}\neq\varnothing$.

Let $Q\in \mathbf{T}_s(A)$ be such that $\mathrm{Max}(Q)=\{x\}$.
Let $\psi_1,\dots, \psi_{(m-1)!}$ be the distinct  permutations in the set $\Omega=\{\psi \in \mathrm{Sym}(A): \psi(x)=x\}$.
Consider $p\in \mathbf{T}_s(A)^*$ such that $\mathrm{Dom}(p)=[(m-1)!]$ and, for every $i\in \mathrm{Dom}(p)$,
$p(i)=Q^{\psi_i}$. Note that, for every $i\in \mathrm{Dom}(p)$, $\mathrm{Max}(p(i))=\{x\}$.
It is immediately checked that $N(p)=(m-1)!N_x+R$,
where 
\[
R=\sum_{\psi \in \Omega} N\left(Q^{\psi}\cap (A\setminus\{x\})^2\right)=\sum_{\psi \in \Omega} N\left(Q\cap (A\setminus\{x\})^2\right)^{\psi}.
\]
We complete the proof showing that $R\in \mathcal{R}(A)$. Indeed, note that $c^R(x',y')=0$ if $x'=x$ or $y'=x$. Moreover, for every $\sigma\in\Omega$, $R^\sigma=R$. Then, using an argument analogous to the one used in the proof of Proposition \ref{costanti-sim}, 
we see that there exists $\alpha\in\mathbb{N}$ such that $R=\alpha K_{A\setminus \{x\}}\in \mathcal{R}(A)$. 
\end{proof}

\begin{proposition}\label{gamma-trunc}
Let $\mathbf{D}(A)\subseteq \mathbf{R}(A)$ with $\mathbf{D}(A)\supseteq \mathbf{T}_s(A)$, for some $s\in [m-1]$ with $s\geq 2$. Then
\[
\left(\mathbb{Q} N(\mathbf{D}(A)^*)+\mathcal{R}(A)\right)\cap  \mathcal{PS}(A)=\mathcal{PS}(A).
\]
\end{proposition}

\begin{proof} Let $s\in [m-1]$, with $s\geq 2$, be such that $\mathbf{D}(A)\supseteq \mathbf{T}_s(A)$ and let $W\coloneq \mathbb{Q} N(\mathbf{T}_s(A)^*)+\mathcal{R}(A)$.
In order to prove the desired equality, it is enough to show that $W=\mathcal{N}(A)$. Since $W$ is subspace of $\mathcal{N}(A)$, in order to get the equality
$W=\mathcal{N}(A)$ we show that, for every $(x,y)\in A^2_*$, $N_{xy}\in W$.

In order to simplify the notation, assume without loss of generality that $A=[m]$. By standard considerations about permutations of alternatives, it is enough to show that $N_{12}\in W$.

Consider $Q\in \mathbf{T}_s(A)$ be such that $\mathrm{Min}(Q)=A\setminus [s]$ and $1\succ_Q 2 \succ_Q\ldots\succ_Q s$.
Let $p\in \mathbf{T}_s(A)^*$ be such that $\mathrm{Dom}(p)=\{1\}$ and $p(1)=Q$. 
Then, we have 
\begin{equation}\label{ecco}
N(p)=\sum_{i\in[s]}N_i-\sum_{\substack{3\leq u\leq s\\ j\in [u-1]}}N_{uj}-N_{21}+K_{A\setminus [s]}\in W.
\end{equation}
By Proposition \ref{s-truncated-CON}, for every $x\in A$, there exist $p\in\mathbf{T}_s(A)^*$, $k\in \mathbb{N}$  and $R\in \mathcal{R}(A)$ such that $N(p)=kN_x+R.$ Thus, for every $x\in A$, we have that $N_x\in W$. As a consequence, we have that $\sum_{i\in[s]}N_i\in W$. Define now $M\coloneq \sum_{\substack{3\leq u\leq s\\ j\in [u-1]}}N_{u j}$. Then, by \eqref{ecco} and recalling that $K_{A\setminus [s]}\in \mathcal{R}(A)\subseteq W$, we deduce that $M+N_{21}\in W$.

Consider $(1\, 2)\in\mathrm{Sym}(A)$. 
Let $q\in \mathbf{T}_s(A)^*$ be such that $\mathrm{Dom}(q)=\{1\}$ and 
$q(1)=Q^{(1\,2)}$. 
Then we have 
\begin{equation*}
N(q)=\sum_{i\in[s]}N_i-\sum_{\substack{3\leq u\leq s\\ j\in [u-1]}}N_{uj}-N_{12}+K_{A\setminus [s]}\in W,
\end{equation*}
which implies $M+N_{12}\in W$. As a consequence, we have $M+N_{21}+M+N_{12}\in W$. Since $N_{21}+N_{12}\in \mathcal{R}(A)\subseteq W$, we get $M\in V$. Hence also $N_{12}\in W$, as desired.
\end{proof}

\begin{theorem}\label{truncated-all} 
Let $Y\subseteq [m-1]_0$ with $Y\cap [m-1]\neq\varnothing$. Then
$\mathbf{T}_Y(A)^*$ is regular. 
\end{theorem}

\begin{proof} If  $Y\cap [m-1]=\{1\}$, then we have $Y=\{1\}$ and hence $\mathbf{T}_Y(A)^*=\mathbf{T}_1(A)^*=\mathbf{Di}_1(A)^*$, or $Y=\{0,1\}$ and $\mathbf{T}_Y(A)^*=\mathbf{T}_{\{0,1\}}(A)^*=\mathbf{Di}_{\{1,m\}}(A)^*$. In both cases,
 by Theorem \ref{dicotomici-all}, we have that 
$\mathbf{T}_Y(A)^*$ is regular.
Assume now that $Y\cap [m-1]\neq \{1\}$. Thus, there exists $s\in Y\cap [m-1]$ with $s\ge 2$ and $\mathbf{T}_Y(A)^*\supseteq \mathbf{T}_s(A)^*$.
We prove that $\mathbf{T}_Y(A)^*$ satisfies the conditions $(\alpha)$, $(\beta)$ and $(\gamma)$ of Definition \ref{regular}.
Since $\mathbf{T}_Y(A)$ is closed under permutation of alternatives, by Proposition \ref{new-prop}, we have that 
$\mathbf{T}_Y(A)^*$ is CPA, CA and CRS. Thus, condition $(\alpha)$  is satisfied. 
By Proposition \ref{beta} we deduce that condition $(\beta)$ is satisfied.
By Proposition \ref{gamma-trunc} we deduce that condition $(\gamma)$ is satisfied.
By Proposition \ref{s-truncated-CON} we deduce that condition $(\epsilon)$ is satisfied. 
\end{proof}

\end{document}